\documentclass[
reprint,
pre,
amsmath,amssymb,
aps,
showpacs,
showkeys
]{revtex4-2}

\usepackage{graphicx}
\usepackage{dcolumn}
\usepackage{amsthm}
\usepackage{bm}
\usepackage{array}
\usepackage{booktabs}
\usepackage{multirow}

\newtheorem{theorem}{Theorem}
\newtheorem{remark}{Remark}

\begin{document}
	
	\preprint{APS/123-QED}
	
	\title{Exact solution and pair correlation functions for a generalized three-chain Ising tube with multispin interactions}
	
	\author{P.V. Khrapov}
	\email{pvkhrapov@gmail.com}
	\altaffiliation{ORCID: 0000-0002-6269-0727}
	\author{N.S. Volkov}
	\email{nikita.volkov92595@gmail.com}
	\altaffiliation{ORCID: 0009-0008-7095-4092}
	\affiliation{Independent Researchers, Moscow, Russia}
	
	\date{\today}
	
	\begin{abstract}
		We obtain an exact solution for a generalized three-chain Ising tube (TCGIT) of length $L$ with toroidal boundary conditions and the most general $C_3$-invariant Hamiltonian on an elementary prism, containing 20 independent coupling constants, including an external magnetic field. Using an $8\times 8$ transfer matrix, we derive the exact partition function of the finite system and obtain the free energy, internal energy, specific heat, magnetization, magnetic susceptibility, and entropy in the thermodynamic limit $L\to\infty$. In the general case, $\lambda_{\max}$ is determined by a quartic equation, whereas in the principal special case with even-spin interactions (PSC) the spectrum simplifies substantially: the characteristic polynomial factorizes, and $\lambda_{\max}$ is given by the root of a quadratic equation. For mirror-symmetric subfamilies, we derive explicit formulas for the pair correlation functions and express the magnetization in terms of the components of the eigenvector associated with $\lambda_{\max}$; in the even-spin case with $h=0$, the magnetization vanishes. Important special cases include the width-three planar model with nearest-neighbor, next-nearest-neighbor, and plaquette interactions, including the entropy limit $S(T\to0^+)=(\ln 2)/3$ for $k\ge 0$ and $S(T\to0^+)=0$ for $k<0$, as well as the width-three planar triangular model with distinct nearest-neighbor couplings, three-spin interactions involving neighboring triangles, and an external field.
	\end{abstract}
	
	\maketitle
	
	\section{\label{sec:Intro}Introduction}
	
	The Ising model remains a cornerstone of statistical physics and continues to be studied in various lattice geometries and interaction schemes. Chain-like spin systems provide a useful bridge between one-dimensional models and their higher-dimensional counterparts. Notably, quasi-one-dimensional geometries can already exhibit signatures that are characteristic of higher-dimensional lattice models. In particular, the Ising tube geometry provides a valuable testbed for investigating frustration and multispin interactions in a confined geometry \cite{Yurishchev1, Yurishchev2, Yurishchev3, Cirillo}.
	
	Exact solutions and correlation functions for two- and three-chain Ising tubes with pairwise interactions were obtained in \cite{Kalok, Yokota}, while in \cite{Khrapov} such results were derived for a two-chain tube with all possible multispin interactions. These studies revealed the rich and diverse behavior of the models and demonstrated the usefulness of analytical methods beyond the strictly one-dimensional case. Interest in multispin interactions has also been reinforced by their connections to lattice gauge theories and string models: the gonihedric Ising model, which incorporates additional next-nearest-neighbor and four-spin plaquette interactions, emerged from the random surface approach to string theory \cite{Cirillo, Savvidy1, Savvidy2, Savvidy3, Bathas}. Finally, the generalized three-chain Ising tube can be interpreted as a prismatic spin nanotube \cite{Farchakh1, Farchakh2, Sarli, Deviren1, Deviren2, Kaneyoshi}, with three spin chains forming the tube’s circumference. In this context, each “layer” of three spins corresponds to a ring around the nanotube axis, suggesting that the model may be relevant to the magnetic or electronic properties of nanotubes.
	
	Here we study a generalized three-chain Ising tube with toroidal boundary conditions, including all possible multispin interactions within each elementary prism of the three-chain lattice that are invariant under rotations by $2\pi/3$. Using the transfer-matrix formalism \cite{Yokota, Montroll, Kramers1, Kramers2}, we obtain the partition function for a finite tube of length $L$, along with analytical expressions for physical quantities in the thermodynamic limit ($L \to \infty$). In particular, explicit formulas are derived for the free energy, internal energy, specific heat, magnetization, magnetic susceptibility, and entropy. To determine the eigenvalues of the $8 \times 8$ transfer matrix, we exploit the symmetries of the elementary Hamiltonian. The resulting permutation matrices, which commute with the transfer matrix, allow us to determine the full spectrum and the corresponding eigenvectors.
	
	Knowledge of the complete transfer-matrix spectrum and the structure of its eigenvectors enables the study of the model’s correlation properties. Explicit expressions for the pair correlation functions are derived, and a formula for the magnetization in terms of the components of the normalized eigenvector corresponding to the largest eigenvalue is obtained. Notably, for a broad class of tubes with interactions involving only an even number of spins (two-spin, four-spin, and six-spin interactions) in each interaction cluster, the analysis simplifies considerably. In this case, the characteristic polynomial factorizes into polynomials of lower degree, and consequently four of the eight transfer-matrix eigenvalues (including the largest eigenvalue) are roots of two quadratic equations, while the remaining four are roots of linear equations. As a result, the thermodynamic-limit free energy admits a particularly simple closed form in terms of the root of a quadratic polynomial, and the magnetization in zero external field vanishes. Owing to the symmetry of the normalized eigenvector, the magnetization formula yields zero, in agreement with previous results \cite{Fedro}. The ground-state structure is described for tubes with even-spin interactions \cite{Kashapov}. Using a Hamiltonian with two-, four-, and six-spin interactions as an example, we illustrate possible spin configurations corresponding to each ground state. Furthermore, an exact expression for the inverse correlation length is provided, and its behavior near and at the boundaries between different ground-state regions in the low-temperature regime is analyzed.
	
	The results obtained in this work not only provide a fairly general solution for the three-chain tube but also encompass several important special cases of independent interest. For tubes involving only pairwise (two-spin) interactions, the formulas reduce to the well-known results for three-chain Ising tubes with nearest- and next-nearest-neighbor interactions \cite{Yokota}. By appropriately restricting the interaction parameters, we also obtain exact solutions for a planar Ising tube with nearest-neighbor, next-nearest-neighbor, and plaquette interactions. This reduction yields exact thermodynamic quantities in the parameter regime corresponding to the planar gonihedric model \cite{Savvidy2, Savvidy3, Bathas, Cirillo}. Furthermore, as a particular case of the general three-chain solution, we consider a planar triangular Ising model with all interactions within each elementary triangle and various three-spin interactions between neighboring triangles.
	
	The remainder of the paper is organized as follows. Section~\ref{sec:ModelDescription} introduces the model, defines the Hamiltonian, and presents the main results for the general TCGIT case, both in the thermodynamic limit and for a finite system of length $L$. Section~\ref{sec3:MainSpecCase} discusses the principal special case with only even-spin interactions (two-spin, four-spin, and six-spin interactions, with four-spin interactions not necessarily of the plaquette type). Section~\ref{sec4} derives pair correlation functions and analyzes the ground-state structure and inverse correlation length for a symmetric subset of the model. Section~\ref{sec:PlanarModel} treats the width-three planar model with nearest-neighbor, next-nearest-neighbor, and plaquette interactions, including the planar gonihedric case. Section~\ref{sec6} considers the width-three planar triangular model with different nearest-neighbor interactions, triple interactions between neighboring triangles, and coupling to an external field. Section~\ref{sec7} contains the proofs of Theorems~\ref{theorem1}--\ref{theorem4}.
		
	\section{\label{sec:ModelDescription}Model description and main results}
	
	\begin{figure}[h]
		\centering
		\includegraphics[width=1\linewidth]{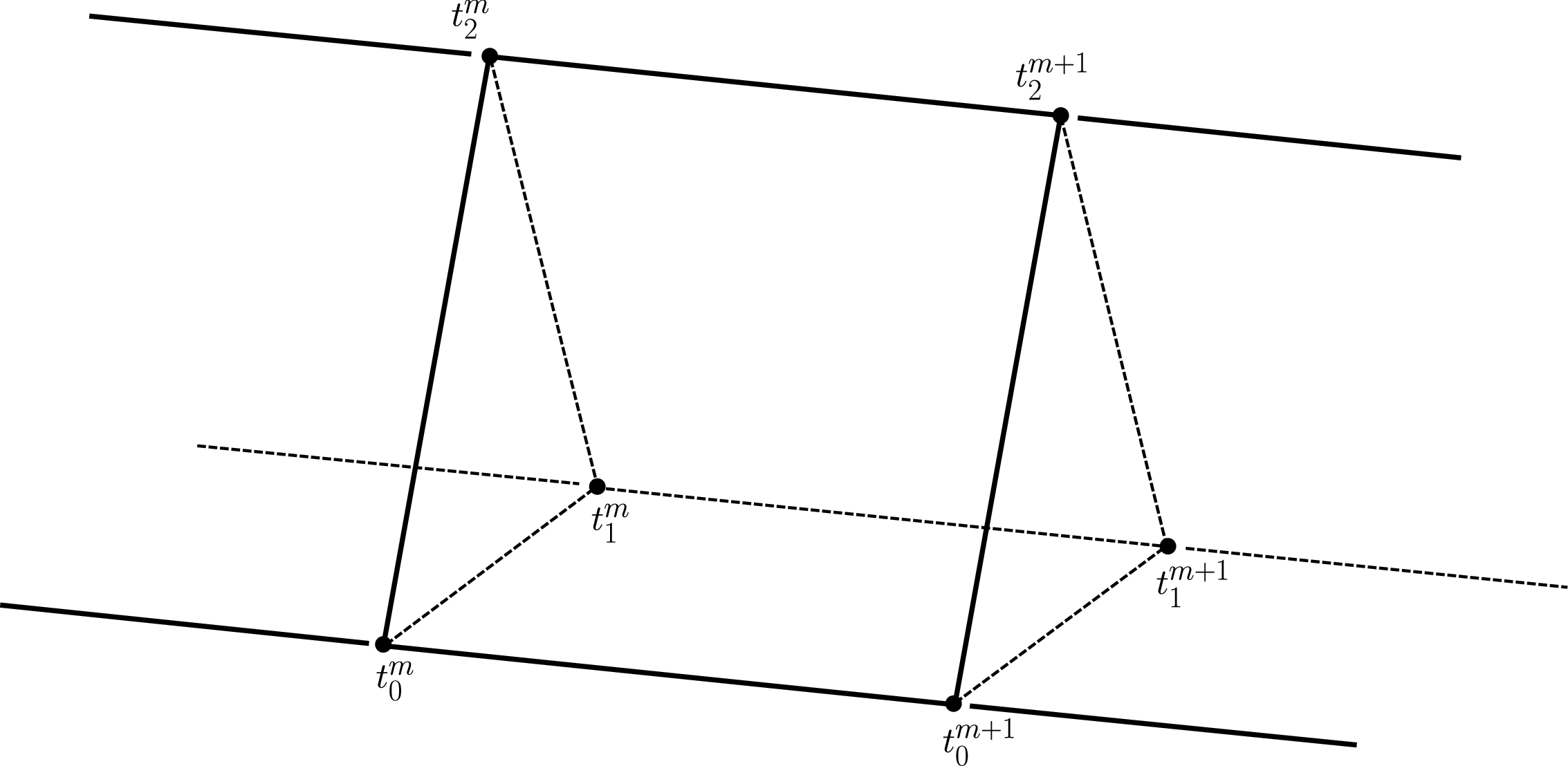}
		\caption{Three-chain lattice of the Ising tube.}
		\label{figure1} 
	\end{figure}
	
	We consider a quasi-one-dimensional three-chain Ising tube with $3 \times L$ sites (Fig.~\ref{figure1}), where $L$ denotes the number of layers in the model.
	\begin{eqnarray}
		\mathcal{L}_{L} = \bigl\{t_q^m, ~ m = 0, 1, \ldots, L-1, \nonumber \\
		t_q^L \equiv t_q^0,~ q = 0 , 1, 2\bigr\}.
		\label{1}
	\end{eqnarray}
	
	At each lattice site $t_q^m$ we place an Ising spin, with its state determined by the value of the spin $\sigma_q^m \equiv \sigma_{t_q^m}$, where $m = 0, 1, \ldots, L-1$ and $q = 0, 1, 2$.
	
	Let $\Omega^m_P = \{t_q^m, t_q^{m+1}, q = 0, 1, 2\}$ denote the elementary triangular prism (Fig.~\ref{figure1}).

	We introduce a generating Hamiltonian that includes one representative from each interaction cluster inequivalent under $C_3$ symmetry:
	\begin{equation}
		\hat{\mathcal{H}}^{(m)} = \hat{\mathcal{H}}^{(m)}_1 + \hat{\mathcal{H}}^{(m)}_2 + \hat{\mathcal{H}}^{(m)}_3 + \hat{\mathcal{H}}^{(m)}_4 + \hat{\mathcal{H}}^{(m)}_5 + \hat{\mathcal{H}}^{(m)}_6,
		\label{3}
	\end{equation}
	where
	\begin{equation}
		\label{H1}
		\hat{\mathcal{H}}^{(m)}_1 = -\frac{h}{2}\bigg(\sigma_0^m+\sigma_0^{m+1}\bigg)
	\end{equation}
	\begin{eqnarray}
		\label{H2}
		\hat{\mathcal{H}}^{(m)}_2 &=& -\bigg(\frac{1}{2} \bigg[J_{\{t_0^m, t_1^m\}}  \sigma_0^m \sigma_1^{m} \nonumber \\
		&+& J_{\{t_0^{m}, t_1^{m}\}} \sigma_0^{m+1} \sigma_1^{m+1}\bigg]
		 \nonumber \\
		&+&J_{\{t_0^m, t_0^{m+1}\}} \sigma_0^m \sigma_0^{m+1}  \nonumber \\
		&+&J_{\{t_0^m, t_1^{m+1}\}} \sigma_0^m \sigma_1^{m+1} \nonumber \\
		&+&J_{\{t_0^m, t_2^{m+1}\}} \sigma_0^m \sigma_2^{m+1}\bigg)
	\end{eqnarray}
	\begin{eqnarray}
		\label{H3}
		\hat{\mathcal{H}}^{(m)}_3 &=& \nonumber \\
		&-&\bigg(
		\frac{1}{6} \bigg[J_{\{t_0^m, t_1^m, t_2^m\}}  \sigma_0^m \sigma_1^{m} \sigma_2^{m} \nonumber \\
		&+&J_{\{t_0^m, t_1^m, t_2^m\}}  \sigma_0^{m+1} \sigma_1^{m+1} \sigma_2^{m+1}\bigg] \nonumber \\
		&+&J_{\{t_0^m, t_1^m, t_0^{m+1}\}}\sigma_0^m \sigma_1^m \sigma_0^{m+1}  \nonumber \\
		&+&J_{\{t_0^m, t_1^m, t_1^{m+1}\}}\sigma_0^m \sigma_1^m \sigma_1^{m+1}  \nonumber \\
		&+&J_{\{t_0^m, t_1^m, t_2^{m+1}\}}\sigma_0^m \sigma_1^m \sigma_2^{m+1}  \nonumber \\
		&+&J_{\{t_0^m, t_0^{m+1}, t_1^{m+1}\}}\sigma_0^m \sigma_0^{m+1} \sigma_1^{m+1}  \nonumber \\
		&+&J_{\{t_0^m, t_0^{m+1}, t_2^{m+1}\}}\sigma_0^m \sigma_0^{m+1} \sigma_2^{m+1}  \nonumber \\
		&+&J_{\{t_0^m, t_1^{m+1}, t_2^{m+1}\}}\sigma_0^m \sigma_1^{m+1} \sigma_2^{m+1}
		\bigg)
	\end{eqnarray}	
	\begin{eqnarray}
		\label{H4}
		\hat{\mathcal{H}}^{(m)}_4 &=& \nonumber \\
		 &-&\bigg( J_{\{t_0^m, t_1^m, t_2^m, t_0^{m+1}\}}  \sigma_0^m \sigma_1^{m} \sigma_2^{m}\sigma_0^{m+1} \nonumber \\
		 &+&J_{\{t_0^m, t_1^m, t_0^{m+1}, t_1^{m+1}\}}\sigma_0^m \sigma_1^m \sigma_0^{m+1} \sigma_1^{m+1} \nonumber \\
		 &+&J_{\{t_0^m, t_1^m, t_0^{m+1}, t_2^{m+1}\}}\sigma_0^m \sigma_1^m \sigma_0^{m+1} \sigma_2^{m+1} \nonumber \\
		 &+&J_{\{t_0^m, t_1^m, t_1^{m+1}, t_2^{m+1}\}}\sigma_0^m \sigma_1^m \sigma_1^{m+1} \sigma_2^{m+1} \nonumber \\
		 &+&J_{\{t_0^m, t_0^{m+1}, t_1^{m+1}, t_2^{m+1}\}}\sigma_0^m \sigma_0^{m+1} \sigma_1^{m+1} \sigma_2^{m+1}
		 \bigg)
	\end{eqnarray}	
	\begin{eqnarray}
		\label{H5}
		\hat{\mathcal{H}}^{(m)}_5 &=& \nonumber \\
		&-&\bigg( 
		J_{\{t_0^m, t_1^m, t_2^m, t_1^{m+1}, t_2^{m+1}\}}\sigma_0^m \sigma_1^m \sigma_2^m \sigma_1^{m+1} \sigma_2^{m+1} \nonumber \\
		&+&J_{\{t_1^m, t_2^m, t_0^{m+1}, t_1^{m+1}, t_2^{m+1}\}}\sigma_1^m \sigma_2^m \sigma_0^{m+1} \sigma_1^{m+1} \sigma_2^{m+1}
		\bigg) \nonumber \\
	\end{eqnarray}	
	\begin{eqnarray}
		\label{H6}
		\hat{\mathcal{H}}^{(m)}_6 &=& \nonumber \\
		&-&\frac{1}{3}J_{\{t_0^m, t_1^m, t_2^m, t_0^{m+1}, t_1^{m+1},t_2^{m+1}\}}  \sigma_0^m \sigma_1^{m} \sigma_2^{m}\sigma_0^{m+1} \sigma_1^{m+1} \sigma_2^{m+1} \nonumber \\
	\end{eqnarray}	
	
	The coefficient $J_{\{t_0^m\}} = J_{\{t_0^{m+1}\}}$ represents the coupling to an external magnetic field; hence, in what follows, we denote it by $h$.
	
	For a function $\varphi(\sigma_0^m, \sigma_1^m, \sigma_2^m, \sigma_0^{m+1}, \sigma_1^{m+1}, \sigma_2^{m+1})$ depending on $\sigma_t$ with $t \in \Omega^m_P$, we introduce the $2\pi/3$ rotation operator $R^m_{2\pi/3}$:
	\begin{eqnarray}
		&&R^m_{\frac{2\pi}{3}}\big(\varphi(\sigma_0^m, \sigma_1^m, \sigma_2^m, \sigma_0^{m+1}, \sigma_1^{m+1}, \sigma_2^{m+1})\big) \nonumber \\
		&&=\varphi(\sigma_1^m, \sigma_2^m, \sigma_0^m, \sigma_1^{m+1}, \sigma_2^{m+1}, \sigma_0^{m+1}), 
		\label{4}
	\end{eqnarray}
	then the Hamiltonian of the model has the form
	\begin{equation}
		\mathcal{H}_L(\sigma, h) = \sum_{m = 0}^{L-1}\mathcal{H}^{(m)},
		\label{5}
	\end{equation}
	where
	\begin{equation}
		\mathcal{H}^{(m)} = \hat{\mathcal{H}}^{(m)} + R^m_{\frac{2\pi}{3}}\hat{\mathcal{H}}^{(m)} + \big(R^m_{\frac{2\pi}{3}}\big)^2\hat{\mathcal{H}}^{(m)}.
		\label{6}
	\end{equation}
	\begin{remark}
		In this formulation (\ref{5}) of the Hamiltonian, the factor of $\frac{1}{3}$ must be included for the $J_{\{t_0^m,t_1^m,t_2^m\}}$ interaction and for $\hat{\mathcal{H}}^{(m)}_6$ to account for rotational symmetry.
		
		Certain interaction clusters appear twice in this representation: $\{t_0^m\}$ and $\{t_0^{m+1}\}$, $\{t_0^m, t_1^m\}$ and $\{t_0^{m+1}, t_1^{m+1}\}$, $\{t_0^m, t_1^m, t_2^m\}$ and $\{t_0^{m+1}, t_1^{m+1}, t_2^{m+1}\}$. The factor of $\frac{1}{2}$ is introduced before these interactions to avoid double counting contributions from adjacent layers.
		\label{remark0}
	\end{remark}
	\begin{remark}
		In this parametrization, there are 20 independent couplings (including the external field), after accounting for $C_3$ symmetry and the shared intra-layer terms.
		\label{param_count_remark}
	\end{remark}
		
	The partition function is given by
	\begin{equation}
		Z_{L} (h, T) = \sum_{\sigma} \exp\bigg\{-\frac{\mathcal{H}_L(\sigma, h)}{k_BT}\bigg\},
		\label{7}
	\end{equation}
	where $k_B$ denotes the Boltzmann constant (for simplicity, $k_B$ is set to 1), and the summation runs over all spin configurations.
	
	To compute the partition function, we introduce an $8 \times 8$ transfer matrix, with elements defined as follows:
	\begin{equation}
		\theta_{k,l} = \theta_{\big(\sigma_{0}^m, \sigma_{1}^m, \sigma_{2}^m), (\sigma_0^{m+1}, \sigma_1^{m+1}, \sigma_2^{m+1}\big)} = \exp \bigg\{- \frac{\mathcal{H}^{(m)}}{T}\bigg\},
		\label{8}
	\end{equation}
	\[
		k = 1 + \sum_{q=0}^{2}2^{q-1}\big(1-\sigma_q^m\big),
	\]
	\[
		l = 1 + \sum_{q=0}^{2}2^{q-1}\big(1-\sigma^{m+1}_q\big).
	\]
	The structure of the transfer matrix is summarized in Table~\ref{table:transfer-matrix}.
	
	\begin{table}[h]
		\centering
		\caption{Structure of the transfer matrix $\theta$.}
		\label{table:transfer-matrix}
		\begin{tabular}{|c|c|c|l|l|l|l|l|l|l|l|l|}
			\hline
			\multicolumn{3}{|l|}{\multirow{3}{*}{}}&\multicolumn{4}{c|}{$+$}&\multicolumn{4}{c|}{$-$}&\multicolumn{1}{c|}{$\sigma_{2}^{m+1}$}\\
			\cline{4-12}
			\multicolumn{3}{|l|}{}&\multicolumn{2}{c|}{$+$}&\multicolumn{2}{c|}{$-$}&\multicolumn{2}{c|}{$+$}&\multicolumn{2}{c|}{$-$}&\multicolumn{1}{c|}{$\sigma_{1}^{m+1}$}\\
			\cline{4-12}
			\multicolumn{3}{|l|}{}&\multicolumn{1}{c|}{$+$}&\multicolumn{1}{c|}{$-$}&\multicolumn{1}{c|}{$+$}&\multicolumn{1}{c|}{$-$}&\multicolumn{1}{c|}{$+$}&\multicolumn{1}{c|}{$-$}&\multicolumn{1}{c|}{$+$}&\multicolumn{1}{c|}{$-$}&\multicolumn{1}{c|}{$\sigma_{0}^{m+1}$}\\
			\hline
			\multirow{4}{*}{$+$}&\multirow{2}{*}{$+$}&$+$&&&&&&&&&\multirow{8}{*}{}\\
			\cline{3-11}
			&&$-$&&&&&&&&&\\
			\cline{2-11}
			&\multirow{2}{*}{$-$}&$+$&&&&&&&&&\\
			\cline{3-11}
			&&$-$&&&&&&&&&\\
			\cline{1-11}
			\multirow{4}{*}{$-$}&\multirow{2}{*}{$+$}&$+$&&&&&&&&&\\
			\cline{3-11}
			&&$-$&&&&&&&&&\\
			\cline{2-11}
			&\multirow{2}{*}{$-$}&$+$&&&&&&&&&\\
			\cline{3-11}
			&&$-$&&&&&&&&&\\
			\hline
			$\sigma_{2}^{m}$&$\sigma_{1}^{m}$&$\sigma_{0}^{m}$&\multicolumn{9}{l|}{}\\
			\hline
		\end{tabular}
	\end{table}
	
	The partition function can therefore be written as
	\begin{equation}
		Z_{L} (h, T) = \operatorname{Tr}\theta^L.
		\label{9}
	\end{equation}
	
	The free energy per spin is defined in the standard way \cite{Baxter}:
	\begin{equation}
		f(h, T) = - \frac{T}{N} \ln Z_{L}(h, T),
		\label{10}
	\end{equation}
	where $N = 3L$ denotes the total number of lattice sites in the system.
	
	The internal energy and specific heat per lattice site, as well as the magnetization, susceptibility, and entropy, are given by \cite{Baxter}, respectively, as follows:
	\begin{equation}
		u(h, T) = -T^2 \frac{\partial}{ \partial T}\biggl(\frac{f(h, T)}{T}\biggr),
		\label{11}
	\end{equation}
	\begin{eqnarray}
		&&c(h, T) = \frac{\partial}{\partial T}u(h,T) \nonumber \\
		&&= -2T\frac{\partial}{\partial T}\biggl(\frac{f(h, T)}{T}\biggr) 
		- T^2 \frac{\partial^2}{\partial T^2}\biggl(\frac{f(h, T)}{T}\biggr),
		\label{12}
	\end{eqnarray}
	\begin{equation}
		M(h,T)=-\frac{\partial}{\partial h} f(h,T),
		\label{13}
	\end{equation}
	\begin{equation}
		\chi(h,T) = \frac{\partial}{\partial h} M(h,T),
		\label{14}
	\end{equation}
	\begin{equation}
		S(h,T)=-\frac{\partial}{\partial T} f(h,T).
		\label{15}
	\end{equation}
	
	For later use, we introduce the reduced matrices:
	\begin{equation}
		\tau^1 = 
		\begin{pmatrix}
			\theta_{11}& \theta_{12} + \theta_{13} + \theta_{15}& \theta_{14} + \theta_{16} + \theta_{17}& \theta_{18}\\
			\theta_{21}& \theta_{22} + \theta_{23} + \theta_{25}& \theta_{24} + \theta_{26} + \theta_{27}& \theta_{28}\\
			\theta_{41}& \theta_{42} + \theta_{43} + \theta_{45}& \theta_{44} + \theta_{46} + \theta_{47}& \theta_{48}\\
			\theta_{81}& \theta_{82} + \theta_{83} + \theta_{85}& \theta_{84} + \theta_{86} + \theta_{87}& \theta_{88}
		\end{pmatrix},
		\label{16}
	\end{equation}
	\begin{eqnarray}
		&\tau^2& \nonumber \\
		&=& 
		\begin{pmatrix}
			e^{\frac{2 \pi}{3} i }  \theta_{52} + e^{-\frac{2 \pi}{3} i }  \theta_{53} + \theta_{55}& e^{-\frac{2 \pi}{3} i }  \theta_{54} + e^{\frac{2 \pi}{3} i }  \theta_{56} + \theta_{57} \\
			e^{\frac{2 \pi}{3} i }  \theta_{72} + e^{-\frac{2 \pi}{3} i }  \theta_{73} + \theta_{75}& e^{-\frac{2 \pi}{3} i }  \theta_{74} + e^{\frac{2 \pi}{3} i }  \theta_{76} + \theta_{77}
		\end{pmatrix}, \nonumber \\
		\label{17}
	\end{eqnarray}
	\begin{eqnarray}
		&\tau^3& \nonumber \\
		&=& 
		\begin{pmatrix}
			e^{-\frac{2 \pi}{3} i } \theta_{52} + e^{\frac{2 \pi}{3} i }  \theta_{53} + \theta_{55}& e^{\frac{2 \pi}{3} i }  \theta_{54} + e^{-\frac{2 \pi}{3} i }  \theta_{56} + \theta_{57} \\
			e^{-\frac{2 \pi}{3} i }  \theta_{72} + e^{\frac{2 \pi}{3} i } \theta_{73} + \theta_{75}& e^{\frac{2 \pi}{3} i }  \theta_{74} + e^{-\frac{2 \pi}{3} i }  \theta_{76} + \theta_{77}
		\end{pmatrix}. \nonumber \\
		\label{18}
	\end{eqnarray}
	
	We factorize the characteristic polynomial of the matrix $\tau^1$ in Eq.~(\ref{16}) according to Ferrari’s factorization of the quartic into two quadratic polynomials:
	\begin{equation}
		\lambda^4+a_3\lambda^3+a_2\lambda^2+a_1\lambda+a_0 =(\lambda^2+b_0\lambda+c_0)(\lambda^2+b_1\lambda+c_1).
		\label{20}
	\end{equation}
	
	Details of this factorization and the selection of the maximal root $\lambda_{\max}$ of Eq.~(\ref{20}) are given in Appendix A.
	
	\begin{theorem}
		\textbf{General case.} For the model with Hamiltonian (\ref{3})--(\ref{6}), the free energy, internal energy, specific heat, magnetization, magnetic susceptibility, and entropy per lattice site in the thermodynamic limit are given by
		\begin{equation}
			f(h, T) = -T\frac{\ln\lambda_{\max}}{3},
			\label{21}
		\end{equation}
		\begin{equation}
			u(h, T) =  T^2 \frac{\partial}{\partial T} \frac{\ln\lambda_{\max}}{3},
			\label{22}
		\end{equation}
		\begin{equation}
			c(h, T) = 2T \frac{\partial}{\partial T} \frac{\ln \lambda_{\max}}{3}
			+T^2 \frac{\partial^2}{\partial T^2}\frac{\ln\lambda_{\max}}{3},
			\label{23}
		\end{equation}
		\begin{equation}
			M(h,T) = \frac{T}{3} \frac{1}{\lambda_{\max}} \frac{\partial \lambda_{\max}}{\partial h},
			\label{24}
		\end{equation}
		\begin{equation}
			\chi(h,T) = \frac{T}{3}\biggl[\frac{1}{\lambda_{\max}}\frac{\partial^2 \lambda_{\max}}{\partial h^2}  -\biggl(\frac{1}{\lambda_{\max}} \frac{\partial \lambda_{\max}}{\partial h}\biggr)^2 \biggr],
			\label{25}
		\end{equation}
		\begin{equation} 
			S(h,T) = \frac{1}{3} \biggl(\ln\lambda_{\max} + \frac{T}{\lambda_{\max}} \frac{\partial \lambda_{\max}}{\partial T}\biggr),
			\label{26}
		\end{equation}
		where the largest eigenvalue $\lambda_{\max} = \lambda_{\max}(h, T)$ of the transfer matrix $\theta$ is the unique positive Perron root (spectral radius) of the matrix $\tau^1$. Depending on the parameter region, $\lambda_{\max}$ is determined by the following expression (selection rule for the Perron eigenvalue):
		\begin{equation}
			\lambda_{\max} = \max \biggl\{  \bigg|\frac{-b_0 + \sqrt{b_0^2 - 4c_0}}{2}\bigg|, ~ \bigg|\frac{-b_1 + \sqrt{b_1^2 - 4c_1}}{2}\bigg|  \biggr\},
			\label{27}
		\end{equation}
		where one of the expressions inside the absolute value may become complex.
		\label{theorem1}
	\end{theorem}
	
	\begin{theorem}
		For a finite system of length $L$ with periodic boundary conditions along the tube axis, the partition function can be written as
		\begin{equation}
			Z_{L} = \sum_{k = 1}^{8}\lambda_k^L,
			\label{31}
		\end{equation}
		where $\lambda_k$,~$k=1,\ldots,8$ are the roots of the characteristic polynomial of the transfer matrix $\theta$. Specifically:
		\begin{equation}
			\lambda_{1,2} = \frac{-b_0 \pm \sqrt{b_0^2 - 4c_0}}{2},
			\label{32}
		\end{equation}
		\begin{equation}
			\lambda_{3,4} = \frac{-b_1 \pm \sqrt{b_1^2 - 4c_1}}{2}
			\label{33}
		\end{equation}
		are the roots of the quartic characteristic polynomial (\ref{20}) of the matrix in Eq.~(\ref{16}), obtained using Ferrari’s method (Appendix~\ref{appendixA}), and the eigenvalues:
		\begin{equation}
			\lambda_{5,6} = \frac{-b_2 \pm \sqrt{b_2^2 - 4c_2}}{2},
			\label{34}
		\end{equation}
		\begin{equation}
			\lambda_{7,8} = \frac{-b_3 \pm \sqrt{b_3^2 - 4c_3}}{2}
			\label{35}
		\end{equation}
		are the roots of the characteristic quadratics of the matrices $\tau^2$~(\ref{17}) and $\tau^3$~(\ref{18}), respectively, where
		\begin{equation}
			\begin{gathered}
				b_r = -(\tau_{11}^r + \tau_{22}^r),\\
				c_r = \tau_{11}^r \tau_{22}^r - \tau_{12}^r \tau_{21}^r, \\
				r = 2,~3.
			\end{gathered}
			\label{36}
		\end{equation}
		\label{theorem2}
	\end{theorem}
	
	\begin{remark}
		If, in addition, the model satisfies the mirror-symmetry condition
		\begin{eqnarray}
			&&\mathcal{H}^{(m)} (\sigma_0^m, \sigma_1^m, \sigma_2^m, \sigma_0^{m+1}, \sigma_1^{m+1}, \sigma_2^{m+1}) \nonumber \\
			&&= \mathcal{H}^{(m)} (\sigma_0^{m+1}, \sigma_1^{m+1}, \sigma_2^{m+1}, \sigma_0^m, \sigma_1^m, \sigma_2^m),
			\label{mirror_hamiltonian}
		\end{eqnarray}
		then the eigenvalues in Eqs.~(\ref{34}) and~(\ref{35}) each have algebraic multiplicity two, because the characteristic polynomials of the matrices~(\ref{17}) and~(\ref{18}) coincide.
		\label{remark2}
	\end{remark}
	
	Eq.~(\ref{mirror_hamiltonian}) implies that the Hamiltonian is symmetric with respect to a plane intersecting the torus at the midpoints of the edges perpendicular to its plane.

	\section{\label{sec3:MainSpecCase}The principal special case (PSC): a tube with multispin interactions involving an even number of spins}
	
	\subsection{Model description and results}
	
	In this case, the components $\hat{\mathcal{H}}^{(m)}_1, \hat{\mathcal{H}}^{(m)}_3, \hat{\mathcal{H}}^{(m)}_5$ of the elementary generating Hamiltonian (\ref{3}) are zero, and the transfer matrix $\theta$ is centrosymmetric.
	
	For Theorems \ref{theorem3} and \ref{theorem4}, we introduce the matrices
	\begin{equation}
		\tau^4 = 
		\begin{pmatrix}
			\theta_{11} + \theta_{18}& \theta_{12} + \theta_{13}+\theta_{14} + \theta_{15}+\theta_{16} + \theta_{17}\\
			\theta_{21} + \theta_{28}& \theta_{22} + \theta_{23}+\theta_{24} + \theta_{25}+\theta_{26} + \theta_{27}
		\end{pmatrix},
		\label{37}
	\end{equation}
	\begin{equation}
		\tau^5 = 
		\begin{pmatrix}
			\theta_{11} - \theta_{18}& \theta_{12} + \theta_{13} - \theta_{14} + \theta_{15} - \theta_{16} - \theta_{17}\\
			\theta_{21} - \theta_{28}& \theta_{22} + \theta_{23} - \theta_{24} + \theta_{25} - \theta_{26} - \theta_{27}
		\end{pmatrix}.
		\label{38}
	\end{equation}
	We also define
	\[
		\begin{gathered}
			b_r = -(\tau_{11}^r + \tau_{22}^r),\\
			c_r = \tau_{11}^r\tau_{22}^r - \tau_{12}^r \tau_{21}^r, \\
			r = 4,~5.
		\end{gathered}
	\]
	
	\begin{theorem}
		In the thermodynamic limit, the free energy, internal energy, specific heat, and entropy per lattice site are given by Eqs.~(\ref{21})--(\ref{26}), where $\lambda_{\max}$ denotes the largest eigenvalue of the transfer matrix $\theta$, which has the form
		\begin{eqnarray}
			\lambda_{\max} = \frac{-b_4 + \sqrt{b_4^2-4c_4}}{2}.
			\label{39}
		\end{eqnarray}
		At zero external field, the magnetization vanishes by symmetry; the magnetic susceptibility is then obtained in the standard way by introducing a field term into the Hamiltonian.
		\label{theorem3}
	\end{theorem}
	
	\begin{theorem}
			The partition function for a finite tube of length $L$ can be written as in Eq.~(\ref{31}), where $\lambda_k,~k = 1,\ldots,4$ are the roots of the characteristic polynomials of the matrices $\tau^4$~(\ref{37}) and $\tau^5$~(\ref{38}):
		\begin{eqnarray} 
			\lambda_{1,2} = \frac{-b_4 \pm \sqrt{b_4^2 - 4c_4}}{2},
			\label{40}
		\end{eqnarray}
		\begin{eqnarray}
			\lambda_{3,4} = \frac{-b_5 \pm \sqrt{b_5^2 - 4c_5}}{2},
			\label{41}
		\end{eqnarray}
		and the other eigenvalues have the form:
		\begin{equation}
			\lambda_5 = \big(\theta_{22} + \theta_{27}\big) + e^{\frac{2\pi}{3}i}\big(\theta_{23} + \theta_{26}\big) + e^{-\frac{2\pi}{3}i}\big(\theta_{24} + \theta_{25}\big),
			\label{42}
		\end{equation}
		\begin{equation}
			\lambda_6 = \big(\theta_{22} - \theta_{27}\big) + e^{\frac{2\pi}{3}i}\big(\theta_{23} - \theta_{26}\big) - e^{-\frac{2\pi}{3}i}\big(\theta_{24} - \theta_{25}\big),
			\label{43}
		\end{equation}
		\begin{equation}
			\lambda_7 = \big(\theta_{22} + \theta_{27}\big) + e^{-\frac{2\pi}{3}i}\big(\theta_{23} + \theta_{26}\big) + e^{\frac{2\pi}{3}i}\big(\theta_{24} + \theta_{25}\big),
			\label{44}
		\end{equation}
		\begin{equation}
			\lambda_8 = \big(\theta_{22} - \theta_{27}\big) + e^{-\frac{2\pi}{3}i}\big(\theta_{23} - \theta_{26}\big) - 	e^{\frac{2\pi}{3}i}\big(\theta_{24} - \theta_{25}\big),
			\label{45}
		\end{equation}
		where $\theta_{2l}$ are the elements of the transfer matrix $\theta$~(\ref{8}), $l = 2,\ldots,7$.
		\label{theorem4}
	\end{theorem}
	
	\begin{remark}
		Under mirror-symmetry~(\ref{mirror_hamiltonian}), the transfer matrix $\theta$ is symmetric, and its spectrum~(\ref{40})--(\ref{45}) is real. In this case, the equalities $\lambda_5 = \lambda_7$ and $\lambda_6 = \lambda_8$ hold.
		\label{remark3}
	\end{remark}
	
	\section{\label{sec4} Pair correlations}
	In this section, we consider subclasses of the TCGIT and PSC models whose Hamiltonians satisfy the symmetry condition (\ref{mirror_hamiltonian}), which ensures that the transfer matrix is symmetric. We denote these models by TCGITSH and PSCSH, respectively. We focus on these symmetric subfamilies, where the transfer matrix is real symmetric and hence admits an orthonormal eigenbasis.
	
	\subsection{Pair correlations for the TCGITSH model}
	
	The transfer matrix in this case has the following form:
	\begin{equation} \label{cor.1}
		\theta = 
		\begin{pmatrix}
			a&b&b&c&b&c&c&d\\
			b&e&f&g&f&w&o&j\\
			b&f&e&w&f&o&g&j\\
			c&g&w&k&o&l&l&m\\
			b&f&f&o&e&g&w&j\\
			c&w&o&l&g&k&l&m\\
			c&o&g&l&w&l&k&m\\
			d&j&j&m&j&m&m&n
		\end{pmatrix},
	\end{equation}
	the unitary matrix $P$ that diagonalizes the transfer matrix is given by:
	\begin{equation} \label{cor.2}
		P = 
		\begin{pmatrix}
			\vec{v}_1^1&\vec{v}_1^2&\vec{v}_1^3&\vec{v}_1^4&\vec{v}_2^1&\vec{v}_2^2&\vec{v}_3^1&\vec{v}_3^2
		\end{pmatrix},
	\end{equation}
	where $\vec{v}_1^s$, $\vec{v}_2^t$, and $\vec{v}_3^u$ denote the orthonormal eigenvectors corresponding to the eigenvalues $\lambda_1, \ldots, \lambda_8$ of the transfer matrix (\ref{cor.1}), which are obtained from the characteristic equations of matrices (\ref{16}), (\ref{17}), and (\ref{18}), respectively, and can be expressed as follows:
	\begin{equation} \label{cor.3}
		\vec{v}_1^s = 
		\begin{pmatrix}
			a_1^s & a_2^s & a_2^s & a_3^s & a_2^s & a_3^s & a_3^s & a_4^s
		\end{pmatrix}^{\mathrm{T}},
	\end{equation}
	\begin{equation} \label{cor.4}
		\vec{v}_2^t = n_2^t
		\begin{pmatrix}
			0 & \omega y_2^t & \overline{\omega} y_2^t & \overline{\omega} & y_2^t & \omega & 1 & 0
		\end{pmatrix}^{\mathrm{T}},
	\end{equation}
	\begin{equation} \label{cor.5}
		\vec{v}_3^u = n_3^u
		\begin{pmatrix}
			0 & \overline{\omega} y_3^u & \omega y_3^u & \omega & y_3^u & \overline{\omega} & 1 & 0
		\end{pmatrix}^{\mathrm{T}},
	\end{equation}
	where
	$s = 1, \ldots, 4,$ \\
	$t, u= 1,~2$
	\[
		y_2^t = \frac{w + g\omega + o\overline{\omega}}{f - e + \lambda_{t + 4}},
	\] 
	\[
		y_3^u = \frac{w + g\overline{\omega} + o\omega}{f - e + \lambda_{u + 6}},
	\]
	\[
		\omega = e^{\frac{2\pi}{3}i},
	\]
	and $n_2^t, ~n_3^u$ are the normalizing coefficients, respectively.
	
	The eigenvalues $\lambda_1,\ldots,\lambda_4$ of the transfer matrix~(\ref{cor.1}) coincide with the eigenvalues of the matrix $\tau^1$ owing to the symmetry of $\theta$ and the structure of its eigenvectors corresponding to $\lambda_1,\ldots,\lambda_4$ (see the proof of Theorem~\ref{theorem1}):
	\[
		\begin{pmatrix}
			x_1 & x_2 & x_2 & x_3 & x_2 & x_3 & x_3 & x_4
		\end{pmatrix}^{\mathrm{T}}.
	\]
	
	In this case, the matrix $\tau^1$ in Eq.~(\ref{16}) takes the form:
	\[
		\tau^1 = 
		\begin{pmatrix}
			a&3b&3c&d\\
			b&e+2f&g+w+o&j\\
			c&g+w+o&k+2l&m\\
			d&3j&3m&n
		\end{pmatrix}.
	\]
	
	Let $\lambda$ denote an eigenvalue of the matrix. To determine the corresponding eigenvector, we choose one component as a reference value
	\[
		\vec{x} = 
		\begin{pmatrix}
			x_1 & x_2 & x_3 & x_4
		\end{pmatrix}^{\mathrm{T}},
	\]
	for example $x_4 = 1$ (or choose another component if $x_4 = 0$). This reduces the problem to solving an inhomogeneous system of linear equations with three unknowns, which is solved using Cramer’s rule \cite{Howard}. The resulting solution can then be expressed as:
	\[
		x_s = \frac{\Delta_s}{\Delta}, \quad s = 1,~2,~3,
	\]
	and then the normalized eigenvector of the transfer matrix:
	\[
		\vec{v}_1 = n_1
		\begin{pmatrix}
			\frac{\Delta_1}{\Delta}&  \frac{\Delta_2}{\Delta}& \frac{\Delta_2}{\Delta}&  \frac{\Delta_3}{\Delta}&  \frac{\Delta_2}{\Delta}&  \frac{\Delta_3}{\Delta}&  \frac{\Delta_3}{\Delta}& 1
		\end{pmatrix}^{\mathrm{T}},
	\]
	\[
		n_1 = \frac{1}{\sqrt{x_1^2+3x_2^2+3x_3^2+1}},
	\]
	which can be rewritten in the form (\ref{cor.3}).
	
	Similarly, given that the eigenvectors of the transfer matrix corresponding to the eigenvalues $\lambda_5, \lambda_6$ and $\lambda_7, \lambda_8$ can be written, respectively, as:
	\[
		\begin{gathered}
			\vec{v}_2 = 
			\begin{pmatrix}
				0 & \omega x_1' & \overline{\omega} x_1' & \overline{\omega} x_2' & x_1' & \omega x_2' & x_2' & 0
			\end{pmatrix}^{\mathrm{T}}, \\
			\vec{v}_3 = 
			\begin{pmatrix}
				0 & \overline{\omega} x_1'' & \omega  x_1'' & \omega x_2'' & x_1'' & \overline{\omega} x_2'' & x_2'' & 0
			\end{pmatrix}^{\mathrm{T}}.
		\end{gathered}
	\]
	This leads to Eqs.~(\ref{cor.4}) and (\ref{cor.5}).
	
	We must also consider the possibility that the eigenvalues $\lambda_1,\ldots,\lambda_4$ associated with the vectors in Eq.~(\ref{cor.3}) are degenerate. By the Perron--Frobenius theorem \cite{Horn}, $\lambda_1 = \lambda_{\max}$ is always a simple root, since all entries of the matrix $\tau^1$ are positive.
	
	In the case where the eigenvalue has multiplicity 3, i.e., $\lambda_2 = \lambda_3 = \lambda_4 = \lambda$, the corresponding eigenspace admits the following orthonormal basis, with $n_1^2$, $n_1^3$, and $n_1^4$ denoting normalization coefficients:
	\begin{equation} \label{cor.6}
		\vec{v}_1^2 = n_1^2
		\begin{pmatrix}
			0 & d & d & 0 & d & 0 & 0 & -3b
		\end{pmatrix}^{\mathrm{T}},
	\end{equation}
	\begin{equation} \label{cor.7}
		\vec{v}_1^3 = n_1^3
		\begin{pmatrix}
			0 & 3bc & 3bc & x_3 & 3bc & x_3 & x_3 & 3cd
		\end{pmatrix}^{\mathrm{T}},
	\end{equation}
	where $x_3 = -(3b^2 + d^2)$,
	\begin{equation} \label{cor.8}
		\vec{v}_1^4 = n_1^4
		\begin{pmatrix}
			x_1 & x_2 & x_2 & x_3 & x_2 & x_3 &	x_3 & x_4
		\end{pmatrix}^{\mathrm{T}},
	\end{equation}
	where $x_1 = d_2 + 3b^2 + 3c^2$, $x_2 = b(\lambda - a)$, $x_3 = c(\lambda - a)$, $x_4 = d(\lambda - a)$.
	
	In the case where an eigenvalue has multiplicity 2, i.e., $\lambda_2 = \lambda_3 = \lambda$ (the cases $\lambda_2 = \lambda_4 = \lambda$ and $\lambda_3 = \lambda_4 = \lambda$ are treated similarly and are not considered here), a convenient spanning set is:
	\[
		\vec{f}_2 = 
		\begin{pmatrix}
			x_1 & x_2 & x_2 & 0 & x_2 & 0 & 0 & x_4
		\end{pmatrix}^{\mathrm{T}},
	\]
	where $x_1 = d(e + 2f - \lambda) - 3bj$, $x_2 = (a - \lambda)j - db$, $x_4 = 3b_2 - (a - \lambda)(e + 2f - \lambda)$,
	
	\[
		\vec{f}_3 = 
		\begin{pmatrix}
			x_1 & x_2 & x_2 & x_3 & x_2 & x_3 & x_3 & 0
		\end{pmatrix}^{\mathrm{T}},
	\]
	and $x_1 = 3c(e + 2f - \lambda) - 3b(g + w + o)$, $x_2 = (a - \lambda)(g + w + o) - 3bc$, $x_3 = 3b^2 - (a - \lambda)(e + 2f - \lambda)$, and the corresponding orthonormal eigenvectors are derived using the Gram–Schmidt orthogonalization procedure, with $n_1^2$ and $n_1^3$ denoting the normalization coefficients:
	\begin{equation} \label{cor.9}
		\vec{v}_1^2 = n_1^2\vec{f}_2,
	\end{equation}
	\begin{equation} \label{cor.10}
		\vec{v}_1^3 = n_1^3 \bigg( \vec{f}_3 - \big( \vec{f}_3, \vec{v}_1^2 \big)\vec{v}_1^2 \bigg).
	\end{equation}
	
	In the case of a double eigenvalue, alternative expressions for the corresponding eigenvectors may arise depending on which rows of the matrix
	\[
		\tau^1 - \lambda I
	\]
	are linearly independent. Expressions (\ref{cor.9}) and (\ref{cor.10}) are derived for the case where the first and second rows of the matrix are linearly independent. Expressions for the other cases are derived analogously by solving the corresponding system of linear equations.
	\[
		(\tau^1 - \lambda I)\vec{x} = \vec{0}.
	\]
	
	\begin{theorem}
		In the thermodynamic limit ($L \to \infty$), the two-spin correlation functions for TCGITSH parameter values satisfying the mirror-symmetry condition (\ref{mirror_hamiltonian}) are given by:
		\begin{equation} \label{cor.compact}
			\begin{gathered}
				G_{(p,0),(q,j)} = \sum_{k=2}^{4} \left(\frac{\lambda_k}{\lambda_{\max}}\right)^{q-p}  \left[a_1^1 a_1^k + a_2^1 a_2^k - a_3^1 a_3^k - a_4^1 a_4^k\right]^2  \\
				+ \sum_{r=5}^{6} \left(\frac{\lambda_r}{\lambda_{\max}}\right)^{q-p} \Phi_r^{(j)},
			\end{gathered}
		\end{equation}
		where $ j = 0,~1,~2$ and
		\begin{equation} \label{cor.compact.Phi}
			\Phi_r^{(j)} = 4 \sum_{s=2}^{3} |n_s^{r-4}|^2  \left| a_3^1 \, \omega^{\,j\varepsilon_s} - a_2^1 \, y_s^{r-4} \, \omega^{-j\varepsilon_s} \right|^2,
		\end{equation}
		with $\varepsilon_2 = +1$, $\varepsilon_3 = -1$, $\omega = e^{2\pi i / 3}$, and $\overline{\omega} = \omega^{-1}$. 
		\label{theorem5}
	\end{theorem}
	
	\textbf{Proof.}
	
	The two-spin correlation function \cite{Baxter, Yokota} is given by:
	\begin{equation}  \label{cor.58}
		G_{(p,\alpha),(q,\beta)} = \big\langle\sigma_{\alpha}^p \sigma_{ \beta}^q\big\rangle - \big\langle\sigma_{ \alpha}^p\big\rangle\big\langle\sigma_{\beta}^q\big\rangle,
	\end{equation}
	where the mean values are given by \cite{Yokota, Baxter, Fisher}:
	\begin{equation}  \label{cor.59}
		\big\langle\sigma_{\alpha}^p \sigma_{\beta}^q\big\rangle = \frac{\operatorname{Tr}\big(S_\alpha\theta^{q-p}S_\beta\theta^{L+p-q}\big)}{Z_L},
	\end{equation}
	\begin{equation}  \label{cor.60}
		\big\langle\sigma_\gamma^k\big\rangle = \frac{\operatorname{Tr}\big(S_\gamma\theta^{L}\big)}{Z_L},
	\end{equation}
	where $S_l$ ($l = 0,~1,~2$) are the diagonal spin-layer matrices, which for the present transfer matrix take the form:
	\begin{equation} \label{cor.61}
		S_0 = \operatorname{diag} \big(1,~-1,~1,~-1,~1,~-1,~1,~-1\big),
	\end{equation}
	\begin{equation} \label{cor.62}
		S_1 = \operatorname{diag} \big(1,~1,~-1,~-1,~1,~1,~-1,~-1\big),
	\end{equation}
	\begin{equation} \label{cor.63}
		S_2 = \operatorname{diag} \big(1,~1,~1,~1,~-1,~-1,~-1,~-1\big).
	\end{equation}
	
	The expression (\ref{cor.59}) can be rewritten as follows:
	\begin{eqnarray} \label{cor.64}
		\big\langle\sigma_{\alpha}^p \sigma_{\beta}^q\big\rangle &=& \frac{\operatorname{Tr}\big(\theta^{p-1}S_{\alpha}\theta^{q-p}S_{\beta}\theta^{L-q+1}\big)}{\operatorname{Tr}\theta^L} \\ \nonumber
		&=& \frac{\operatorname{Tr}\big(\Lambda^{p-1}S'_{\alpha}\Lambda^{q-p}S'_{\beta}\Lambda^{L-q+1}\big)}{\operatorname{Tr}\Lambda^L}.
	\end{eqnarray}
	
	Similarly, the expression (\ref{cor.60}) has the form:
	\begin{equation}  \label{cor.65}
		\big\langle\sigma_\gamma^k\big\rangle = \frac{\operatorname{Tr}\big(S'_\gamma\Lambda^{L}\big)}{\operatorname{Tr}\Lambda^L},
	\end{equation}
	where $\Lambda$ is the diagonal matrix obtained by diagonalizing the transfer matrix, and
	\begin{equation}  \label{cor.66}
		S'_l = P^{-1}S_lP,
	\end{equation} 
	where $P$ is the unitary matrix in Eq.~(\ref{cor.2}) that diagonalizes the transfer matrix, and its inverse matrix:
	\begin{equation} \label{cor.68}
		P^{-1} = P^*.
	\end{equation}
	
	Substituting Eqs.~(\ref{cor.64}) and~(\ref{cor.65}) into Eq.~(\ref{cor.58}), we obtain formulas for the pair correlation functions. In the thermodynamic limit ($L\to\infty$), these expressions reduce to Eq.~(\ref{cor.compact}), and the rotational invariance of the system implies $G_{(p,0),(q,0)} = G_{(p,1),(q,1)} = G_{(p,2),(q,2)}$, $G_{(p,0),(q,1)} = G_{(p,1),(q,2)} = G_{(p,2),(q,0)}$, and $G_{(p,0),(q,2)} = G_{(p,1),(q,0)} = G_{(p,2),(q,1)}$.
	
	\subsection{\label{sec3:Magnetization}Magnetization in the thermodynamic limit from the eigenvector associated with the largest transfer-matrix eigenvalue}

	Throughout this subsection, the transfer matrix is assumed to be symmetric, which is ensured when the Hamiltonian satisfies mirror-symmetry~(\ref{mirror_hamiltonian}).

	By definition, the magnetization is the thermal average of a spin. It can be computed from Eq.~(\ref{cor.60}).
	
	Under these assumptions, the magnetization can be written as:
	\begin{equation} \label{73}
		M = \frac{\operatorname{Tr}\big(S'_{\gamma}\Lambda'^{L}\big)}{\operatorname{Tr}\big(\Lambda'^{L}\big)},
	\end{equation}
	where
	\[
		\Lambda' = \operatorname{diag}\bigg(1, \frac{\lambda_2}{\lambda_{\max}}, \ldots, \frac{\lambda_8}{\lambda_{\max}}\bigg).
	\]
	
	In the thermodynamic limit, the magnetization $(\ref{73})$ is equal to:
	\begin{equation}  \label{74}
		M = \operatorname{Tr}\big(S'_{\gamma}I_{11}\big),
	\end{equation}
	where
	\[
		I_{11} = 
		\begin{cases}
			1, & \text{when}~i = 1~\text{and}~j = 1, \\
			0, & \text{otherwise}.
		\end{cases}
	\]
	
	Using Eq.~(\ref{cor.66}), we can rewrite Eq.~(\ref{74}) as
	\[
		M = \operatorname{Tr}\big(P^{-1}S_{\gamma}PI_{11}\big) = \operatorname{Tr}\big(S_{\gamma}PI_{11}P^{*}\big).
	\]
	
	As a result, the magnetization in the thermodynamic limit has the form:
	\begin{equation} \label{75}
		M = v_1^2 + v_2^2 - v_3^2 - v_4^2,
	\end{equation}
	where $v_1$, $v_2$, $v_3$, $v_4$ are the components of the normalized eigenvector of the transfer matrix corresponding to the largest eigenvalue:
	\[
	\vec{v} = 
		\begin{pmatrix}
			v_1&v_2&v_2&v_3&v_2&v_3&v_3&v_4
		\end{pmatrix}^{\mathrm{T}},
	\]
	normalized to unit Euclidean norm:
	\[
		||\vec{v}||_2 = 1.
	\]
	
	\subsection{Pair correlations for the PSCSH model}
	
	The Hamiltonian is given by Eq.~(\ref{5}), where:
	\begin{widetext}
		\begin{eqnarray} \label{46}
			&&\mathcal{H}^{(m)}  \nonumber \\
			&&= -\frac{J_1}{2}\bigg(\sigma_{0}^m\sigma_{1}^m + \sigma_{0}^m\sigma_{2}^m + \sigma_{1}^m\sigma_{2}^m + \sigma_{0}^{m+1}\sigma_{1}^{m+1} + \sigma_{0}^{m+1}\sigma_{2}^{m+1} + \sigma_{1}^{m+1}\sigma_{2}^{m+1}\bigg) \nonumber \\
			&&-J_2\bigg(\sigma_{0}^m\sigma_{0}^{m+1} 
			+ \sigma_{1}^m\sigma_{1}^{m+1} + \sigma_{2}^m\sigma_{2}^{m+1}\bigg)
			-J_3\bigg(\sigma_{0}^m\sigma_{1}^{m+1} + \sigma_{1}^m\sigma_{0}^{m+1} + \sigma_{2}^m\sigma_{0}^{m+1} + \sigma_{0}^m\sigma_{2}^{m+1} +\sigma_{1}^m\sigma_{2}^{m+1} + \sigma_{2}^m\sigma_{1}^{m+1}\bigg) \nonumber \\
			&&-J_4\bigg(\sigma_{0}^m\sigma_{1}^m\sigma_{2}^m\sigma_{0}^{m+1} + \sigma_{0}^m\sigma_{1}^m\sigma_{2}^m\sigma_{1}^{m+1} + \sigma_{0}^m\sigma_{1}^m\sigma_{2}^m\sigma_{2}^{m+1}
			+\sigma_{0}^m\sigma_{0}^{m+1}\sigma_{1}^{m+1}\sigma_{2}^{m+1}
			+\sigma_{1}^m\sigma_{0}^{m+1}\sigma_{1}^{m+1}\sigma_{2}^{m+1} + \sigma_{2}^m\sigma_{0}^{m+1}\sigma_{1}^{m+1}\sigma_{2}^{m+1}\bigg) \nonumber \\
			&&-J_5\bigg(\sigma_{0}^m\sigma_{1}^m\sigma_{0}^{m+1}\sigma_{1}^{m+1} + \sigma_{0}^m\sigma_{2}^m\sigma_{0}^{m+1}\sigma_{2}^{m+1} + \sigma_{1}^m\sigma_{2}^m\sigma_{1}^{m+1}\sigma_{2}^{m+1}\bigg) \nonumber \\
			&&-J_6\bigg(\sigma_{0}^m\sigma_{1}^m\sigma_{0}^{m+1}\sigma_{2}^{m+1} + \sigma_{0}^m\sigma_{2}^m\sigma_{1}^{m+1}\sigma_{2}^{m+1}
			+\sigma_{1}^m\sigma_{2}^m\sigma_{0}^{m+1}\sigma_{1}^{m+1} + \sigma_{0}^m\sigma_{1}^m\sigma_{1}^{m+1}\sigma_{2}^{m+1} +\sigma_{0}^m\sigma_{2}^m\sigma_{0}^{m+1}\sigma_{1}^{m+1} + \sigma_{1}^m\sigma_{2}^m\sigma_{0}^{m+1}\sigma_{2}^{m+1}\bigg) \nonumber \\ &&-J_7\sigma_{0}^m\sigma_{1}^m\sigma_{2}^m\sigma_{0}^{m+1}\sigma_{1}^{m+1}\sigma_{2}^{m+1},
		\end{eqnarray}
	\end{widetext}
	and the transfer matrix (\ref{8}) has the form
	\begin{equation} \label{47}
		\theta = 
		\begin{pmatrix}
			a&b&b&c&b&c&c&d\\
			b&e&f&g&f&g&w&c\\
			b&f&e&g&f&w&g&c\\
			c&g&g&e&w&f&f&b\\
			b&f&f&w&e&g&g&c\\
			c&g&w&f&g&e&f&b\\
			c&w&g&f&g&f&e&b\\
			d&c&c&b&c&b&b&a
		\end{pmatrix},
	\end{equation}
	where
	\begin{equation*}
		a = \exp\bigg\{\frac{3J_1 + 3J_2 + 6J_3 + 6J_4 + 3J_5 + 6J_6 + J_7}{T}\bigg\},
	\end{equation*}
	\begin{equation*}
		b = \exp\bigg\{\frac{J_1 + J_2 + 2J_3 - 2J_4 - J_5 - 2J_6 - J_7}{T}\bigg\},
	\end{equation*}
	\begin{equation*}
		c = \exp\bigg\{\frac{J_1 - J_2 - 2J_3 + 2J_4 - J_5 - 2J_6 + J_7}{T}\bigg\},
	\end{equation*}
	\begin{equation*}
		d = \exp\bigg\{\frac{3J_1 - 3J_2 - 6J_3 - 6J_4 + 3J_5 + 6J_6 - J_7}{T}\bigg\},
	\end{equation*}
	\begin{equation*}
		e = \exp\bigg\{\frac{- J_1 + 3J_2 - 2J_3 - 2J_4 + 3J_5 - 2J_6 + J_7}{T}\bigg\},
	\end{equation*}
	\begin{equation*}
		f = \exp\bigg\{\frac{- J_1 - J_2 + 2J_3 - 2J_4 - J_5 + 2J_6 + J_7}{T}\bigg\},
	\end{equation*}
	\begin{equation*}
		g = \exp\bigg\{\frac{-J_1 + J_2 - 2J_3 + 2J_4 - J_5 + 2J_6 - J_7}{T}\bigg\},
	\end{equation*}
	\begin{equation*}
		w = \exp\bigg\{\frac{-J_1 - 3J_2 + 2J_3 + 2J_4 + 3J_5 - 2J_6 - J_7}{T}\bigg\},
	\end{equation*}
	and its eigenvalues are given by Eqs.~(\ref{40})--(\ref{45}):
	\begin{eqnarray} \label{49} 
		\lambda_{1,2}&=&\frac{a + d + e + 2f + 2g + w \pm \sqrt{C_1}}{2},
	\end{eqnarray}
	\begin{eqnarray}  \label{50}
		\lambda_{3,4}&=&\frac{a - d + e + 2f - 2g - w \pm \sqrt{C_2}}{2},
	\end{eqnarray}
	\begin{eqnarray} \label{51}
		\lambda_{5,6} &=& e + w - g - f,
	\end{eqnarray}
	\begin{eqnarray} \label{52}
		\lambda_{7,8} &=& e - w + g - f,
	\end{eqnarray}
	where:
	\begin{equation} \label{53}
		\begin{gathered}
			C_1 = \big(a+d-e-2f-2g-w\big)^2 +12\big(b+c\big)^2,\\
			C_2 = \big(-a+d+e+2f-2g-w\big)^2 + 12\big(b-c\big)^2.
		\end{gathered}
	\end{equation}
	
	\begin{theorem}
		In the thermodynamic limit, the two-spin correlations for the model with Hamiltonian (\ref{46}) are given by:
		\begin{eqnarray} \label{54} 
			&&G_{(p,0),(q,0)} = G_{(p,1),(q,1)} = G_{(p,2),(q,2)} \nonumber \\
			&&=\biggl(A_1B_2 + \frac{1}{3}B_1A_2\biggr)^2\biggl(\frac{\lambda_3}{\lambda_1}\biggr)^{q-p} \nonumber \\
			&&+\biggl(A_1A_2 - \frac{1}{3}B_1B_2\biggr)^2\biggl(\frac{\lambda_4}{\lambda_1}\biggr)^{q-p} +\frac{8}{9}B_1^2\biggl(\frac{\lambda_7}{\lambda_1}\biggr)^{q-p},
		\end{eqnarray}
		\begin{eqnarray} \label{55} 
			&&G_{(p,0),(q,1)} = G_{(p,0),(q,2)} = G_{(p,1),(q,2)} \nonumber \\
			&&=\biggl(A_1B_2 + \frac{1}{3}B_1A_2\biggr)^2\biggl(\frac{\lambda_3}{\lambda_1}\biggr)^{q-p} \nonumber \\
			&&+\biggl(A_1A_2 - \frac{1}{3}B_1B_2\biggr)^2\biggl(\frac{\lambda_4}{\lambda_1}\biggr)^{q-p}
			-\frac{4}{9}B_1^2\biggl(\frac{\lambda_7}{\lambda_1}\biggr)^{q-p},
		\end{eqnarray}
		where $A_1,~B_1,~A_2,~B_2$ are defined by Eq.~(\ref{56}) in the Appendix, and the average value of one spin is zero:
		\begin{equation}  \label{57}
			\big\langle\sigma\big\rangle  = 0.
		\end{equation}
		\label{theorem6}
	\end{theorem}

	In this case, the unitary matrix that diagonalizes the transfer matrix is given by:
	\begin{eqnarray} \label{67}
		&&P = \nonumber \\
		&&\frac{1}{\sqrt{6}}
		\begin{pmatrix}
			\sqrt{3}A_1&\sqrt{3}B_1&0&0&\sqrt{3}B_2&-\sqrt{3}A_2&0&0\\
			B_1&-A_1&\omega&1&A_2&B_2&-\omega&-1\\
			B_1&-A_1&\omega^2&\omega^2&A_2&B_2&-\omega^2&-\omega^2\\
			B_1&-A_1&1&\omega&-A_2&-B_2&1&\omega\\
			B_1&-A_1&1&\omega&A_2&B_2&-1&-\omega\\
			B_1&-A_1&\omega^2&\omega^2&-A_2&-B_2&\omega^2&\omega^2\\
			B_1&-A_1&\omega&1&-A_2&-B_2&\omega&1\\
			\sqrt{3}A_1&\sqrt{3}B_1&0&0&-\sqrt{3}B_2&\sqrt{3}A_2&0&0
		\end{pmatrix}, \nonumber \\
	\end{eqnarray}
	
	From Eqs.~(\ref{cor.66}), (\ref{cor.68}), and (\ref{67}), we obtain the matrix $S'_0$:
	\begin{equation} \label{70} 
		S'_{0} =
		\begin{pmatrix}
			0& S'_{0_{12}} \\
			S'_{0_{21}} & 0
		\end{pmatrix},
	\end{equation}
	and matrix $S'_1$:
	\begin{equation} \label{71} 
		S'_{1} = 
		\begin{pmatrix}
			0&S'_{1_{12}}\\
			S'_{1_{21}}&0
		\end{pmatrix},
	\end{equation}
	where the matrices $S'_{0_{12}}, S'_{0_{21}}, S'_{1_{12}}, S'_{1_{21}}$ are given by Eqs.~(\ref{C2})--(\ref{C5}) in the Appendix.
	
	Using Eqs.~(\ref{70}) and~(\ref{71}) in Eq.~(\ref{cor.64}), we obtain expressions for $\big\langle\sigma_{0}^p \sigma_{0}^q\big\rangle$ and $\big\langle\sigma_{0}^p \sigma_{1}^q\big\rangle$, which in the thermodynamic limit ($L \to \infty$) coincide with Eqs.~(\ref{54}) and~(\ref{55}); the rotational invariance of the system then yields $\big\langle\sigma_{1}^p \sigma_{1}^q \big\rangle = \big\langle\sigma_{2}^p \sigma_{2}^q \big\rangle = \big\langle\sigma_{0}^p \sigma_{0}^q \big\rangle$ and $\big\langle\sigma_{0}^p \sigma_{2}^q \big\rangle = \big\langle\sigma_{1}^p \sigma_{2}^q \big\rangle = \big\langle\sigma_{0}^p \sigma_{1}^q \big\rangle$.

	Similarly, using Eqs.~(\ref{70}) and~(\ref{71}) in Eq.~(\ref{cor.65}), we obtain
	\begin{equation}  \label{72}
		\big\langle\sigma_{0}^{k_0}\big\rangle = \big\langle\sigma_{1}^{k_1}\big\rangle = \big\langle\sigma_{2}^{k_2}\big\rangle = 0
	\end{equation}
	
	Equation~(\ref{72}) implies that the magnetization vanishes in this case.  Therefore, the two-spin correlation functions defined in Eq.~(\ref{cor.58}) reduce to Eqs. (\ref{54}) and (\ref{55}).
	
	We note that the expressions for the pair correlations~(\ref{54}) and~(\ref{55}), as well as the spectrum of the transfer matrix in Eqs.~(\ref{49})--(\ref{52}), coincide with the known results of \cite{Yokota}, with which our model agrees when the parameters $J_4,\ldots,J_7$ in the Hamiltonian~(\ref{46}) are set to zero.
	
	We next examine two illustrative ground-state examples for the PSCSH model and characterize the behavior of the inverse correlation length at the boundaries between ground-state regions. Both examples are given in Appendix \ref{subsecD:GroundStates}.

	\section{\label{sec:PlanarModel}Width-three planar model with nearest-neighbor, next-nearest-neighbor and plaquette interactions}
	
	The three-chain tube is the simplest toroidal geometry whose unfolding produces a planar lattice model.
	
	We next consider a width-three planar model with nearest-neighbor, next-nearest-neighbor, and plaquette interactions; each layer contains three plaquettes.
	
	This model is obtained by unfolding the lattice shown in Fig. \ref{figure1} into the plane, with the boundary conditions $t_3^m \equiv t_0^m$ and $t_3^{m+1} \equiv t_0^{m+1}$ (Fig. \ref{figure8}).

	\begin{figure}[!htbp]
		\centering
		\includegraphics[width=0.85\linewidth]{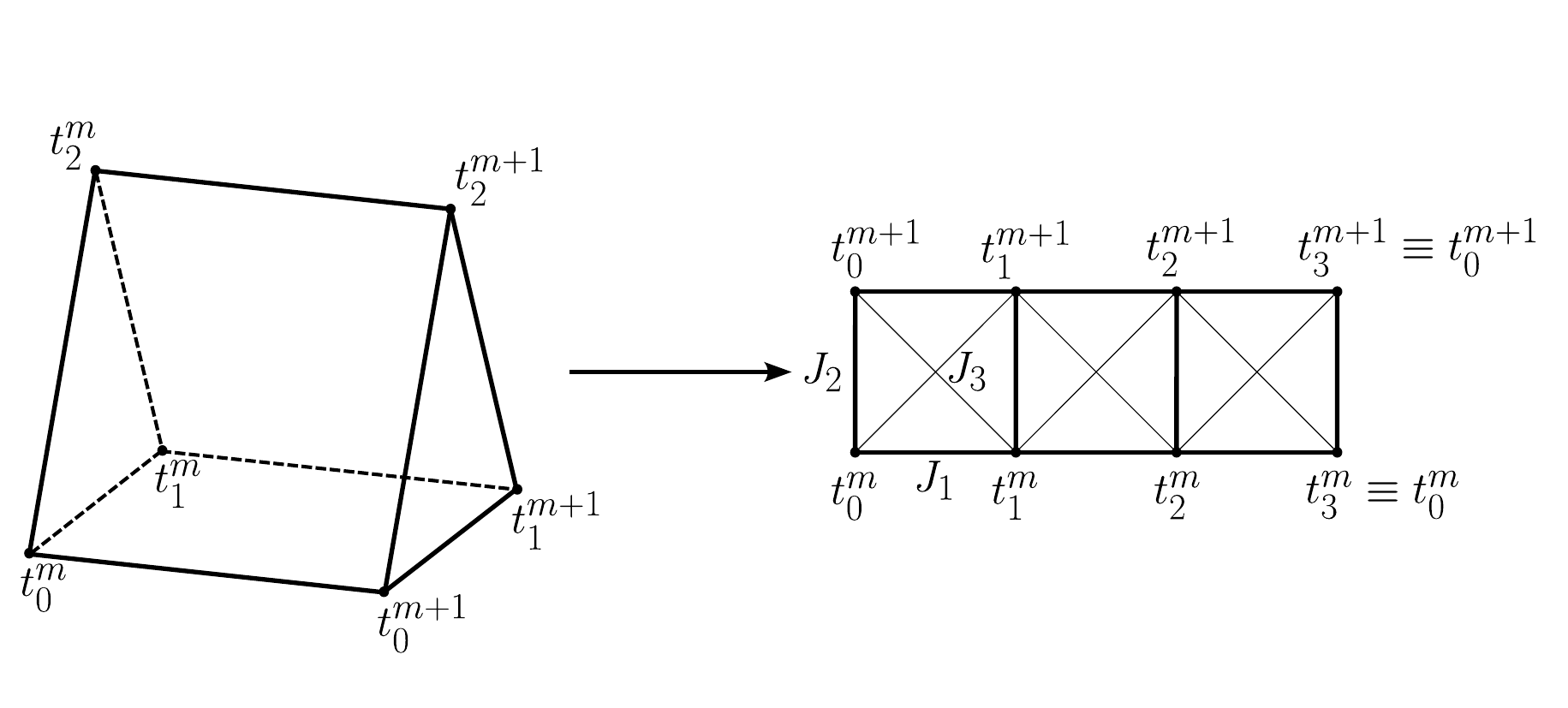}
		\caption{\label{figure8} Planar Ising model with nearest-neighbor, next-nearest-neighbor, and plaquette interactions obtained by unfolding the three-chain Ising tube into the plane.}
	\end{figure}

	Here $J_1^{\operatorname{pl}}$ corresponds to the nearest-neighbor interaction within a layer, $J_2^{\operatorname{pl}}$ corresponds to the nearest-neighbor interaction on adjacent layers, $J_3^{\operatorname{pl}}$ corresponds to the next-nearest-neighbor interaction, and the coupling $J_4^{\operatorname{pl}}$ is not shown explicitly in Fig.~\ref{figure8}; it denotes the plaquette interaction.
	
	The Hamiltonian is given by Eq.~(\ref{5}), where the Hamiltonian of one layer of the model is of the form:
	\begin{widetext}
		\begin{eqnarray} \label{77}
			\mathcal{H}^{(m)}&& \nonumber \\ 
			&=&-\frac{J_1^{\operatorname{pl}}}{2}\bigg(
			\sigma_{0}^m\sigma_{1}^m +
			\sigma_{0}^m\sigma_{2}^m +
			\sigma_{1}^m\sigma_{2}^m +
			\sigma_{0}^{m+1}\sigma_{1}^{m+1} +
			\sigma_{0}^{m+1}\sigma_{2}^{m+1} +
			\sigma_{1}^{m+1}\sigma_{2}^{m+1}
			\bigg) \nonumber \\ 
			&&-J_2^{\operatorname{pl}}\bigg(
			\sigma_{0}^{m}\sigma_{0}^{m+1} +
			\sigma_{1}^{m}\sigma_{1}^{m+1} +
			\sigma_{2}^{m}\sigma_{2}^{m+1}
			\bigg) \nonumber \\
			&&-J_3^{\operatorname{pl}}\bigg(
			\sigma_{0}^{m}\sigma_{1}^{m+1} +
			\sigma_{1}^{m}\sigma_{2}^{m+1} + \sigma_{2}^{m}\sigma_{0}^{m+1} +
			\sigma_{0}^{m}\sigma_{2}^{m+1} +
			\sigma_{2}^{m}\sigma_{1}^{m+1} +
			\sigma_{1}^{m}\sigma_{0}^{m+1} 
			\bigg) \nonumber \\
			&&-J_4^{\operatorname{pl}} \bigg(
			\sigma_{0}^{m}\sigma_{1}^{m}\sigma_{0}^{m+1}\sigma_{1}^{m+1} +
			\sigma_{1}^{m}\sigma_{2}^{m}\sigma_{1}^{m+1}\sigma_{2}^{m+1} +\sigma_{0}^{m}\sigma_{2}^{m}\sigma_{0}^{m+1} \sigma_{2}^{m+1}
			\bigg),
		\end{eqnarray}
	\end{widetext}
	where the coupling constants correspond to the following parameters of the TCGIT model, as given in Eqs.~(\ref{H2}) and (\ref{H4}):
	\[
		\begin{gathered}
			J_1^{\operatorname{pl}} = J_{\{t_0^m,t_1^m\}}, \\
			J_2^{\operatorname{pl}} = J_{\{t_0^m,t_0^{m+1}\}}, \\
			J_3^{\operatorname{pl}} = J_{\{t_0^m,t_1^{m+1}\}} = J_{\{t_0^m,t_2^{m+1}\}} , \\
			J_4^{\operatorname{pl}} = J_{\{t_0^m,t_1^m, t_0^{m+1}, t_1^{m+1}\}}
		\end{gathered}
	\]
	For this planar model, we again use an $8 \times 8$ transfer matrix whose structure is the same as that in Eq.~(\ref{47}), with entries $a, b, c, d, e, f, g$, and $w$ redefined according to Eq.~(\ref{77}). The eigenvalues of the transfer matrix of this planar model are then given by Eqs.~(\ref{49})--(\ref{52}), and all physical observables are described by the expressions obtained for the PSC model: Theorems~\ref{theorem3} and~\ref{theorem4} apply to this model as well, as do the expressions for pair correlations in Theorem~\ref{theorem6}.
	
	As an illustration, Fig.~\ref{figure10} shows the behavior of the free energy, internal energy, specific heat, and entropy for several parameter choices $J_1^{\operatorname{pl}}, \ldots, J_4^{\operatorname{pl}}$. One of these cases corresponds to the so-called planar gonihedric model, where $J_1^{\operatorname{pl}} = J_2^{\operatorname{pl}} = k$, $J_3^{\operatorname{pl}} = -\frac{k}{2}$, and $J_4^{\operatorname{pl}} = \frac{1-k}{2}$ \cite{Bathas}.

	\begin{figure*}[!htbp]
		\centering
		\begin{minipage}{0.24\linewidth}
			\includegraphics[width=\linewidth]{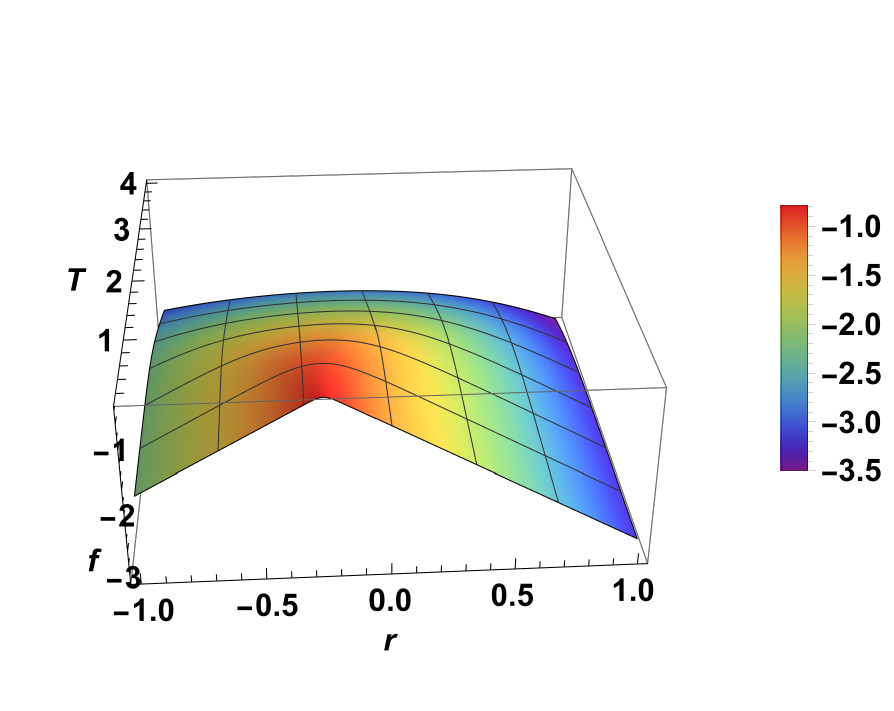} \\ (a)
		\end{minipage}
		\hfill
		\begin{minipage}{0.24\linewidth}
			\includegraphics[width=\linewidth]{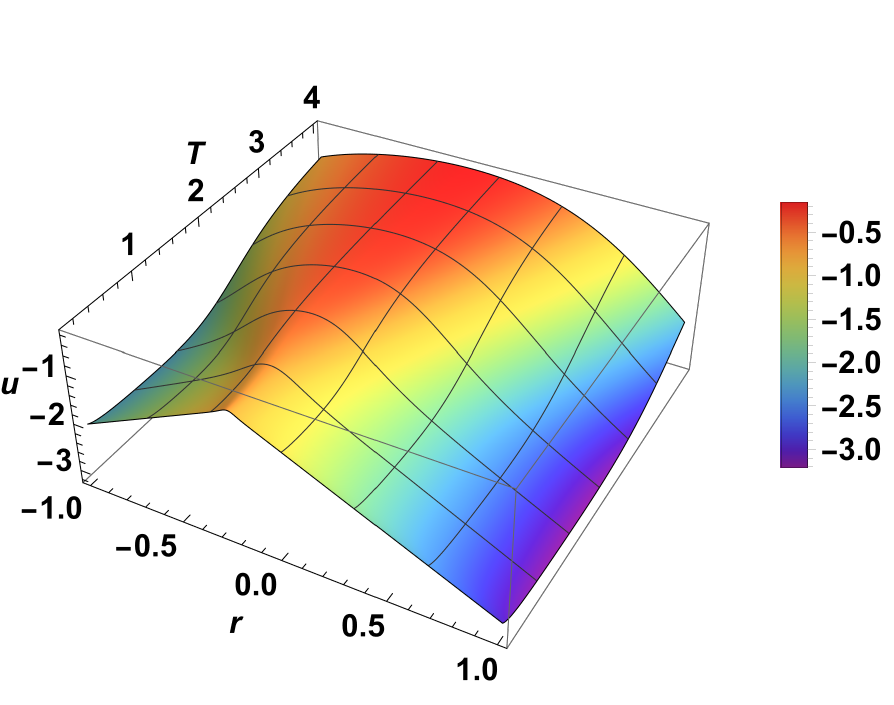} \\ (b)
		\end{minipage}
		\hfill
		\begin{minipage}{0.24\linewidth}
			\includegraphics[width=\linewidth]{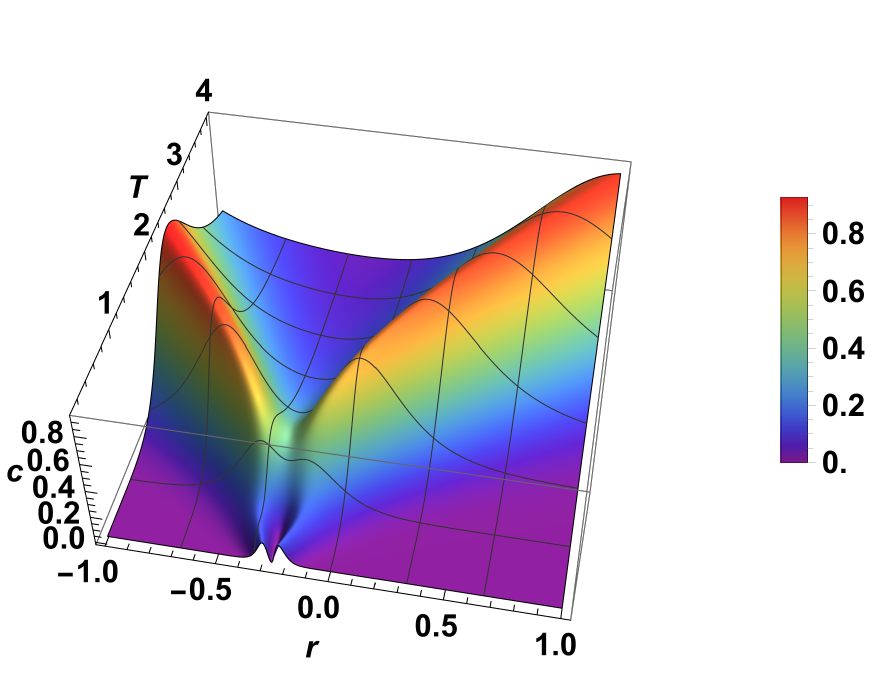}  \\ (c)
		\end{minipage}
		\hfill
		\begin{minipage}{0.24\linewidth}
			\includegraphics[width=\linewidth]{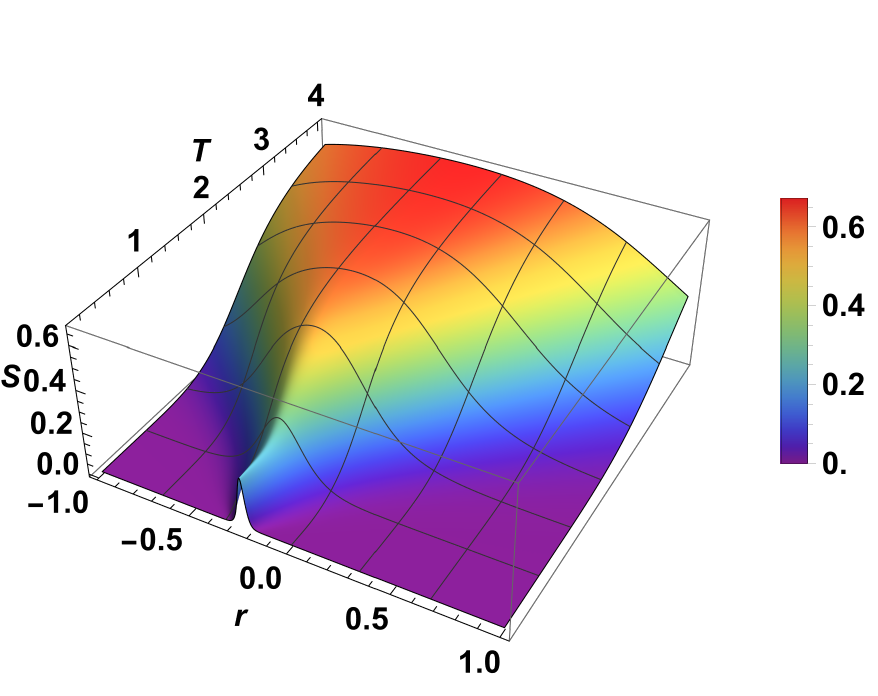}  \\ (d)
		\end{minipage}
		\caption{\label{figure10} Planar model with nearest-neighbor, next-nearest-neighbor, and plaquette interactions. Plots of free energy (a), internal energy (b), specific heat (c) and entropy (d) for the case of $J_1^{\operatorname{pl}} = J_2^{\operatorname{pl}} = \frac{1}{2}, J_4^{\operatorname{pl}} = \frac{1}{4}$ and $J_3^{\operatorname{pl}} = r \in [-1, 1], T\in[0.1, 4]$.}
	\end{figure*}

	\subsection{Width-three planar gonihedric model}
	
	An important special case of the width-three planar model with nearest-neighbor, next-nearest-neighbor, and plaquette interactions considered above is the planar gonihedric model, one choice of parameters shown in Fig.~\ref{figure10}.
	
	The Hamiltonian of one layer of the planar gonihedric model is given by Eq.~(\ref{77}), where \cite{Bathas}:
	\begin{equation} \label{78}
		\begin{gathered}
			J_1^{\operatorname{pl}} = J_2^{\operatorname{pl}} = k, \\
			J_3^{\operatorname{pl}} = -\frac{k}{2}, \\
			J_4^{\operatorname{pl}} = \frac{1-k}{2},
		\end{gathered}
	\end{equation}
	i.e., all spin-spin interactions $J_q^{\operatorname{pl}}$ in the gonihedric model depend on one parameter $k$.
	
	 \textbf{Proposition 1.} Theorems \ref{theorem3}, \ref{theorem4}, and \ref{theorem6} apply to the planar gonihedric model as well as to the more general planar model considered above.
	
	One of the parameter choices shown in Fig.~\ref{figure10} corresponds to the planar gonihedric model (when $J_3^{\operatorname{pl}} = -\frac{1}{4}$).
	
	In Fig.~\ref{figure10}, the distinctive low-temperature behavior occurs precisely at $J_3^{\operatorname{pl}}=-\tfrac{1}{4}$, which corresponds to the planar gonihedric model ($J_3^{\operatorname{pl}}/J_1^{\operatorname{pl}} = -\tfrac{1}{2}$). Additional plots not shown here, for other positive values of $J_1^{\operatorname{pl}} = J_2^{\operatorname{pl}}$ and $J_4^{\operatorname{pl}}$, exhibit the same low-temperature behavior: the distinctive feature appears precisely when $J_3^{\operatorname{pl}}/J_1^{\operatorname{pl}} = -\tfrac{1}{2}$.
	
	To illustrate the low-temperature dependence on the parameter $k$, we plot the free energy, internal energy, specific heat, and entropy as functions of temperature $T$ and $k$ (Fig.~\ref{figure11}). For non-negative values of $k$, the entropy does not approach zero as the temperature tends to zero.

	\begin{figure*}[!htbp]
		\centering
		\begin{minipage}{0.24\linewidth}
			\includegraphics[width=\linewidth]{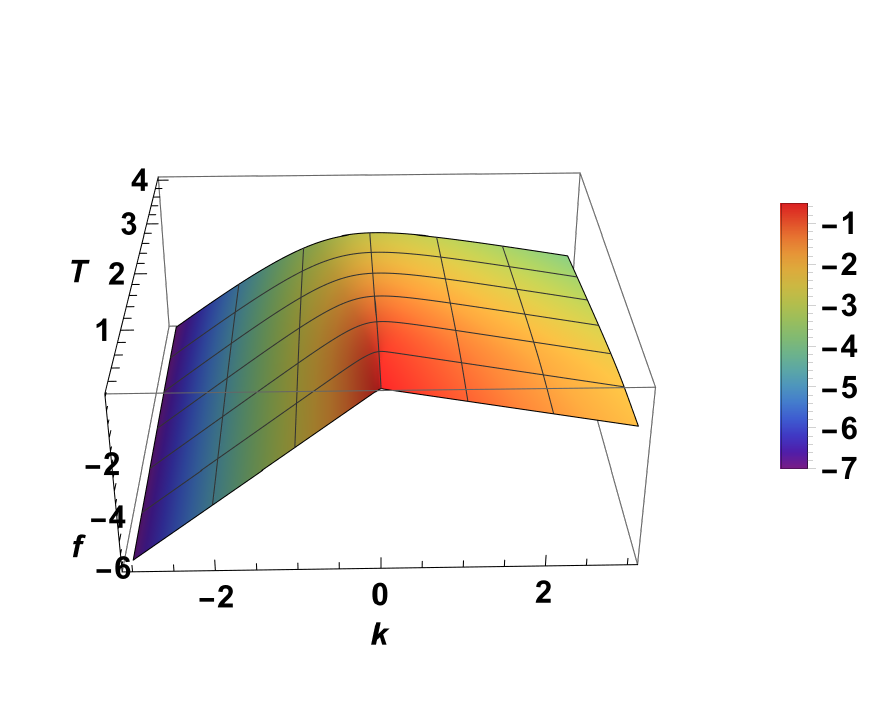}   \\ (a)
		\end{minipage}
		\hfill
		\begin{minipage}{0.24\linewidth}
			\includegraphics[width=\linewidth]{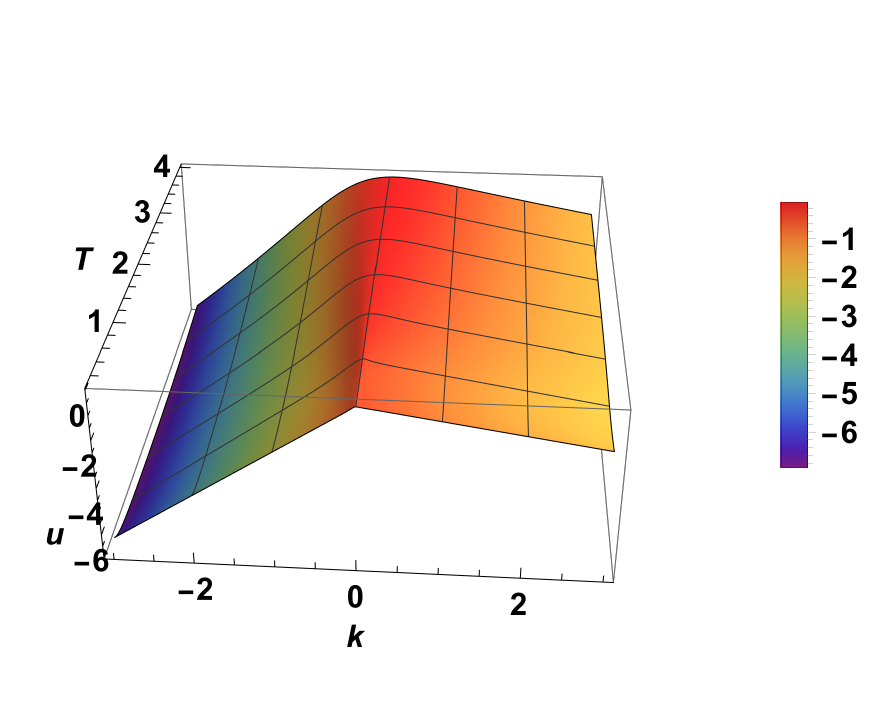}  \\ (b)
		\end{minipage}
		\hfill
		\begin{minipage}{0.24\linewidth}
			\includegraphics[width=\linewidth]{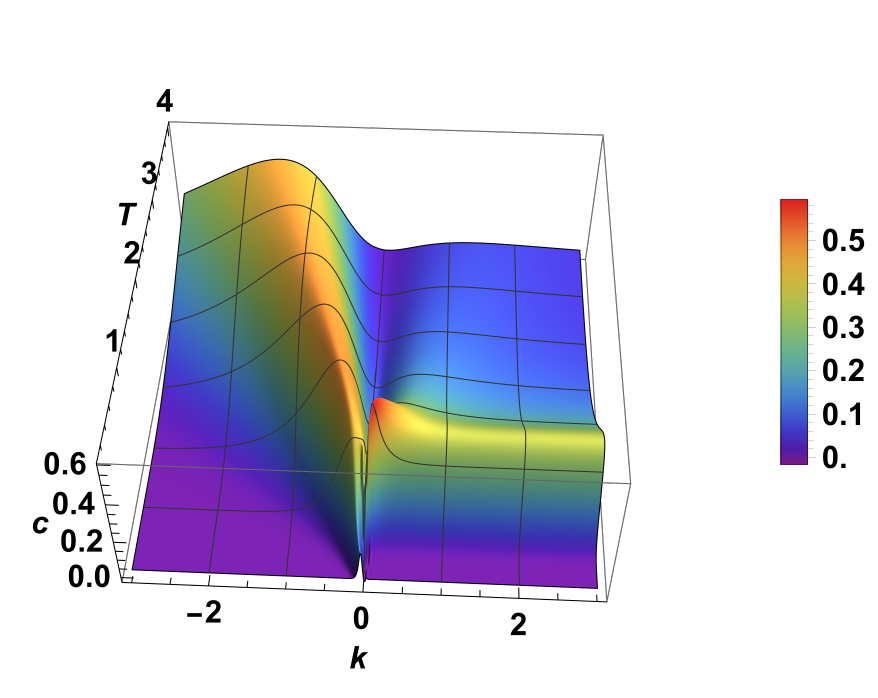}   \\ (c)
		\end{minipage}
		\hfill
		\begin{minipage}{0.24\linewidth}
			\includegraphics[width=\linewidth]{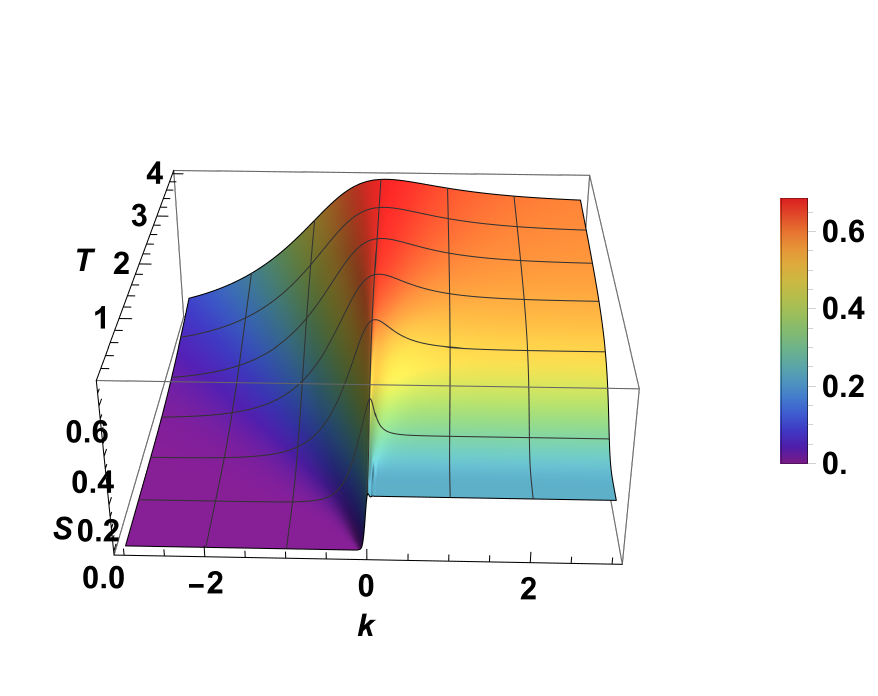}  \\ (d)
		\end{minipage}
		\caption{\label{figure11} Planar gonihedric model. Plots of free energy (a), internal energy (b), specific heat (c) and entropy (d) of the gonihedric model at $k \in [-3, 3], T\in[0.1, 4]$.}
	\end{figure*}
	
	\textbf{Proposition 2.} For the planar gonihedric model, the entropy approaches the following value as $T \to 0+$:
	\begin{equation}
		\lim_{T \to 0+}S(k, T) = 
		\begin{cases}
			\frac{\ln2}{3}, ~~ k \ge 0, \\
			0, ~~ k < 0
		\end{cases}
		\label{entropy9}
	\end{equation}
	
	Let $a = (\sigma_0^m, \sigma_1^m, \sigma_2^m)$ and $b = (\sigma_0^{m+1}, \sigma_1^{m+1}, \sigma_2^{m+1})$ denote the spin configurations of two adjacent layers, so that $a, b \in \{\pm 1\}^3$. The interlayer interaction energy can be written explicitly as
	\begin{eqnarray*}
		H(a,b,k) &=& \\
		&& -\frac{k}{2}\bigg[(a_0 a_1 + a_1 a_2 + a_2 a_0) + (b_0 b_1 + b_1 b_2 + b_2 b_0)\bigg] \\
		&& - k\big(a_0 b_0 + a_1 b_1 + a_2 b_2\big) \\
		&& + \frac{k}{2}(a_0 b_1 + a_1 b_2 + a_2 b_0 + a_0 b_2 + a_2 b_1 + a_1 b_0) \\
		&& - \frac{1-k}{2}(a_0 a_1 b_0 b_1 + a_1 a_2 b_1 b_2 + a_2 a_0 b_2 b_0).
	\end{eqnarray*}

	Direct enumeration of the 64 pairs $(a,b)$ shows that $H(a,b)$ takes exactly four distinct values, each linear in $k$:
	$$H(a,b) \in \Bigl\{ -\tfrac{3}{2}(k+1),\; -\tfrac{3k-1}{2},\; \tfrac{5k+1}{2},\; \tfrac{13k-3}{2} \Bigr\}.$$
	Define
	$$H_{\min}(k) := \min_{a,b} H(a,b,k).$$

	The energy gap $\Delta(k)$ separating $H_{\min}(k)$ from the next-lowest value is strictly positive for all $k$. The transfer matrix can therefore be written as
	$$\theta(T) = e^{-H_{\min}/T}\bigl(A + R(T)\bigr),$$
	where the adjacency matrix $A = A(k)$ of the minimum-energy transitions is the $\{0,1\}$-matrix
	$$A_{ab} =
	\begin{cases}
		1, & H_{ab} = H_{\min}(k), \\
		0, & \text{otherwise};
	\end{cases}$$
	explicitly,
	\[
	A(k<0) = \begin{pmatrix}
		0 & 0 & 0 & 0 & 0 & 0 & 0 & 0 \\
		0 & 0 & 0 & 0 & 0 & 0 & 1 & 0 \\
		0 & 0 & 0 & 0 & 0 & 1 & 0 & 0 \\
		0 & 0 & 0 & 0 & 1 & 0 & 0 & 0 \\
		0 & 0 & 0 & 1 & 0 & 0 & 0 & 0 \\
		0 & 0 & 1 & 0 & 0 & 0 & 0 & 0 \\
		0 & 1 & 0 & 0 & 0 & 0 & 0 & 0 \\
		0 & 0 & 0 & 0 & 0 & 0 & 0 & 0
	\end{pmatrix},
	\]
	\[
	A(k=0) = \begin{pmatrix}
		1 & 0 & 0 & 0 & 0 & 0 & 0 & 1 \\
		0 & 1 & 0 & 0 & 0 & 0 & 1 & 0 \\
		0 & 0 & 1 & 0 & 0 & 1 & 0 & 0 \\
		0 & 0 & 0 & 1 & 1 & 0 & 0 & 0 \\
		0 & 0 & 0 & 1 & 1 & 0 & 0 & 0 \\
		0 & 0 & 1 & 0 & 0 & 1 & 0 & 0 \\
		0 & 1 & 0 & 0 & 0 & 0 & 1 & 0 \\
		1 & 0 & 0 & 0 & 0 & 0 & 0 & 1
	\end{pmatrix},
	\]
	\[
	A(k>0) = \begin{pmatrix}
		1 & 0 & 0 & 0 & 0 & 0 & 0 & 1 \\
		0 & 1 & 0 & 0 & 0 & 0 & 0 & 0 \\
		0 & 0 & 1 & 0 & 0 & 0 & 0 & 0 \\
		0 & 0 & 0 & 1 & 0 & 0 & 0 & 0 \\
		0 & 0 & 0 & 0 & 1 & 0 & 0 & 0 \\
		0 & 0 & 0 & 0 & 0 & 1 & 0 & 0 \\
		0 & 0 & 0 & 0 & 0 & 0 & 1 & 0 \\
		1 & 0 & 0 & 0 & 0 & 0 & 0 & 1
	\end{pmatrix}.
	\]
	The remainder satisfies $0 \le R_{ab}(T) \le \exp(-\Delta/T)$, and consequently $\|R(T)\|_\infty \le 8\exp(-\Delta/T)$. As a result,
	$$\lambda_{\max}(T) = e^{-H_{\min}/T}\mu(T), \qquad \mu(T) = \rho\bigl(A + R(T)\bigr).$$
	As $T \to 0^+$, $R(T) \to 0$ entrywise, and hence in every matrix norm. The spectral radius $\rho$ is a continuous function on the space of real $8 \times 8$ matrices: by \cite{Horn}, the map $M \mapsto \rho(M)$ is continuous because it equals the maximum modulus of the roots of the characteristic polynomial, whose coefficients are continuous polynomials in the entries of $M$, while the roots depend continuously on these coefficients as a multiset.

	Because the mirror-symmetry condition~(\ref{mirror_hamiltonian}) holds, the transfer matrix $\theta(T)$ is real symmetric, and Weyl's inequalities \cite{Horn} yield
	$$|\mu(T) - \rho(A)| \le \|R(T)\|_2 \le C\exp(-\Delta/T).$$
	Consequently,
	$$\ln\lambda_{\max} = -\frac{H_{\min}}{T} + \ln\mu(T),$$
	$$\frac{\partial}{\partial T}\ln\lambda_{\max} = \frac{H_{\min}}{T^2} + \frac{\mu'(T)}{\mu(T)},$$
	$$T\frac{\partial}{\partial T}\ln\lambda_{\max} = \frac{H_{\min}}{T} + \frac{T\mu'(T)}{\mu(T)}.$$

	It remains to prove that the residual term satisfies $\frac{T\mu'(T)}{\mu(T)} \to 0 \quad \text{as } T\to 0^+$.

	For any fixed $T > 0$, the transfer matrix $\theta(T)$ has all positive entries. Consequently, the rescaled matrix
	$$B(T) = A + R(T)$$
	also has all positive entries:
	$$B_{ij}(T) = \exp\!\left(-\frac{H_{ij} - H_{\min}}{T}\right) > 0 \quad \text{for all } i,j.$$
	(At positions where $H_{ij} = H_{\min}$, the entry is exactly $1$; at all other positions, it is a positive number that vanishes as $T \to 0^+$.) By the Perron--Frobenius theorem for positive matrices \cite{Horn}, the spectral radius $\mu(T) = \rho(B(T))$ is a simple eigenvalue with a strictly positive eigenvector, and every other eigenvalue $\lambda$ satisfies $|\lambda| < \mu(T)$. Hence $\mu(T)$ is simple for every $T > 0$, not only for small $T$. Since each entry of $B(T)$ is a finite sum of exponentials, $B(T)$ is an analytic matrix-valued function of $T$ on $(0,+\infty)$, and analytic perturbation theory therefore implies that $\mu(T)$ is analytic --- and in particular smooth --- on $(0, +\infty)$.

	The first-order perturbation formula
	$$\mu'(T) = v(T)^{\mathrm{T}} \frac{dB}{dT}\, v(T),$$
	in which $v(T)$ is the normalized eigenvector associated with $\mu(T)$, is therefore valid for all $T > 0$. Since $|v^{\mathrm{T}} M v| \le \|M\|_2$ for every unit vector $v$, we have
	$$|\mu'(T)| \le \left\|\frac{dB}{dT}\right\|_2.$$
	The entries of $dB/dT = dR/dT$ are bounded above by $(M_{\max}/T^2)\exp(-\Delta/T)$, where $M_{\max}$ denotes the maximum energy difference. Hence
	$$\left\|\frac{dB}{dT}\right\|_2 \le C\,T^{-2}\exp(-\Delta/T)$$
	with a constant $C$ independent of $T$, so that
	$$\left|\frac{T\mu'(T)}{\mu(T)}\right| = O\!\left(T^{-1}\exp(-\Delta/T)\right) \to 0 \quad \text{as } T \to 0^+.$$
	The continuity of the spectral radius $\rho(M)$ as a function of $M$ \cite{Horn} then guarantees $\mu(T) \to \rho(A)$.

	\subsection*{Limits and physical interpretation}

	Combining the asymptotic analysis with the explicit spectral radii, we obtain
	$$\lim_{T\to 0^+} S(k,T) =
	\begin{cases}
		\tfrac{1}{3}\ln 2 \approx 0.231\,049\dots, & k \ge 0, \\
		0, & k < 0.
	\end{cases}$$

	Physically, the nonzero residual entropy at $k \ge 0$ originates from an exponential ground-state degeneracy $\sim 2^L$: each layer can independently take either of two ferromagnetic configurations without an energy penalty (the flat interface between adjacent layers carries zero energy --- a hallmark of gonihedric models). In the quasi-one-dimensional geometry (fixed width $3$ and length $L \to \infty$), this yields a finite entropy density of $(\ln 2)/3$. For $k < 0$, the minimum-energy transitions enforce a rigid deterministic alternation $a \to -a$ within one of three possible period-2 cycles; the number of periodic ground state sequences is therefore finite or grows at most polynomially, giving zero entropy per spin.

	These limits accurately reproduce the nonzero low-temperature entropy seen in the plots of the planar gonihedric model (Fig.~\ref{figure11}) for $k \ge 0$ and its vanishing for $k < 0$.
	
	In this case, using Eq. (\ref{26}), the entropy in the limit $T \to 0+$ is given by
	\begin{equation}
		\lim_{T \to 0+} S(k, T) = \frac{1}{3} \ln \rho\big(A(k)\big),
		\label{entropy7}
	\end{equation}
	where 
	\begin{equation}
		\rho\big(A(k)\big) = 
		\begin{cases}
			2, ~~ k \ge 0, \\
			1, ~~ k < 0.
		\end{cases}
		\label{entropy8}
	\end{equation}
	
	\section{\label{sec6} Width-three planar triangular model}
	
	We next consider a width-three planar triangular model with distinct nearest-neighbor couplings, three-spin interactions involving neighboring triangles, and an external field. As in the previous planar case, this model is obtained by unfolding the lattice shown in Fig.~\ref{figure1} into the plane and imposing the boundary conditions $t_3^m \equiv t_0^m$ and $t_3^{m+1} \equiv t_0^{m+1}$. For clarity, Fig.~\ref{figure9} also shows the $(m+2)$th layer.
	\begin{figure}[h]
		\includegraphics[width=7.2cm]{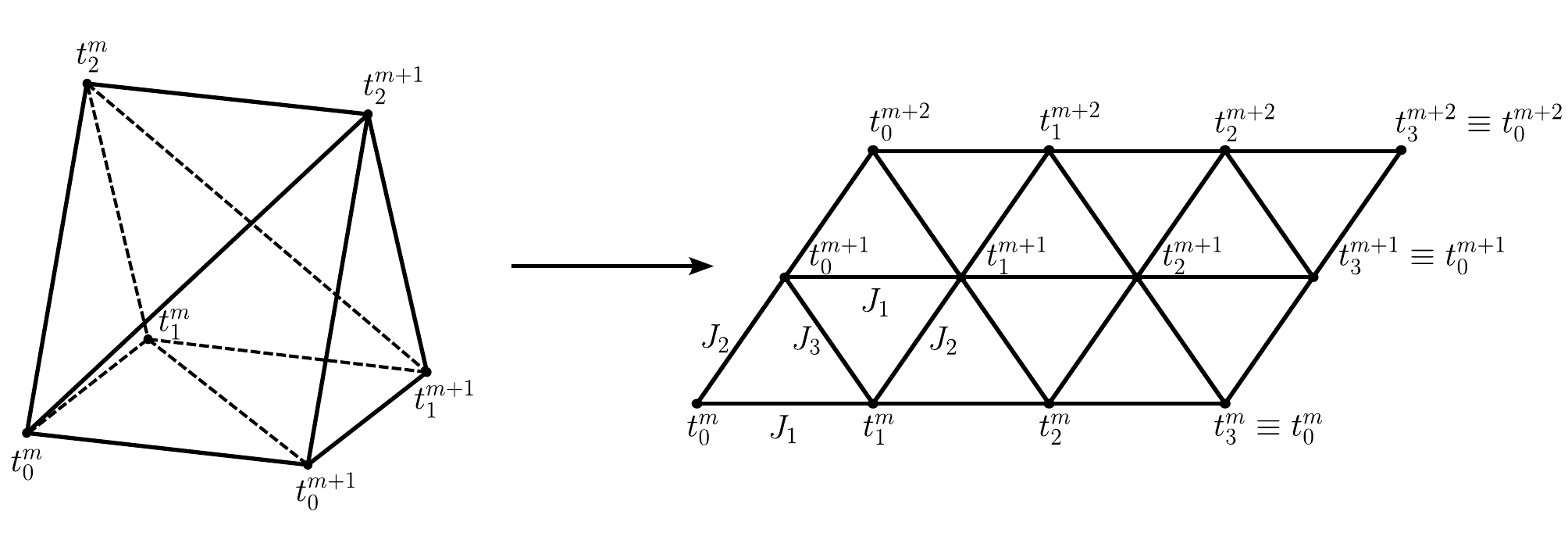}
		\caption{\label{figure9} Planar triangular Ising model obtained by unwrapping the three-chain Ising tube into the plane.}
	\end{figure}
	
	Here, $J_1^{\Delta}, J_2^{\Delta}, J_3^{\Delta}$ are the nearest-neighbor couplings within a single triangle, $J_4^{\Delta}$ and $J_5^{\Delta}$ denote the three-spin interactions between neighboring triangles (not indicated in Fig.~\ref{figure9}), and $h$ represents the external field.
	
	The Hamiltonian is given by Eq.~(\ref{5}), where the Hamiltonian of one layer of the model is of the form:
	
	\begin{widetext}
		\begin{eqnarray} \label{79}
			&&\mathcal{H}^{(m)} \nonumber \\
			&&=-\frac{J_1^{\Delta}}{2}\bigg(
			\sigma_{0}^m\sigma_{1}^m +
			\sigma_{1}^m\sigma_{2}^m +
			\sigma_{0}^m\sigma_{2}^m +
			\sigma_{0}^{m+1}\sigma_{1}^{m+1} +
			\sigma_{1}^{m+1}\sigma_{2}^{m+1} +
			\sigma_{0}^{m+1} \sigma_{2}^{m+1}
			\bigg) 
			-J_2^{\Delta} \bigg(
			\sigma_{0}^m\sigma_{0}^{m+1} +
			\sigma_{1}^m\sigma_{1}^{m+1} +
			\sigma_{2}^m\sigma_{2}^{m+1} 
			\bigg) \nonumber \\
			&&-J_3^{\Delta} \bigg(
			\sigma_{0}^m\sigma_{2}^{m+1} +
			\sigma_{2}^m\sigma_{1}^{m+1} +
			\sigma_{1}^m\sigma_{0}^{m+1}
			\bigg) 
			-J_4^{\Delta} \bigg(
			\sigma_{0}^m\sigma_{1}^m\sigma_{0}^{m+1} +
			\sigma_{1}^m\sigma_{2}^m\sigma_{1}^{m+1} +
			\sigma_{2}^m\sigma_{0}^m\sigma_{2}^{m+1} 
			\bigg) \nonumber \\
			&&-J_5^{\Delta} \bigg(
			\sigma_{1}^m\sigma_{0}^{m+1}\sigma_{1}^{m+1} +
			\sigma_{2}^m\sigma_{1}^{m+1}\sigma_{2}^{m+1} +
			\sigma_{0}^m\sigma_{2}^{m+1}\sigma_{0}^{m+1} 
			\bigg) 
			-\frac{h}{2}\bigg(
			\sigma_{0}^m + \sigma_{1}^m + \sigma_{2}^m +
			\sigma_{0}^{m+1} + \sigma_{1}^{m+1} + \sigma_{2}^{m+1}
			\bigg),
		\end{eqnarray}
	\end{widetext}
	where the coupling constants correspond to the following parameters of the TCGIT model, as given in Eqs.~(\ref{H2}), and (\ref{H3}):
	\[
	\begin{gathered}
		J_1^{\Delta} = J_{\{t_0^m,t_1^m\}}, \\
		J_2^{\Delta} = J_{\{t_0^m,t_0^{m+1}\}}, \\
		J_3^{\Delta} = J_{\{t_0^m,t_2^{m+1}\}}, \\
		J_4^{\Delta} = J_{\{t_0^m,t_1^m, t_0^{m+1}\}}, \\
		J_5^{\Delta} = J_{\{t_0^m,t_0^{m+1}, t_2^{m+1}\}}.
	\end{gathered}
	\]
	
	The transfer matrix of this model contains 22 distinct elements. Its eigenvalues are again given by Eqs.~(\ref{32})--(\ref{35}).
	
	\textbf{Proposition 3.}  The conclusions of Theorems \ref{theorem1} and \ref{theorem2} remain valid for the planar triangular model with Hamiltonian (\ref{79}).
	
	As in the two preceding cases, Fig.~\ref{figure12} shows the free energy, internal energy, specific heat, magnetization, magnetic susceptibility, and entropy in the thermodynamic limit for parameter values with $J_q^{\Delta} < 0$, except for $J_3^{\Delta}$, which may take positive values.

	\begin{figure*}[!htbp]
		\centering
		\begin{minipage}{0.24\linewidth}
			\includegraphics[width=\linewidth]{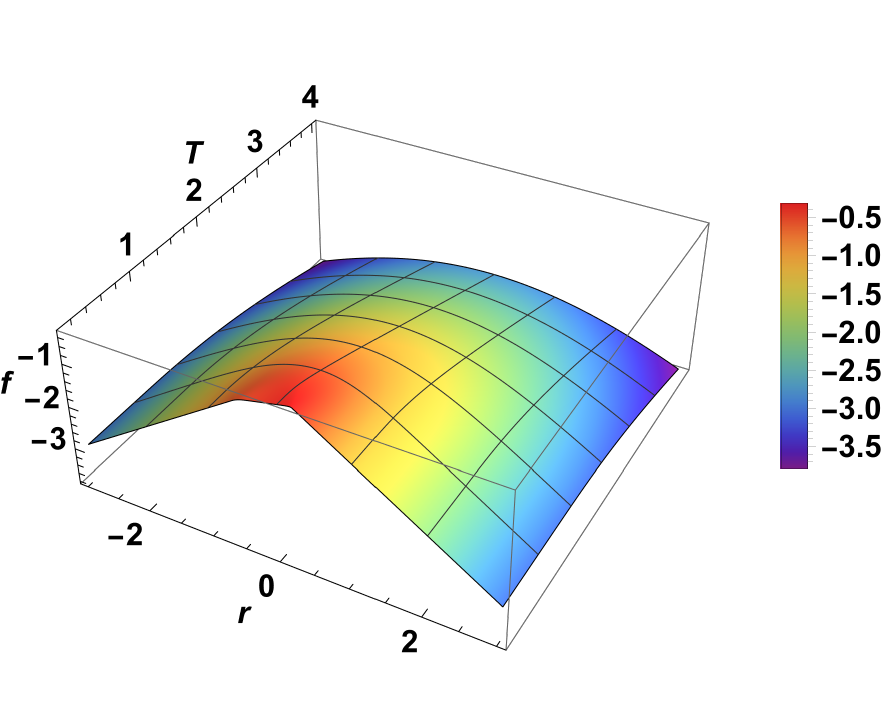} \\ (a)
		\end{minipage}
		\hfill
		\begin{minipage}{0.24\linewidth}
			\includegraphics[width=\linewidth]{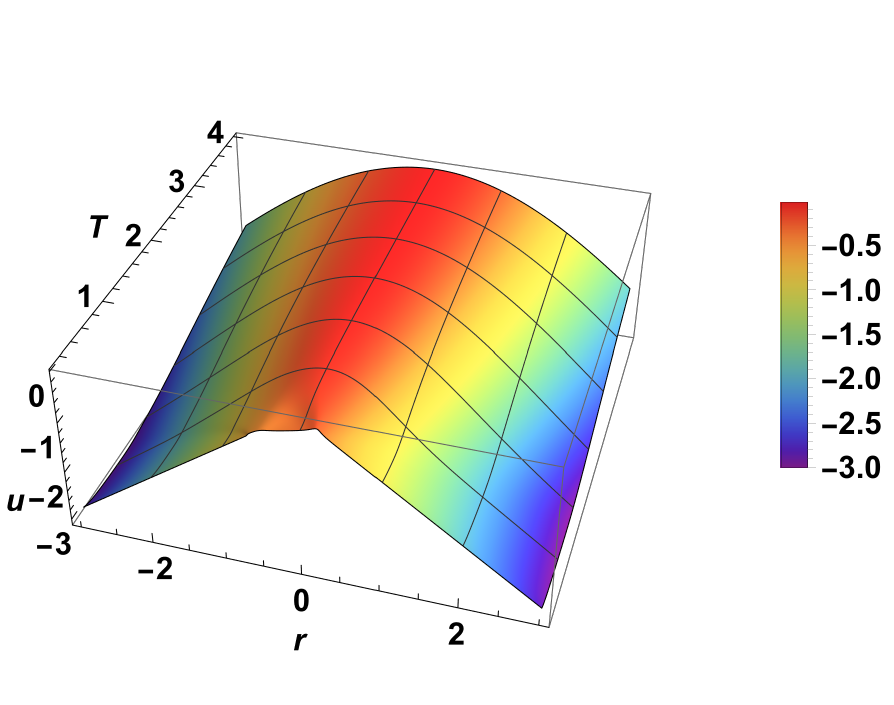} \\ (b)
		\end{minipage}
		\hfill
		\begin{minipage}{0.24\linewidth}
			\includegraphics[width=\linewidth]{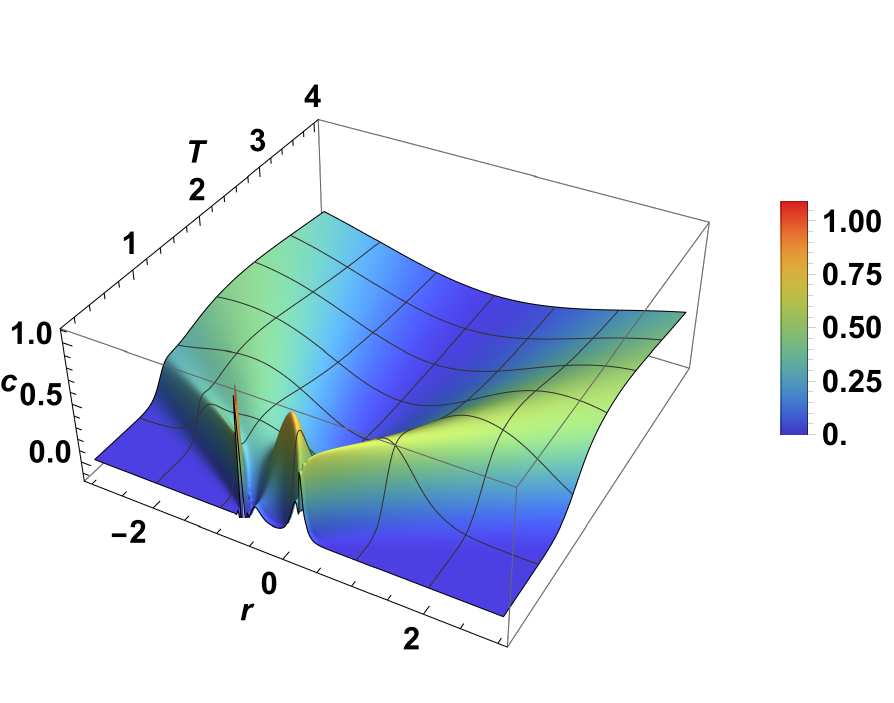} \\ (c)
		\end{minipage}
		\hfill
		\begin{minipage}{0.24\linewidth}
			\includegraphics[width=\linewidth]{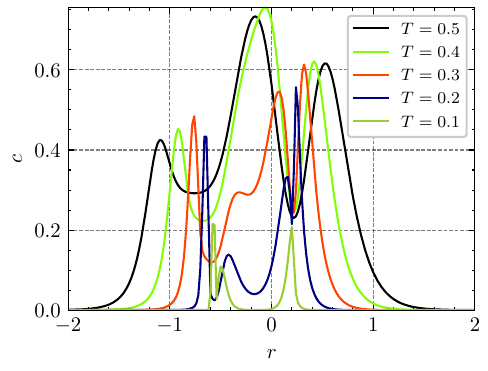} \\ (d)
		\end{minipage}
		\\[1ex]
		\begin{minipage}{0.32\linewidth}
			\includegraphics[width=\linewidth]{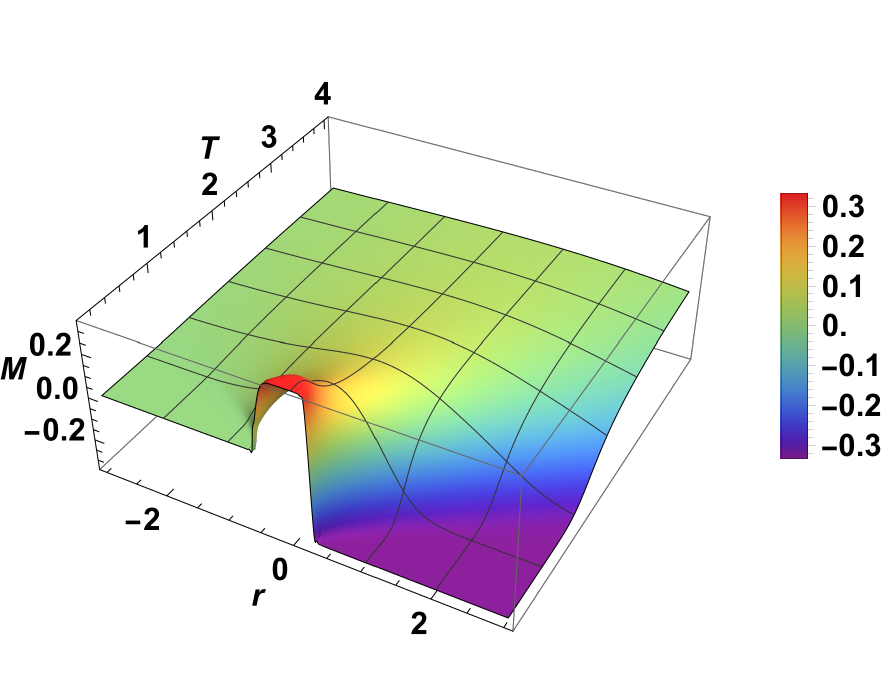}  \\ (e)
		\end{minipage}
		\hfill
		\begin{minipage}{0.32\linewidth}
			\includegraphics[width=\linewidth]{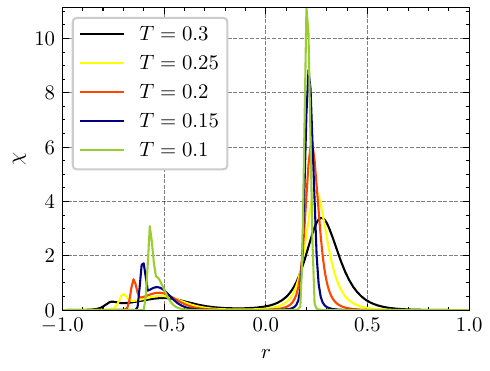} \\ (f)
		\end{minipage}
		\hfill
		\begin{minipage}{0.32\linewidth}
			\includegraphics[width=\linewidth]{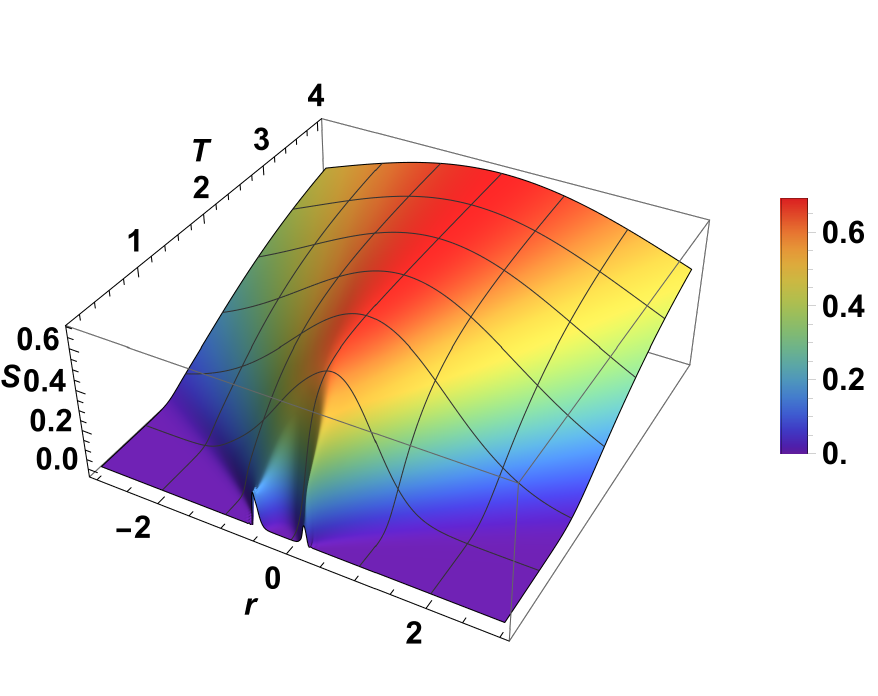} \\ (g)
		\end{minipage}
		\caption{\label{figure12} Planar triangular model. Plots of the free energy (a), internal energy (b), specific heat (c), selected cross sections of the specific heat (d), magnetization (e), selected cross sections of the susceptibility (f), and entropy (g) for the triangular model with $J_1^{\Delta} = -0.1, J_2^{\Delta} = -0.11, J_4^{\Delta} = -0.12, J_5^{\Delta} =- 0.13, h = 0.15$, $J_3^{\Delta} = r, ~T\in[0.1, 4]$.}
	\end{figure*}

	\section{\label{sec7} Theorem Proofs}
	\subsection{Proofs of Theorems \ref{theorem1} and \ref{theorem2}}
	
	To prove Theorems~\ref{theorem1} and~\ref{theorem2}, we determine the full spectrum of the transfer matrix $\theta$ in Eq.~(\ref{8}).

	A permutation matrix $D_1$ implementing a rotation by $\frac{2\pi}{3}$ commutes with $\theta$, $D_1: \big\{ \sigma_0^m \rightarrow \sigma_1^{m}, \sigma_1^m \rightarrow \sigma_2^{m}, \sigma_2^m \rightarrow \sigma_0^{m}, m = 0, \ldots, L-1\big\}$:
	\begin{equation} \label{80}
		D_1 = 
		\begin{pmatrix}
			1&0&0&0&0&0&0&0\\
			0&0&1&0&0&0&0&0\\
			0&0&0&0&1&0&0&0\\
			0&0&0&0&0&0&1&0\\
			0&1&0&0&0&0&0&0\\
			0&0&0&1&0&0&0&0\\
			0&0&0&0&0&1&0&0\\
			0&0&0&0&0&0&0&1
		\end{pmatrix},
	\end{equation}
	its eigenvalues and corresponding eigenspace basis vectors are listed in Table~\ref{table2}.
	
	\begin{table}[b]
		\caption{\label{table2}
			Eigenvalues and eigenspaces of the matrix $D_1$
		}
		\begin{ruledtabular}
			\begin{tabular}{c|c}
				Eigenvalue & The basis vectors of the eigenspace \\
				\colrule
				1 & $
				\begin{gathered}
					(1 ~~0 ~~0 ~~0 ~~0 ~~0 ~~0 ~~0)^{\mathrm{T}} \\
					(0 ~~1 ~~1 ~~0 ~~1 ~~0 ~~0 ~~0)^{\mathrm{T}} \\
					(0 ~~0 ~~0 ~~1 ~~0 ~~1 ~~1 ~~0)^{\mathrm{T}} \\
					(0 ~~0 ~~0 ~~0 ~~0 ~~0 ~~0 ~~1)^{\mathrm{T}}
				\end{gathered}
				$ \\
				\colrule
				$e^{\frac{2\pi}{3} i }$ &
				$
				\begin{gathered}
					(0 ~~0 ~~0 ~~e^{-\frac{2\pi}{3}i} ~~0 ~~e^{\frac{2\pi}{3}i} ~~1 ~~0)^{\mathrm{T}} \\
					(0 ~~e^{\frac{2\pi}{3}i} ~~e^{-\frac{2\pi}{3}i} ~~0 ~~1 ~~0 ~~0 ~~0)^{\mathrm{T}}
				\end{gathered}
				$ \\
				\colrule
				$e^{-\frac{2\pi}{3} i }$ &
				$
				\begin{gathered}
					(0 ~~0 ~~0 ~~e^{\frac{2\pi}{3}i} ~~0 ~~e^{-\frac{2\pi}{3}i} ~~1 ~~0)^{\mathrm{T}} \\
					(0 ~~e^{-\frac{2\pi}{3}i} ~~e^{\frac{2\pi}{3}i} ~~0 ~~1 ~~0 ~~0 ~~0)^{\mathrm{T}}
				\end{gathered}
				$
			\end{tabular}
		\end{ruledtabular}
	\end{table}
	
	Because commuting matrices preserve each other's eigenspaces \cite{Horn}, and any linear combination of vectors within a given eigenspace is again an eigenvector of that eigenspace, the eigenvectors of the transfer matrix $\theta$ fall into three symmetry classes:
	\begin{equation} \label{81}
		\begin{gathered}
			\vec{y}_1 = (x_1 ~~x_2 ~~x_2 ~~x_3 ~~x_2 ~~x_3 ~~x_3 ~~x_4)^{\mathrm{T}},\\
			\vec{y}_2 = (0 ~~e^{\frac{2\pi}{3}i}x_1 ~~e^{-\frac{2\pi}{3}i}x_1 ~~e^{-\frac{2\pi}{3}i}x_2 ~~x_1 ~~e^{\frac{2\pi}{3}i}x_2 ~~x_2 ~~0)^{\mathrm{T}},\\
			\vec{y}_3 = (0 ~~e^{-\frac{2\pi}{3}i}x_1 ~~e^{\frac{2\pi}{3}i}x_1 ~~e^{\frac{2\pi}{3}i}x_2 ~~x_1 ~~e^{-\frac{2\pi}{3}i}x_2 ~~x_2 ~~0)^{\mathrm{T}}.
		\end{gathered}
	\end{equation}
	
	Thus, the eight eigenvectors of $\theta$ fall into these three symmetry classes, namely: four eigenvectors of the form $\overrightarrow{y_1}$, two of the form $\overrightarrow{y_2}$, and two eigenvectors of the form $\overrightarrow{y_3}$.
	
	From the structure of the eigenvectors (\ref{81}) and the structure of the transfer matrix $\theta$ (\ref{8}), it follows that four eigenvalues of $\theta$ coincide with those of $\tau^1$ (\ref{16}), two with those of $\tau^2$ (\ref{17}), and two more with those of $\tau^3$ (\ref{18}).
	
	The characteristic polynomial of $\tau^1$~(\ref{16}) is the quartic~(\ref{20}). In general, we solve it by Ferrari's method; four of the eight eigenvalues of the transfer matrix $\theta$ have the explicit expressions~(\ref{32}) and~(\ref{33}).
	
	The characteristic equations of the matrices $\tau^2$ (\ref{17}), $\tau^3$ (\ref{18}) are quadratic, respectively:
	\[
	\begin{gathered}
		\lambda^2 + b_2\lambda + c_2 = 0,\\
		\lambda^2 + b_3\lambda + c_3 = 0,
	\end{gathered}
	\]
	whose coefficients $b_2, c_2, b_3$, and $c_3$ are given by Eq.~(\ref{36}); therefore, the corresponding eigenvalues are given by Eqs.~(\ref{34}), (\ref{35}).
	
	Since $\theta$ is a real matrix with strictly positive entries, the Perron--Frobenius theorem implies that it has a unique largest eigenvalue $\lambda_{\max}>0$, with an eigenvector whose components are strictly positive. Such an eigenvector cannot belong to the subspaces with complex eigenvalues and must therefore lie in the $\tau^1$ block, i.e., in the symmetry class $\vec{y}_1$. Hence $\lambda_{\max}(h,T)$ satisfies Eq.~(\ref{20}), and comparison with Eqs.~(\ref{32})--(\ref{35}) shows that $\lambda_{\max}$ is determined by Eq.~(\ref{27}).

	\subsection{Proofs of Theorems \ref{theorem3}, \ref{theorem4}}
	
	In this particular case, in addition to commuting with matrix $D_1$ (\ref{80}), the transfer matrix $\theta$ commutes with $D_2$ because of its central symmetry:
	\begin{equation} \label{82} 
		D_2 = 
		\begin{pmatrix}
			0&0&0&0&0&0&0&1\\
			0&0&0&0&0&0&1&0\\
			0&0&0&0&0&1&0&0\\
			0&0&0&0&1&0&0&0\\
			0&0&0&1&0&0&0&0\\
			0&0&1&0&0&0&0&0\\
			0&1&0&0&0&0&0&0\\
			1&0&0&0&0&0&0&0
		\end{pmatrix},
	\end{equation}
	the eigenvalues of $D_2$ and bases of the corresponding eigenspaces are listed in Table \ref{table3}.
	
	\begin{table}[b]
		\caption{\label{table3}
			Eigenvalues and eigenspaces of matrix $D_2$
		}
		\begin{ruledtabular}
			\begin{tabular}{c|c}
				Eigenvalue & The basis vectors of the eigenspace  \\
				\colrule
				1 & $
				\begin{gathered}
					(1 ~~0 ~~0 ~~0 ~~0 ~~0 ~~0 ~~1)^{\mathrm{T}} \\
					(0 ~~1 ~~0 ~~0 ~~0 ~~0 ~~1 ~~0)^{\mathrm{T}}\\
					(0 ~~0 ~~1 ~~0 ~~0 ~~1 ~~0 ~~0)^{\mathrm{T}} \\
					(0 ~~0 ~~0 ~~1 ~~1 ~~0 ~~0 ~~0)^{\mathrm{T}}
				\end{gathered}
				$ \\
				\colrule
				-1 &
				$
				\begin{gathered}
					(-1 ~~0 ~~0 ~~0 ~~0 ~~0 ~~0 ~~1)^{\mathrm{T}} \\
					(0 ~~-1 ~~0 ~~0 ~~0 ~~0 ~~1 ~~0)^{\mathrm{T}} \\
					(0 ~~0 ~~-1 ~~0 ~~0 ~~1 ~~0 ~~0)^{\mathrm{T}} \\
					(0 ~~0 ~~0 ~~-1 ~~1 ~~0 ~~0 ~~0)^{\mathrm{T}}
				\end{gathered}
				$
			\end{tabular}
		\end{ruledtabular}
	\end{table}
	
	Using the invariance of eigenspaces under commuting matrices \cite{Horn}, and assuming that all eigenvectors of matrix $D_2$ are either centrally symmetric or centrally antisymmetric, and the transfer matrix is centrally symmetric for this case, let us write down the first two possible kinds of eigenvectors of the transfer matrix, which will be either centrally symmetric or centrally antisymmetric \cite{Andrew}. Accordingly, the first two possible forms of the eigenvectors are
	\begin{equation} \label{83}
		\begin{gathered}
			\vec{y}_4 = (x_1 ~~x_2 ~~x_2 ~~x_2 ~~x_2 ~~x_2 ~~x_2 ~~x_1)^{\mathrm{T}}, \\
			\vec{y}_5 = (x_1 ~~x_2 ~~x_2 ~~-x_2 ~~x_2 ~~-x_2 ~~-x_2 ~~-x_1)^{\mathrm{T}}.
		\end{gathered}
	\end{equation}
	
	Combining this constraint with the commutation relation between $\theta$ and $D_1$, we obtain, up to an overall scalar factor, the following additional eigenvector forms constructed from the eigenspaces listed in Table \ref{table2}:
	\begin{equation} \label{84}
		\begin{gathered}
			\vec{y}_6 = (0 ~~e^{\frac{2\pi}{3}i}x ~~e^{-\frac{2\pi}{3}i}x ~~x ~~x ~~e^{-\frac{2\pi}{3}i}x ~~e^{\frac{2\pi}{3}i}x ~~0)^{\mathrm{T}}, \\
			\vec{y}_7 = (0 ~-e^{\frac{2\pi}{3}i}x ~-e^{-\frac{2\pi}{3}i}x ~x ~-x ~e^{-\frac{2\pi}{3}i}x ~e^{\frac{2\pi}{3}i}x ~0)^{\mathrm{T}},\\
			\vec{y}_8 = (0 ~~x ~~e^{-\frac{2\pi}{3}i}x ~~e^{\frac{2\pi}{3}i}x ~~e^{\frac{2\pi}{3}i}x ~~e^{-\frac{2\pi}{3}i}x ~~x ~~0)^{\mathrm{T}}, \\
			\vec{y}_9 =  (0 ~-x ~-e^{-\frac{2\pi}{3}i}x ~~e^{\frac{2\pi}{3}i}x ~~-e^{\frac{2\pi}{3}i}x ~~e^{-\frac{2\pi}{3}i}x ~~x ~~0)^{\mathrm{T}}.
		\end{gathered}
	\end{equation}
	
	As a result, all eight eigenvectors of $\theta$ are represented by these types of vectors: two eigenvectors of the form $\vec{y}_4$, two of the form $\vec{y}_5$, and one each of the forms $\vec{y}_6,\ldots,\vec{y}_9$.
	
	From the structure of the eigenvectors~(\ref{83}) and~(\ref{84}), it follows that two eigenvalues of $\theta$ coincide with those of $\tau^4$~(\ref{37}), two with those of $\tau^5$~(\ref{38}), and four with the expressions~(\ref{42})--(\ref{45}).
	
	Similar to the previous proof, the expressions for the eigenvalues of the matrices $\tau^4, \tau^5$ of size $2 \times 2$ coincide with the expressions (\ref{40}), (\ref{41}), and the largest eigenvalue of the matrix $\theta$ is $\lambda_1$, the expression for which is given by Eq.~(\ref{39}).
	
	\section{Conclusion}\label{conclusion}
	
	In this work we have obtained an exact transfer-matrix solution for the generalized three-chain Ising tube (TCGIT) with toroidal boundary conditions and a $C_3$-invariant Hamiltonian defined on an elementary prism. For a finite system, the partition function can be written in the spectral form $Z_L=\sum_{k=1}^8 \lambda_k^L$, whereas in the thermodynamic limit all principal thermodynamic quantities are expressed in terms of the Perron eigenvalue $\lambda_{\max}$; in particular, $f=-(T/3)\ln\lambda_{\max}$. This yields exact expressions for the free energy, internal energy, specific heat, magnetization, magnetic susceptibility, and entropy for the most general class of interactions compatible with the rotational symmetry of the elementary prism. 
	
	For the subclass with even-spin interactions (PSC), the spectral problem simplifies substantially: the characteristic polynomial factorizes, and $\lambda_{\max}$ is determined by the root of a quadratic equation, which leads to especially compact closed-form thermodynamic expressions. For mirror-symmetric subfamilies, we derive explicit formulas for the pair correlation functions and express the magnetization in terms of the components of the eigenvector associated with $\lambda_{\max}$; in the even-spin case with $h=0$, the magnetization vanishes. Special parameter reductions of the general tube also yield exact results for the width-three planar model with nearest-neighbor, next-nearest-neighbor, and plaquette interactions and for the width-three planar triangular model with distinct nearest-neighbor couplings, three-spin interactions involving neighboring triangles, and an external field. The gonihedric limit is particularly illustrative: $S(T\to0^+)=(\ln 2)/3$ for $k\ge0$ and $S(T\to0^+)=0$ for $k<0$.
	
	\appendix
	\section{Determination of the root $y^{(0)}$ of the auxiliary cubic in Ferrari’s method} \label{appendixA}
	
	To factorize the quartic polynomial (\ref{20})
	\[
		Q(\lambda)=\lambda^4+a_3\lambda^3+a_2\lambda^2+a_1\lambda+a_0
	\]
	into two quadratic factors, we write
	\[
		Q(\lambda)=(\lambda^2+b_0\lambda+c_0)(\lambda^2+b_1\lambda+c_1)
	\]
	and compare coefficients:
	\[
		\begin{gathered}
			b_0+b_1=a_3, \\
			b_0b_1+c_0+c_1=a_2,\\
			b_0c_1+b_1c_0=a_1,\\
			c_0c_1=a_0.
		\end{gathered}
	\]
	Setting
	\[
		y=c_0+c_1,
	\]
	we obtain
	\[
		b_0b_1=a_2-y,\qquad c_0c_1=a_0.
	\]
	Hence
	\[
		\begin{gathered}
			b_{0,1}=\frac{a_3}{2}\mp\sqrt{\frac{a_3^2}{4}-a_2+y}, \\
			c_{0,1}=\frac{y}{2}\mp\operatorname{sign}\left(\frac{a_3}{2} y-a_1\right)\sqrt{\frac{y^2}{4}-a_0},
		\end{gathered}
	\]
	up to the choice of the relative sign. Substituting these expressions into
	\[
		b_0c_1+b_1c_0=a_1
	\]
	and squaring yields the resolvent cubic
	\begin{equation}
		y^3+Ay^2+By+C=0,
		\label{A1}
	\end{equation}
	where
	\begin{equation}
		\begin{gathered}
			A=-a_2,\\
			B=a_1a_3-4a_0,\\
			C=-a_0a_3^2+4a_0a_2-a_1^2.
		\end{gathered}
		\label{A2}
	\end{equation}
	
	To solve Eq.~(\ref{A1}), we reduce it to a depressed cubic by
	\begin{equation}
		y=x-\frac{A}{3}.
		\label{A3}
	\end{equation}
	Then Eq.~(\ref{A1}) becomes
	\begin{equation}
		\begin{gathered}
			x^3+px+q=0, \\
			p=B-\frac{A^2}{3}, \\
			q=\frac{2A^3}{27}-\frac{AB}{3}+C.
		\end{gathered}
		\label{A4}
	\end{equation}
	Introduce the Cardano discriminant
	\begin{equation}
		\Delta_C=\left(\frac q2\right)^2+\left(\frac p3\right)^3.
		\label{A5}
	\end{equation}
	
	\paragraph*{\textbf{Case \(\Delta_C>0\).}}
	If \(\Delta_C>0\), Eq.~(\ref{A4}) has a unique real root. Cardano's formula gives
	\begin{equation}
		x_0= \sqrt[3]{-\frac q2+\sqrt{\Delta_C}} + \sqrt[3]{-\frac q2-\sqrt{\Delta_C}},
		\label{A6}
	\end{equation}
	where the cube roots are the real cube roots. The required root of Eq.~(\ref{A1}) is
	\begin{equation}
		y^{(0)}=x_0-\frac A3.
		\label{A7}
	\end{equation}
	
	\paragraph*{\textbf{Case \(\Delta_C<0\).}}
	If \(\Delta_C<0\), all three roots of Eq.~(\ref{A1}) are real. Define
	\begin{equation}
		\begin{gathered}
			R=\sqrt{-\frac p3}=\frac13\sqrt{A^2-3B}, \\
			\eta=\frac{q}{2R^3} = \frac{2A^3-9AB+27C}{2(A^2-3B)^{3/2}}, \\
			|\eta|\le 1.
		\end{gathered}
		\label{A8}
	\end{equation}
	Setting \(x=2R\cos\phi\) in Eq.~(\ref{A4}), we obtain
	\[
		2R^3\cos(3\phi)+q=0, \qquad\text{hence}\qquad \cos(3\phi)=-\eta.
	\]
	Therefore the three real roots are
	\begin{equation}
		y^{(k)}= -\frac A3 + 2R\cos\!\left(\frac13\arccos(-\eta)-\frac{2\pi k}{3}\right), \qquad k=0,1,2.
		\label{A9}
	\end{equation}
	We choose
	\begin{equation}
		y^{(0)}= -\frac A3 + 2R\cos\!\left(\frac13\arccos(-\eta)\right),
		\label{A10}
	\end{equation}
	that is, the \(k=0\) branch of Eq.~(\ref{A9}). This is the largest real root of the resolvent cubic. This choice of root ensures that the quadratic factor \(\lambda^2 + b_0\lambda + c_0\) (with the minus sign in front of the square root for \(b_0\) and the corresponding sign adjustment for \(c_0\)) contains the two largest roots of the original quartic, so that
	\[
		\lambda_{\max} = \frac{-b_0 + \sqrt{b_0^2 - 4c_0}}{2}
	\]
	coincides with the Perron--Frobenius eigenvalue of the transfer matrix.
	
	\paragraph*{\textbf{Case \(\Delta_C=0\).}}
	If \(\Delta_C=0\), all roots of the depressed cubic are real, with at least two equal. The trivial subcase \(p=q=0\) yields the triple root \(x=0\). Otherwise \(p<0\); let \( r=\sqrt{-p/3}>0 \). Then the roots are as follows. If  \( q=2r^3 \), the roots are \( x=-2r \) (simple) and \( x=r \) (double). If \( q=-2r^3 \), the roots are \( x=2r \) (simple) and \( x=-r \) (double).
	
	Equivalently, the trigonometric form (\ref{A9}) with \( |\eta|=1 \) degenerates to the same roots. The selection rule (\ref{A10})  continues to return the largest real root \( y^{(0)} \), guaranteeing that the quadratic factor associated with \( b_0,c_0 \) isolates the Perron--Frobenius eigenvalue \(\lambda_{\max}\) of the transfer matrix (exactly as in the generic cases).
	
	\paragraph*{\textbf{Justification: why this choice yields \(\lambda_{\max}=\lambda_1\) when the quartic has four real roots.}}
	Assume that Eq.~(\ref{20}) has four real roots
	\[
		\lambda_1 > \lambda_2\ge \lambda_3\ge \lambda_4,
	\]
	with \(\lambda_1 > \lambda_2\) guaranteed by the Perron--Frobenius theorem (the transfer matrix \(\theta\) has strictly positive entries, hence its largest eigenvalue is simple and its eigenvector has strictly positive components). The argument extends verbatim to the boundary \(\Delta_C=0\) (repeated roots of the resolvent cubic correspond to repeated eigenvalues of the quartic, but Perron--Frobenius still forces \(\lambda_1>\lambda_2\) for generic parameters).
	
	The three roots of the resolvent cubic are then
	\[
		\begin{gathered}
			y^{(0)}=\lambda_1\lambda_2+\lambda_3\lambda_4,\\
			y^{(1)}=\lambda_1\lambda_3+\lambda_2\lambda_4,\\
			y^{(2)}=\lambda_1\lambda_4+\lambda_2\lambda_3.
		\end{gathered}
	\]
	Moreover,
	\[
		\begin{gathered}
			y^{(0)}-y^{(1)}=(\lambda_1-\lambda_4)(\lambda_2-\lambda_3)\ge 0,\\
			y^{(1)}-y^{(2)}=(\lambda_1-\lambda_2)(\lambda_3-\lambda_4)\ge 0,
		\end{gathered}
	\]
	so that \(y^{(0)}\ge y^{(1)}\ge y^{(2)}\). Hence the prescription (\ref{A10}) selects
	\[
		y_1=y^{(0)}=\lambda_1\lambda_2+\lambda_3\lambda_4.
	\]
	
	By Vieta's formulas,
	\[
		\begin{gathered}
			a_0=\lambda_1\lambda_2\lambda_3\lambda_4,\\
			a_1=-(\lambda_1\lambda_2\lambda_3+\lambda_1\lambda_2\lambda_4+\lambda_1\lambda_3\lambda_4+\lambda_2\lambda_3\lambda_4), \\
			a_2=\lambda_1\lambda_2+\lambda_1\lambda_3+\lambda_1\lambda_4+\lambda_2\lambda_3+\lambda_2\lambda_4+\lambda_3\lambda_4,\\
			a_3=-(\lambda_1+\lambda_2+\lambda_3+\lambda_4).
		\end{gathered}
	\]
	
	And then
	\[
		b_0=-(\lambda_1+\lambda_2),\qquad c_0=\lambda_1\lambda_2
	\]
	(and, correspondingly, \(b_1=-(\lambda_3+\lambda_4)\), \(c_1=\lambda_3\lambda_4\)). Consequently,
	\[
		\begin{gathered}
				\lambda^2+b_0\lambda+c_0=(\lambda-\lambda_1)(\lambda-\lambda_2),\\
			\lambda^2+b_1\lambda+c_1=(\lambda-\lambda_3)(\lambda-\lambda_4),
		\end{gathered}
	\]
	and therefore
	\[
	\lambda_{\max}=\frac{-b_0+\sqrt{b_0^2-4c_0}}{2}=\lambda_1,
	\]
	as required (the square root extracts \(\lambda_1-\lambda_2>0\)).
	
	This proves that, under the selection rule (\ref{A10}), the largest eigenvalue is indeed the root \(\lambda_1\) used throughout the main text (Eq.~(\ref{27})).
	
	In the mirror-symmetric case (Hamiltonian satisfying Eq.~(\ref{mirror_hamiltonian})), the transfer matrix is real symmetric, all its eigenvalues are real, and the quartic (\ref{20}) likewise has four real roots. In this situation the discriminant satisfies \(\Delta_C\le 0\), the trigonometric solution (\ref{A9}) applies (including its degenerate form when \(\Delta_C=0\)), and the selection rule (\ref{A10}) continues to isolate the quadratic factor associated with the two largest eigenvalues.
	
	For $\Delta_{C} > 0$, both cases may occur depending on the spin-spin interaction parameters: either $\lambda_{\max} = \frac{-b_0 + \sqrt{b_0^2 - 4c_0}}{2}$ or $\lambda_{\max} = \frac{-b_1 + \sqrt{b_1^2 - 4c_1}}{2}$. Examples are given below.
	
	\subsection*{Example: the matrix $\tau_1$ has four real roots $\lambda_1, \lambda_2, \lambda_3, \lambda_4$, with $\lambda_{\max}=\max\{\lambda_1,\lambda_2,\lambda_3,\lambda_4\}$}
	
	We give a parameter set from the planar triangular model of Sec.~\ref{sec6}. As stated earlier,
	the conclusions of Theorems~\ref{theorem1} and~\ref{theorem2} remain valid for the Hamiltonian~(\ref{79}). Consequently, examples within this subfamily are sufficient to illustrate all possible cases of the formula~(\ref{27}).
	
	Consider the parameters of the Hamiltonian~(\ref{79})
	\begin{equation}
		\begin{gathered}
			J^{\Delta}_{1}=-\tfrac{1}{2},\quad
			J^{\Delta}_{2}=\tfrac{3}{2},\quad
			J^{\Delta}_{3}=\tfrac{3}{2}, \\
			J^{\Delta}_{4}=-\tfrac{3}{2},\quad
			J^{\Delta}_{5}=\tfrac{3}{2},\quad
			h=0,
		\end{gathered}
	\end{equation}
	and temperature $T=1$.
	
	The characteristic polynomial of the matrix $\tau_{1}$ reads
	\begin{equation}
		\begin{gathered}
			x^{4} - 3748.71\,x^{3}
			+ 2.86619\times 10^{6}\,x^{2} \\
			+ 2.74389\times 10^{9}\,x
			- 2.87108\times 10^{12}.
		\end{gathered}
	\end{equation}
	
	Hence
	\begin{align}
		a_{3} &= -3748.710806744144, &
		a_{2} &= \phantom{-}2.86619\times 10^{6}, \nonumber \\
		a_{1} &= \phantom{-}2.74389\times 10^{9}, &
		a_{0} &= -2.87108\times 10^{12}, \nonumber
	\end{align}
	\begin{align}
		A &= 2.86619\times 10^{6}, &
		B &= 1.19826\times 10^{12}, \nonumber \\
		C &= -9.83459\times 10^{16}, & & \nonumber \\
		p &= -1.54009\times 10^{12}, &
		q &= -6.97671\times 10^{17}, \nonumber
	\end{align}
	\begin{equation}
		\Delta_{C} = -1.36057\times 10^{34} < 0.
	\end{equation}
	
	Since $\Delta_{C}<0$, we apply formulas~(\ref{A8})--(\ref{A10}), which yield
	\begin{equation}
		\lambda_{\max} = 1877.92.
	\end{equation}
	
	\subsection*{Example: $\lambda_{\max}$ belongs to the quadratic factor with $b_{1}$, $c_{1}$}
	
	We take the same parameters of the Hamiltonian~(\ref{79}),
	\begin{equation}
		\begin{gathered}
			J^{\Delta}_{1}=-\tfrac{1}{2},\quad
			J^{\Delta}_{2}=\tfrac{3}{2},\quad
			J^{\Delta}_{3}=\tfrac{3}{2},\\
			J^{\Delta}_{4}=-\tfrac{3}{2},\quad
			J^{\Delta}_{5}=\tfrac{3}{2},\quad
			h=0,
		\end{gathered}
	\end{equation}
	and temperature $T=3$.
	
	The characteristic polynomial of the matrix $\tau_{1}$ reads
	\begin{equation}
		x^{4} - 38.0793\,x^{3}
		+ 333.717\,x^{2}
		+ 650.331\,x
		- 15037.4.
	\end{equation}
	
	Hence
	\begin{align}
		a_{3} &= -38.079266511035314, &
		a_{2} &= +333.7174694811645, \nonumber \\
		a_{1} &= +650.3313078507203, &
		a_{0} &= -15037.434810306542, \nonumber
	\end{align}
	\begin{align}
		A &= -333.717, &
		B &= 35385.6, \nonumber \\
		C &= 1.30879\times 10^{6}, & & \nonumber \\
		p &= -1736.85, &
		q &= 2.49207\times 10^{6}, \nonumber
	\end{align}
	\begin{equation}
		\Delta_{C} = 1.55241\times 10^{12} > 0.
	\end{equation}
	
	We then obtain
	\begin{align}
		x_{0} &= -139.846, &
		y_{0} &= -28.607, \nonumber \\
		b_{0} &= -19.4676, &
		c_{0} &= 109.155, \nonumber
	\end{align}
	\begin{equation}
		\lambda_{1} = 9.73378 + 3.79586\,i,
	\end{equation}
	\begin{align}
		b_{1} &= -18.6117, &
		c_{1} &= -137.762, \nonumber
	\end{align}
	\begin{equation}
		\lambda_{3} = 24.2845, \qquad
		\lambda_{\max} = \lambda_{3}.
	\end{equation}
	
	\subsection*{Example: $\lambda_{\max}$ belongs to the quadratic
		factor with $b_{0}$, $c_{0}$}
	
	We take the same parameters of the Hamiltonian~(\ref{79}),
	\begin{equation}
		\begin{gathered}
			J^{\Delta}_{1}=-\tfrac{1}{2},\quad
			J^{\Delta}_{2}=\tfrac{3}{2},\quad
			J^{\Delta}_{3}=\tfrac{3}{2}, \\
			J^{\Delta}_{4}=-\tfrac{3}{2},\quad
			J^{\Delta}_{5}=\tfrac{3}{2},\quad
			h=0,
		\end{gathered}
	\end{equation}
	and temperature $T=5$.
	
	The characteristic polynomial of the matrix $\tau_{1}$ reads
	\begin{equation}
		x^{4} - 18.2315\,x^{3}
		+ 74.153\,x^{2}
		- 51.9957\,x
		- 221.588.
	\end{equation}
	
	Hence
	\begin{align}
		a_{3} &= -18.231450289983314, &
		a_{2} &= +74.15299412420485, \nonumber \\
		a_{1} &= -51.995664831063195, &
		a_{0} &= -221.5875317726894, \nonumber
	\end{align}
	\begin{align}
		A &= -74.153, &
		B &= 1834.31, \nonumber \\
		C &= 5223.48, & & \nonumber \\
		p &= 1.41766, &
		q &= 20360.1, \nonumber
	\end{align}
	\begin{equation}
		\Delta_{C} = 1.03633\times 10^{8} > 0.
	\end{equation}
	
	We then obtain
	\begin{align}
		x_{0} &= -27.2888, &
		y_{0} &= -2.57115, \nonumber \\
		b_{0} &= -11.6401, &
		c_{0} &= -16.2268, \nonumber
	\end{align}
	\begin{equation}
		\lambda_{1} = 12.8981,
	\end{equation}
	\begin{align}
		b_{1} &= -6.59138, &
		c_{1} &= 13.6557, \nonumber
	\end{align}
	\begin{equation}
		\lambda_{3} = 3.29569 + 1.67155\,i, \qquad
		\lambda_{\max} = \lambda_{1}.
	\end{equation}
	
	\section{Additional parameters for calculating pair correlations in the PSCSH model}
	\begin{eqnarray} \label{56}
		\begin{gathered}
			A_1 = \sqrt{\frac{\sqrt{C_1} + a + d - e - 2f -2g - w}{2\sqrt{C_1}}}, \\
			B_1 = \frac{\sqrt{6}(b+c)}{\sqrt{(\sqrt{C_1} + a + d - e -2f - 2g - w)\sqrt{C_1}}},  \\
			A_2 = \sqrt{\frac{\sqrt{C_2} - a + d + e +2f - 2g - w}{2\sqrt{C_2}}},  \\
			B_2 = \frac{\sqrt{6}(b-c)}{\sqrt{(\sqrt{C_2} - a + d + e + 2f - 2g - w)\sqrt{C_2}}}
		\end{gathered}
	\end{eqnarray}
	\begin{eqnarray}
		\label{C2}
		&&S'_{0_{12}} \nonumber \\
		&&=\begin{pmatrix}
			A_1B_2+\frac{B_1A_2}{3}&\frac{B_1B_2}{3}-A_1A_2&-\frac{2}{3}e^{-\frac{\pi i}{3}}B_1&\frac{2}{3}B_1\\
			B_1B_2-\frac{A_1A_2}{3}&-\frac{A_1B_2}{3}-B_1A_2&\frac{2}{3}e^{-\frac{\pi i}{3}}A_1&-\frac{2}{3}A_1\\
			\frac{2}{3}e^{\frac{\pi i}{3}}A_2&\frac{2}{3}e^{\frac{\pi i}{3}}B_2&-\frac{1}{3}&-\frac{2}{3}e^{\frac{\pi i}{3}}\\
			-\frac{2}{3}A_2&-\frac{2}{3}B_2&-\frac{2}{3}e^{-\frac{\pi i}{3}}&-\frac{1}{3}
		\end{pmatrix},  \nonumber \\
	\end{eqnarray}
	\begin{eqnarray}
		\label{C3}
		&&S'_{0_{21}}  \nonumber \\
		&&= \begin{pmatrix}
			A_1B_2+\frac{B_1A_2}{3}&B_1B_2-\frac{A_1A_2}{3}&\frac{2}{3}e^{-\frac{\pi i}{3}}A_2&-\frac{2}{3}A_2\\
			\frac{B_1B_2}{3}-A_1A_2&-B_1A_2-\frac{A_1B_2}{3}&\frac{2}{3}e^{-\frac{\pi i}{3}}B_2&-\frac{2}{3}B_2\\
			-\frac{2}{3}e^{\frac{\pi i}{3}}B_1&\frac{2}{3}e^{\frac{\pi i}{3}}A_1&-\frac{1}{3}&-\frac{2}{3}e^{\frac{\pi i}{3}}\\
			\frac{2}{3}B_1&-\frac{2}{3}A_1&-\frac{2}{3}e^{-\frac{\pi i}{3}}&-\frac{1}{3} 
		\end{pmatrix},  \nonumber \\
	\end{eqnarray}
	\begin{eqnarray}
		\label{C4}
		&&S'_{1_{12}} \nonumber \\
		&&= 
		\begin{pmatrix}
			A_1B_2+\frac{B_1A_2}{3}&\frac{B_1B_2}{3}-A_1A_2&-\frac{2}{3}e^{\frac{\pi i}{3}}B_1&-\frac{2}{3}e^{\frac{\pi i}{3}}B_1\\
			B_1B_2-\frac{A_1A_2}{3}&-\frac{A_1B_2}{3}-B_1A_2&\frac{2}{3}e^{\frac{\pi i}{3}}A_1&\frac{2}{3}e^{\frac{\pi i}{3}}A_1\\
			\frac{2}{3}e^{-\frac{\pi i}{3}}A_2&\frac{2}{3}e^{-\frac{\pi i}{3}}B_2&-\frac{1}{3}&\frac{2}{3}\\
			\frac{2}{3}e^{-\frac{\pi i}{3}}A_2&\frac{2}{3}e^{-\frac{\pi i}{3}}B_2&\frac{2}{3}&-\frac{1}{3}
		\end{pmatrix},  \nonumber \\
	\end{eqnarray}
	\begin{eqnarray}
		\label{C5}
		&&S'_{1_{21}}  \nonumber \\
		&&= 
		\begin{pmatrix}
			A_1B_2+\frac{B_1A_2}{3}&B_1B_2-\frac{A_1A_2}{3}&\frac{2}{3}e^{\frac{\pi i}{3}}A_2&\frac{2}{3}e^{\frac{\pi i}{3}}A_2\\
			\frac{B_1B_2}{3}-A_1A_2&-B_1A_2-\frac{A_1B_2}{3}&\frac{2}{3}e^{\frac{\pi i}{3}}B_2&\frac{2}{3}e^{\frac{\pi i}{3}}B_2\\
			-\frac{2}{3}e^{-\frac{\pi i}{3}}B_1&\frac{2}{3}e^{-\frac{\pi i}{3}}A_1&-\frac{1}{3}&\frac{2}{3}\\
			-\frac{2}{3}e^{-\frac{\pi i}{3}}B_1&\frac{2}{3}e^{-\frac{\pi i}{3}}A_1&\frac{2}{3}&-\frac{1}{3}
		\end{pmatrix}.  \nonumber \\
	\end{eqnarray}
	
	\section{\label{subsecD:GroundStates}Examples of ground states and the correlation length in the PSCSH model}

	From the structure of the transfer matrix~(\ref{47}), for fixed couplings $J_r$, the single-prism Hamiltonian~(\ref{46}) can take eight distinct values, depending on the spin configuration. As an example, we fix $J_2 = J_4 = J_6 = J_7 = J = 1$ and vary $J_1/J$, $J_3/J$, and $J_5/J$ from $-15$ to $15$. The ground states obtained by minimizing the Hamiltonian ~(\ref{46}) over this region of parameter space are shown in Fig.~\ref{figure2}, together with cross sections at $J_5/J = 15,~0,~-5,~-15$ (Fig.~\ref{figure3}).

	\begin{figure}[!htbp]
		\centering
		\includegraphics[width=1\linewidth]{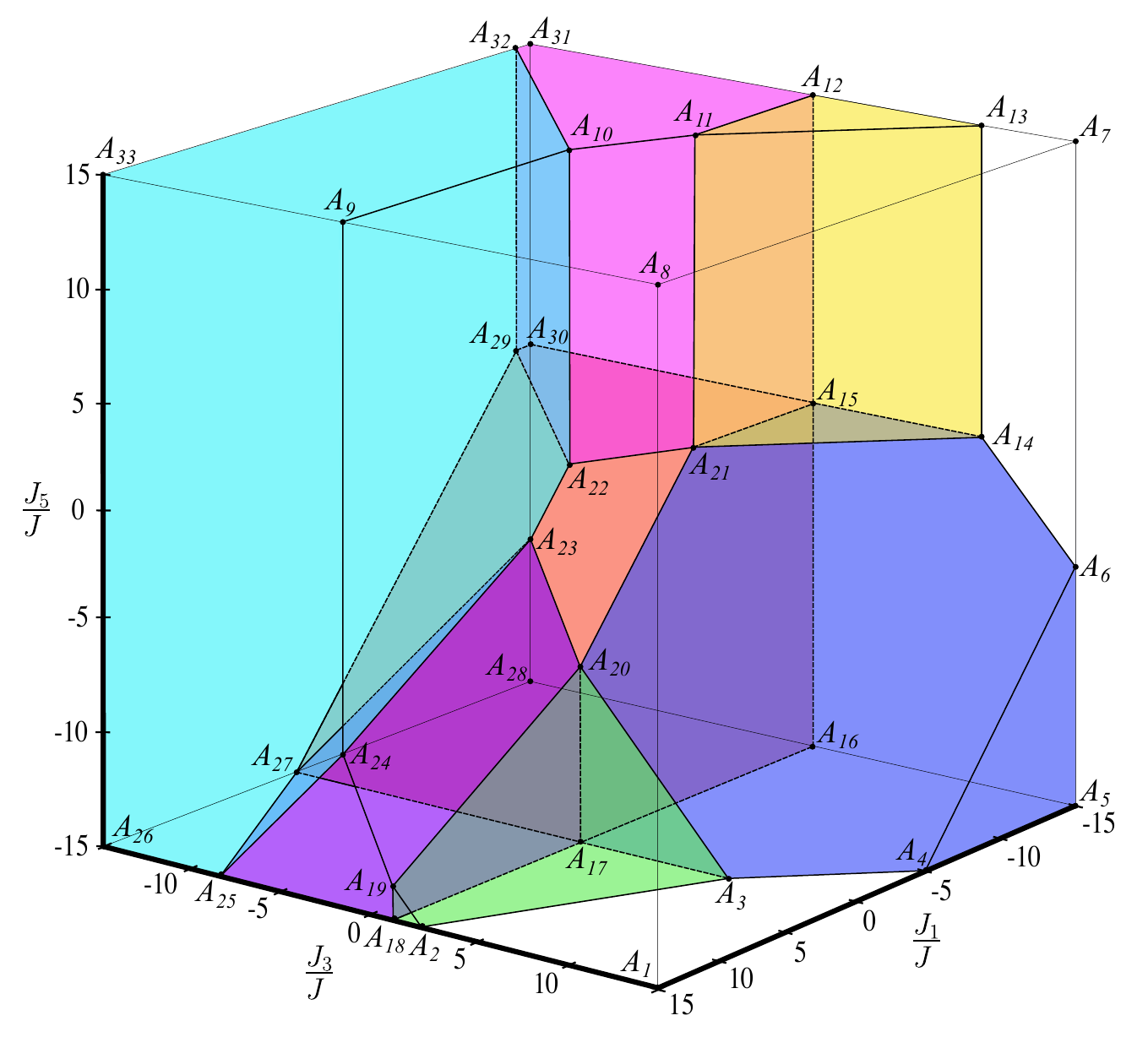}
		\caption{Ground-state regions of the PSCSH model in the $(J_1/J, J_3/J, J_5/J)$ space for $J_2=J_4=J_6=J_7=J=1$.}
		\label{figure2}
	\end{figure}

	\begin{figure*}[!htbp]
		\begin{minipage}{\linewidth}
			\begin{minipage}{0.24\linewidth}
				\includegraphics[width=\linewidth]{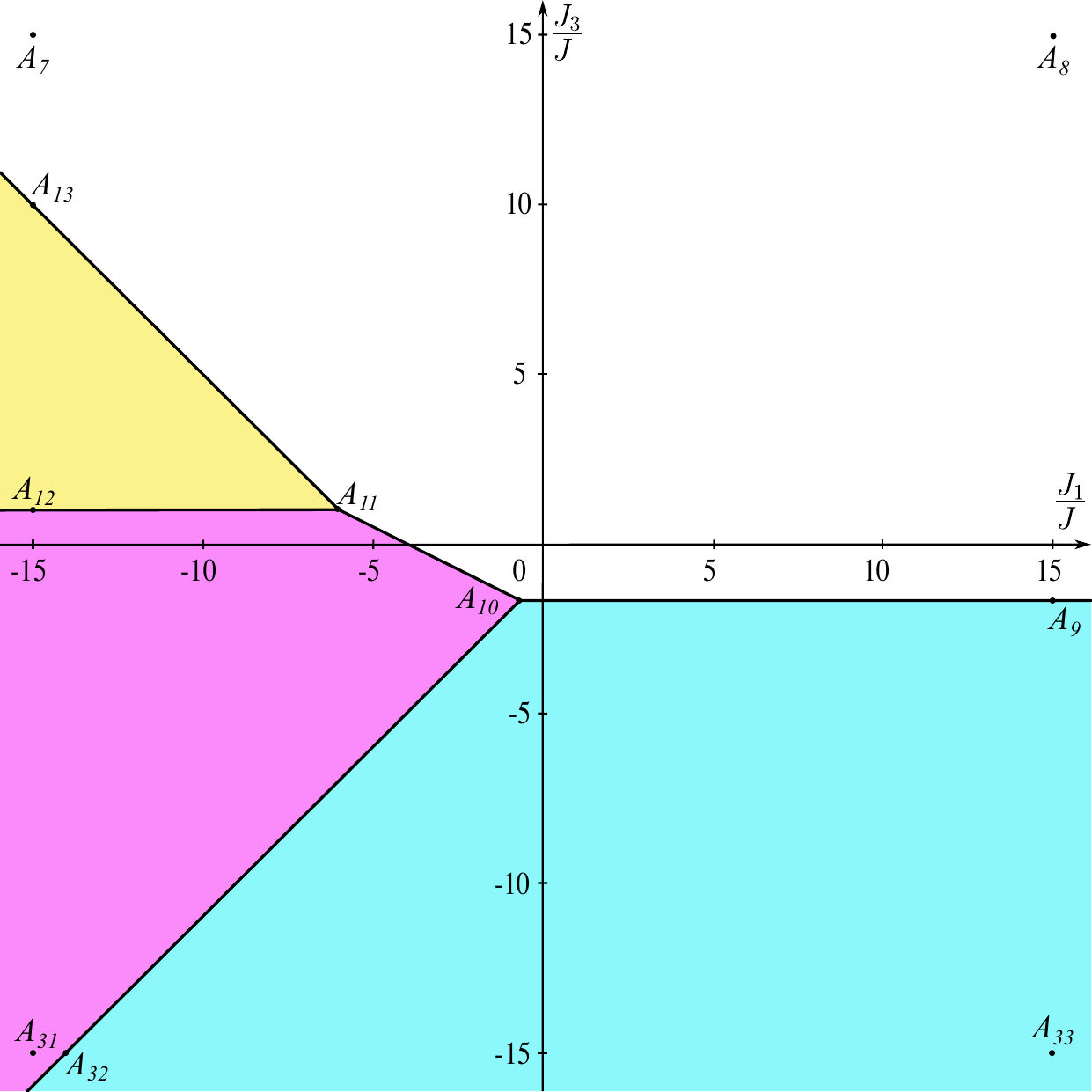} \\ (a)
			\end{minipage}
			\hfill
			\begin{minipage}{0.24\linewidth}
				\includegraphics[width=\linewidth]{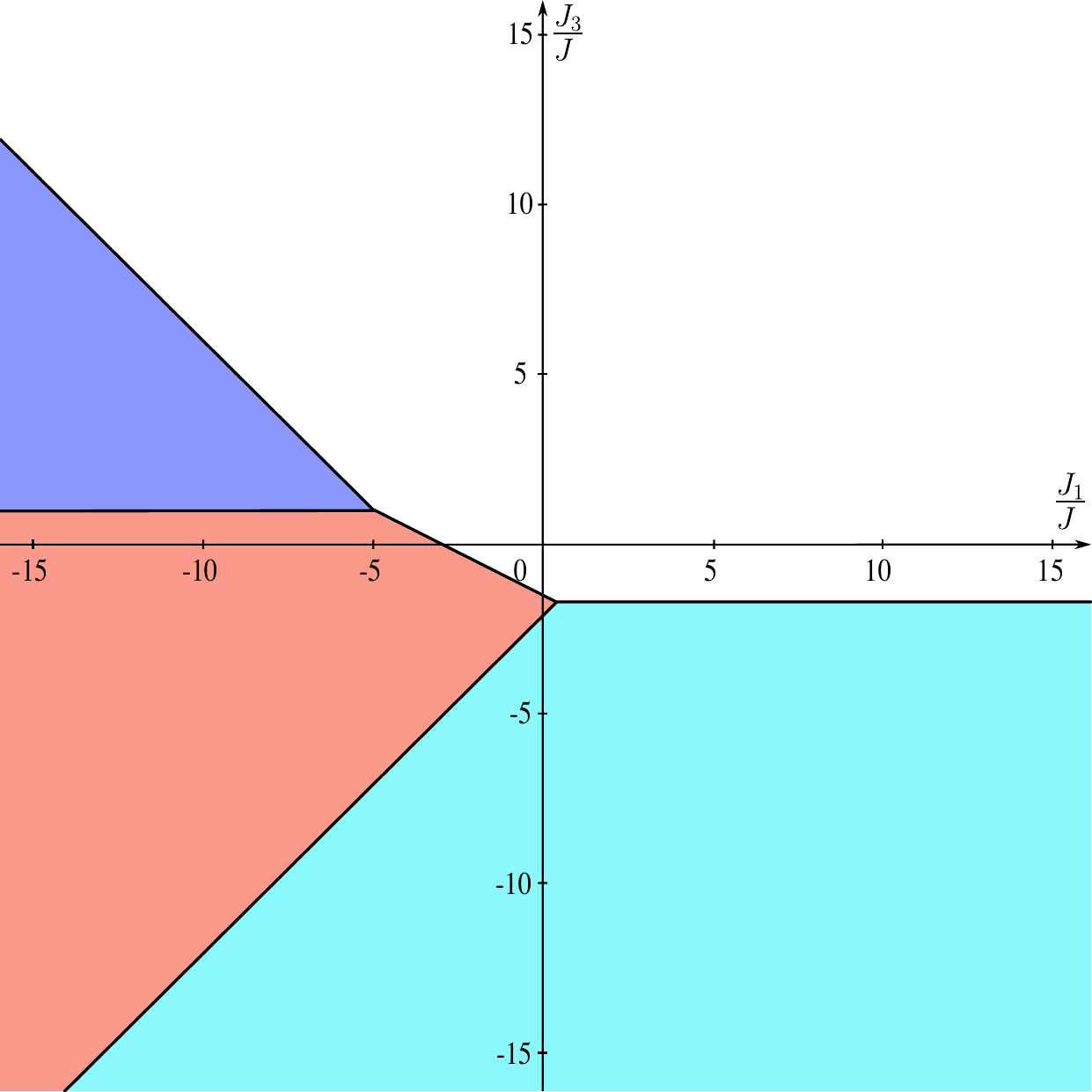} \\ (b)
			\end{minipage}
			\hfill
			\begin{minipage}{0.24\linewidth}
				\includegraphics[width=\linewidth]{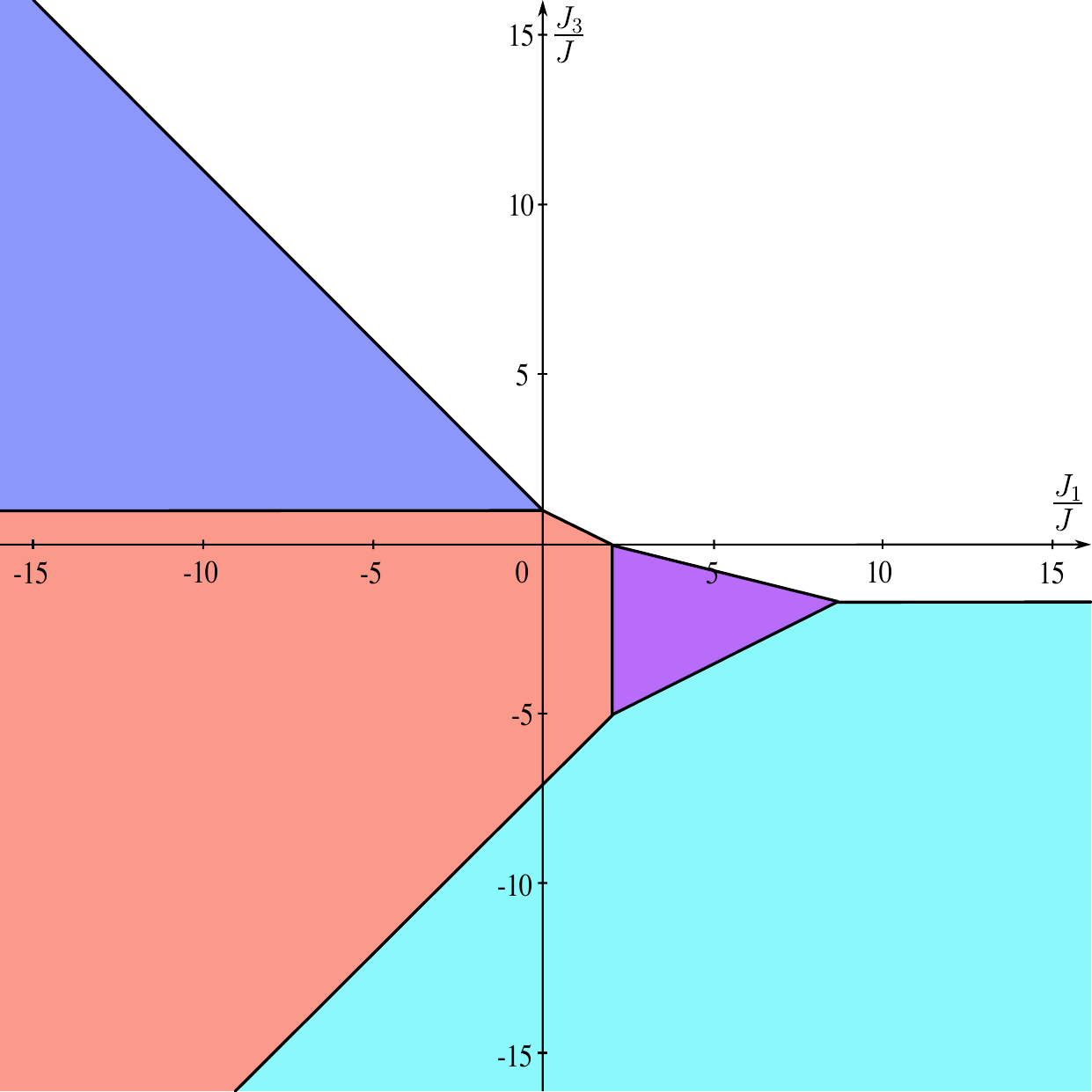} \\ (c)
			\end{minipage}
			\hfill
			\begin{minipage}{0.24\linewidth}
				\includegraphics[width=\linewidth]{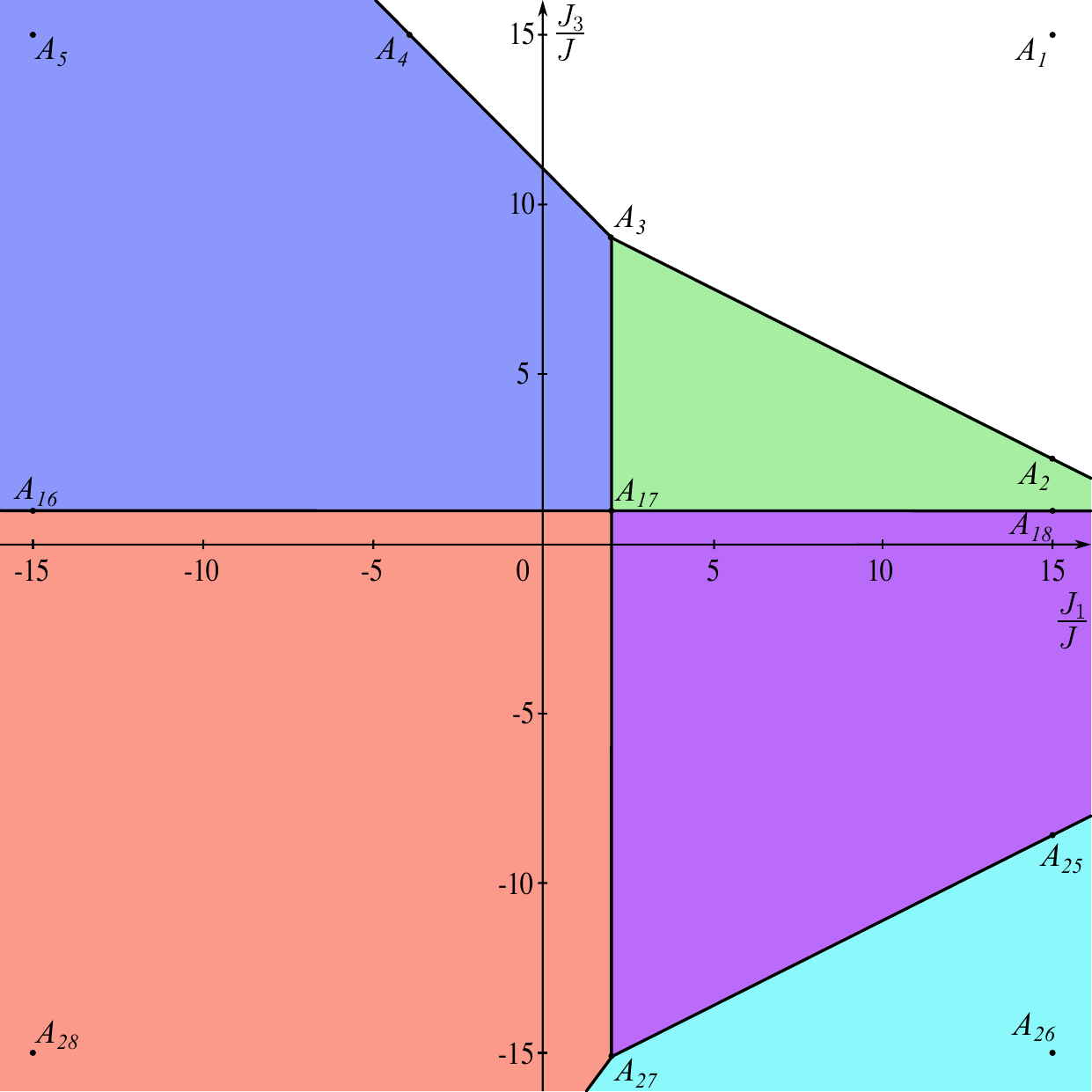} \\ (d)
			\end{minipage}
			\caption{\label{figure3} Ground-state regions of the PSCSH model in the $(J_1/J, J_3/J)$ plane for $J_2=J_4=J_6=J_7=J=1$. Panels (a)--(d) correspond to $J_5/J=15$, $0$, $-5$, and $-15$, respectively.}
		\end{minipage}
	\end{figure*}

	Similarly, fixing $J_3 = J_4 = J_6 = J_7 = J = 1$ and varying $J_1/J$, $J_2/J$, and $J_5/J$ from $-15$ to $15$, we obtain the ground-state diagram shown in Fig.~\ref{figure4}, along with cross sections at $J_5/J = 15,~0,~-5,~-15$ (Fig.~\ref{figure5}).

	\begin{figure*}[!htbp]
		\centering
		\includegraphics[width=0.5\linewidth]{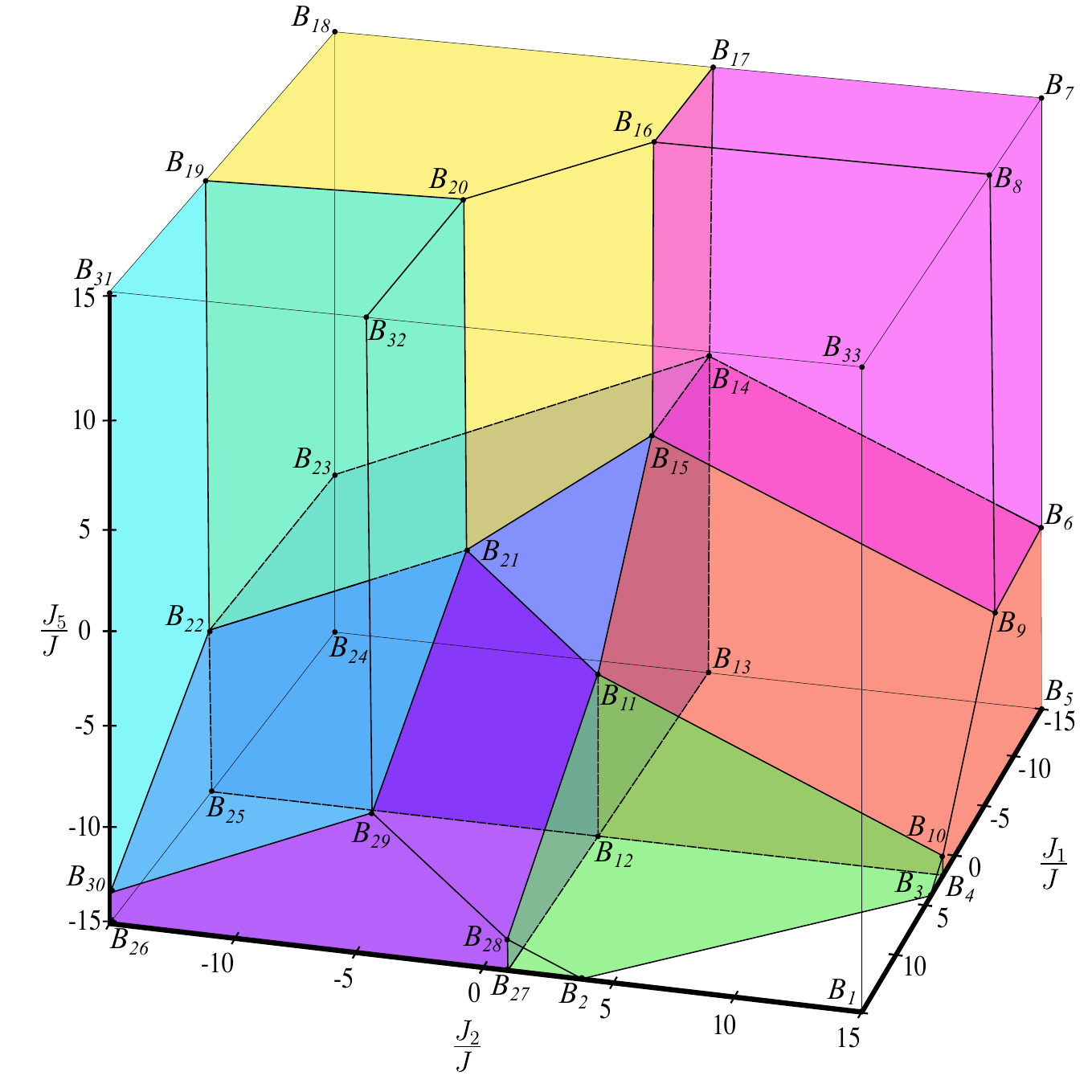}
		\caption{Ground-state regions of the PSCSH model in the $(J_1/J, J_2/J, J_5/J)$ space for $J_3=J_4=J_6=J_7=J=1$.}
		\label{figure4}
	\end{figure*}

	\begin{figure*}[!htbp]
		\begin{minipage}{\linewidth}
			\begin{minipage}{0.24\linewidth}
				\includegraphics[width=\linewidth]{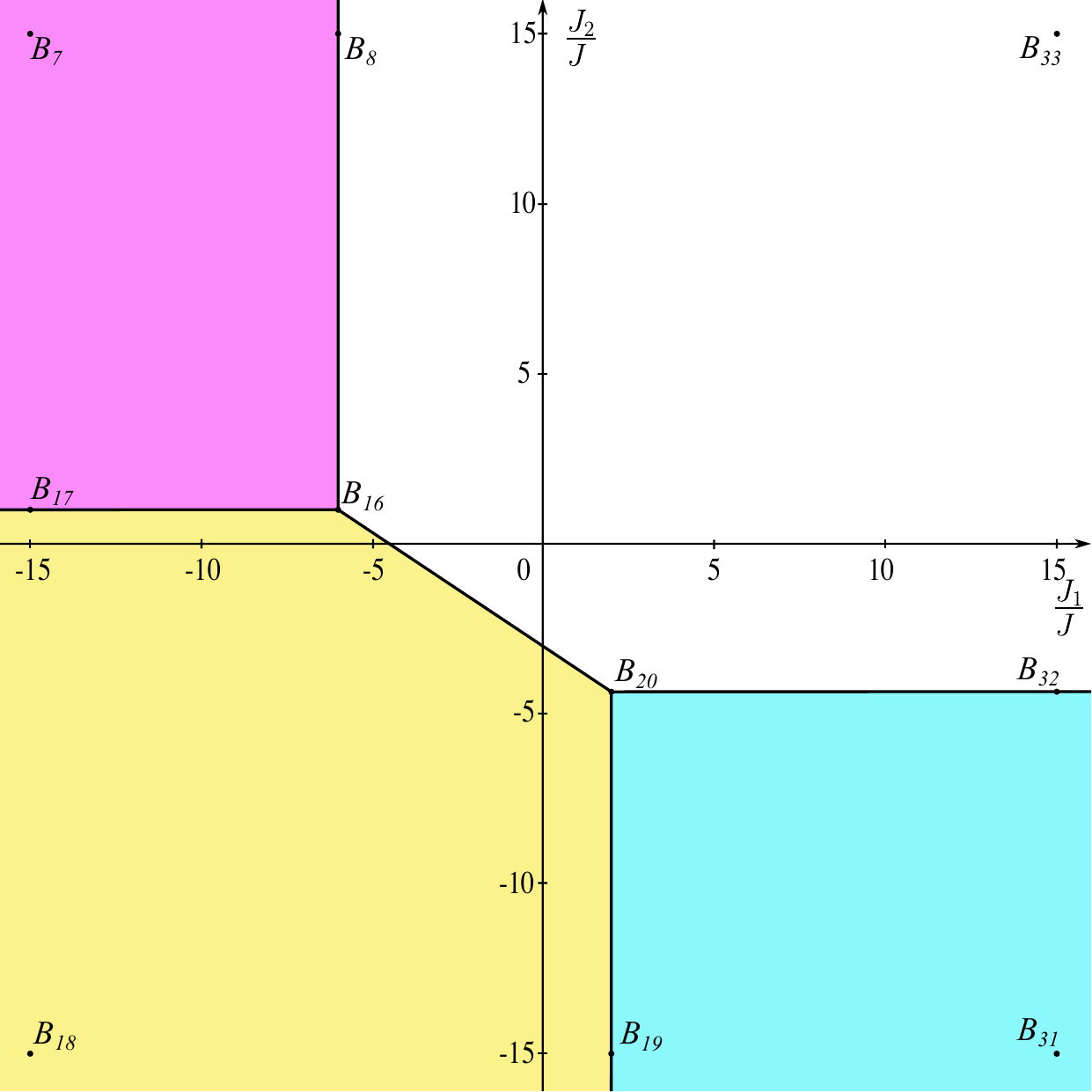} \\ (a)
			\end{minipage}
			\hfill
			\begin{minipage}{0.24\linewidth}
				\includegraphics[width=\linewidth]{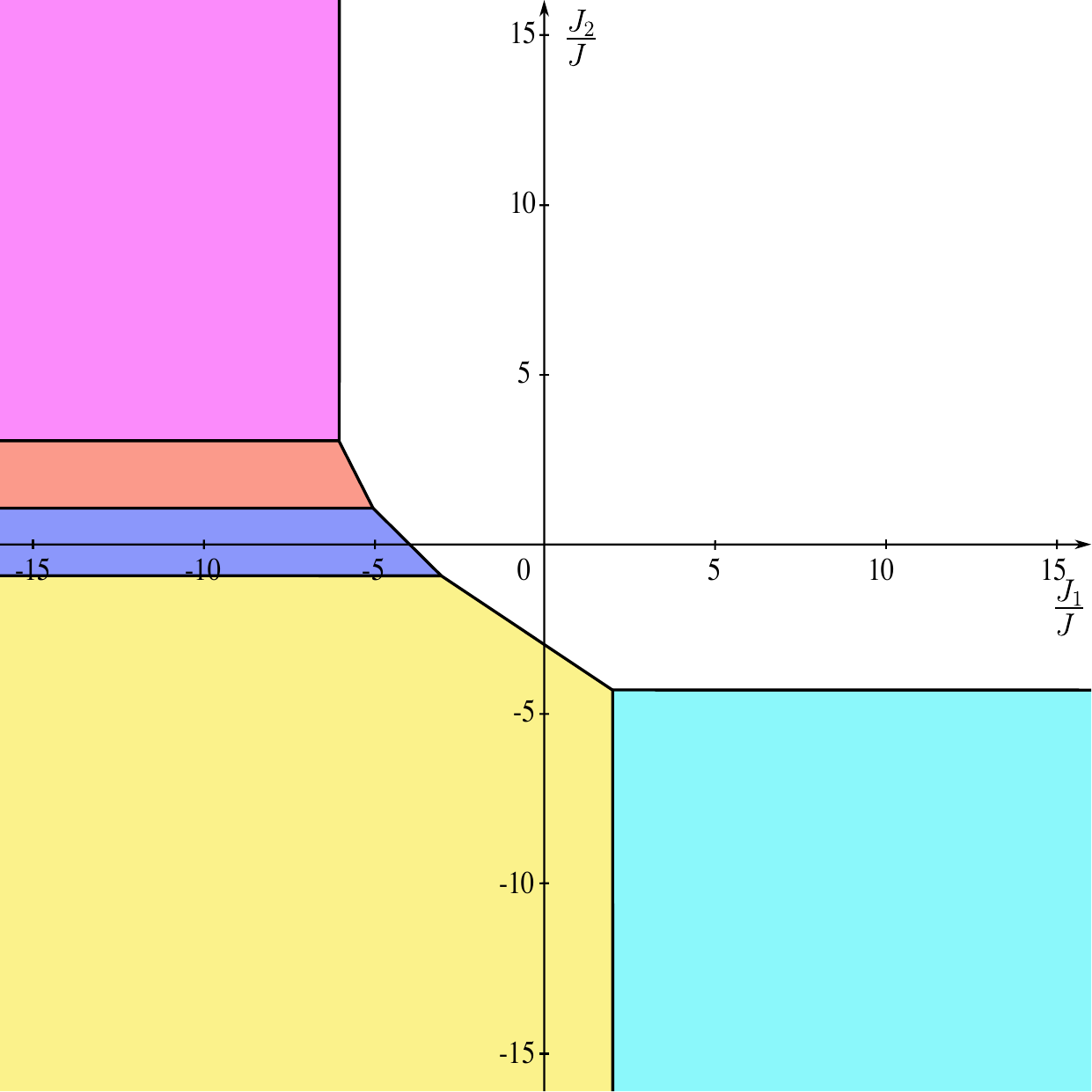} \\ (b)
			\end{minipage}
			\hfill
			\begin{minipage}{0.24\linewidth}
				\includegraphics[width=\linewidth]{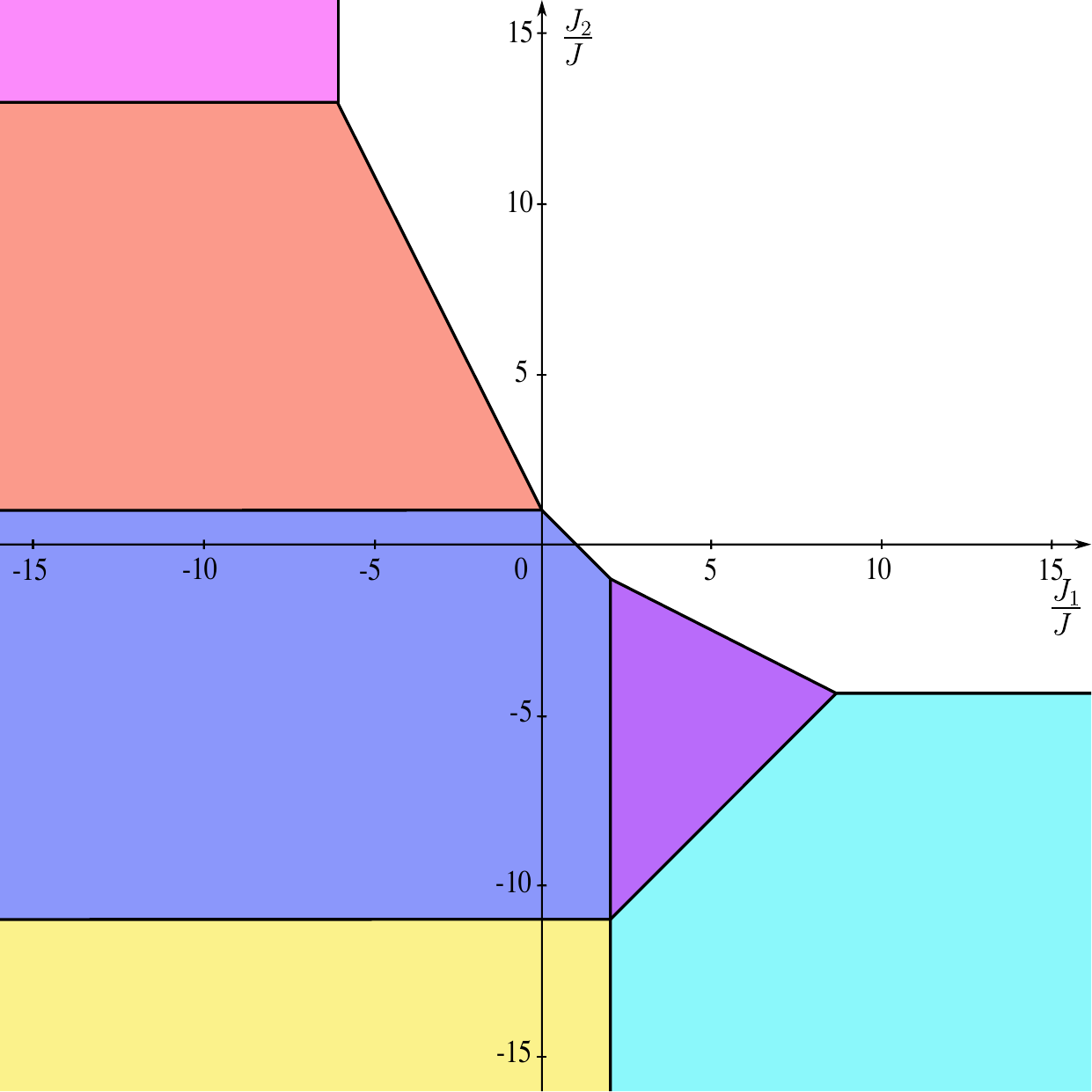} \\ (c)
			\end{minipage}
			\hfill
			\begin{minipage}{0.24\linewidth}
				\includegraphics[width=\linewidth]{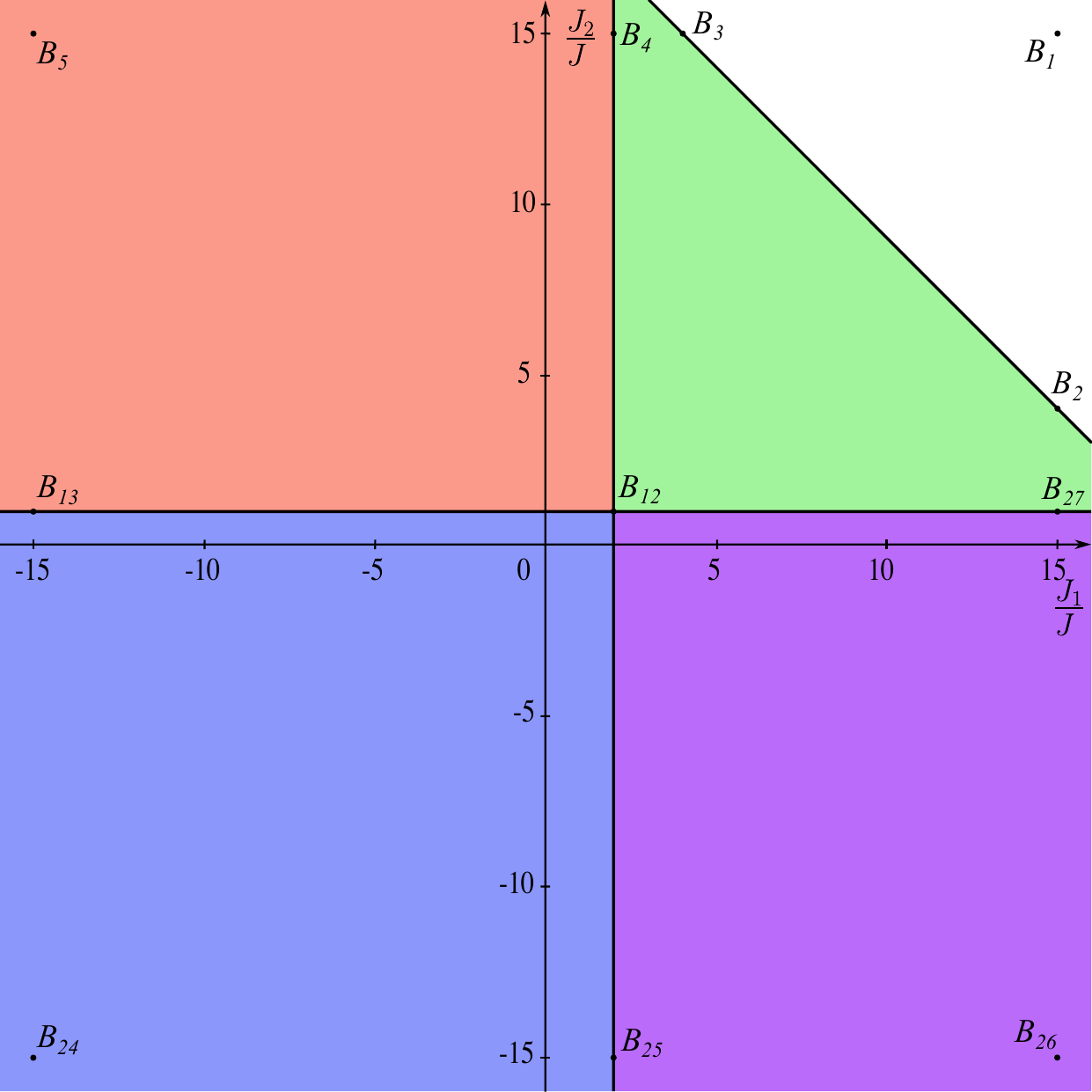} \\ (d)
			\end{minipage}
			\caption{\label{figure5} Ground-state regions of the PSCSH model in the $(J_1/J, J_2/J)$ plane for $J_3=J_4=J_6=J_7=J=1$. Panels (a)--(d) correspond to $J_5/J=15$, $0$, $-5$, and $-15$, respectively.}
		\end{minipage}
	\end{figure*}

	Figure~\ref{figure6} shows the color coding and one representative configuration for each ground state, and the vertex coordinates for Figs.~\ref{figure2} and~\ref{figure4} are listed in Table~\ref{table1}.

	\begin{figure*}[!htbp]
		\centering
		\includegraphics[width=2.4cm]{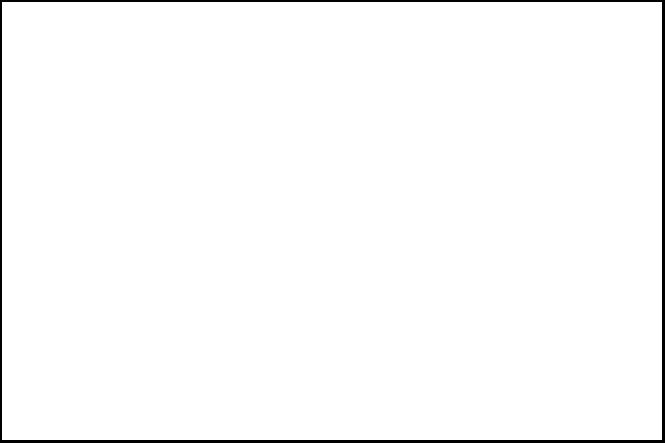}
		\quad
		\includegraphics[width=5.2cm]{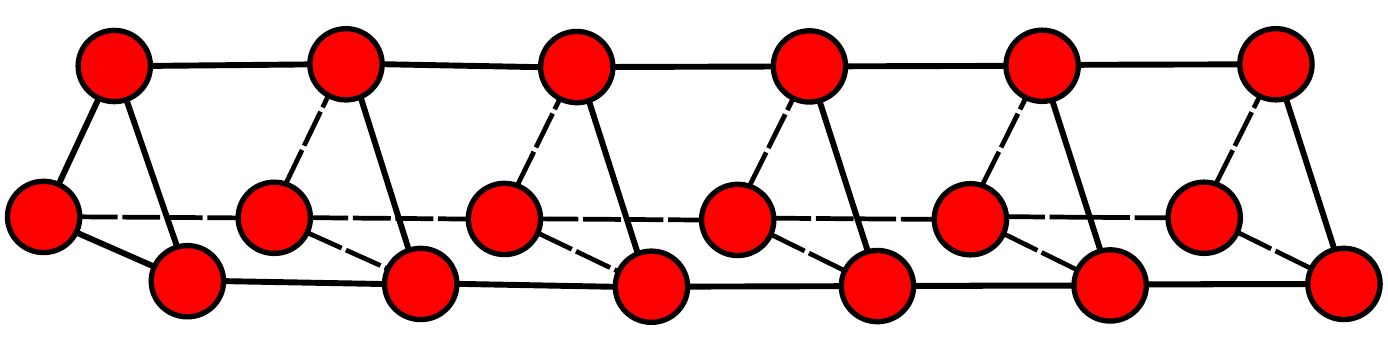}
		\quad
		\includegraphics[width=2.4cm]{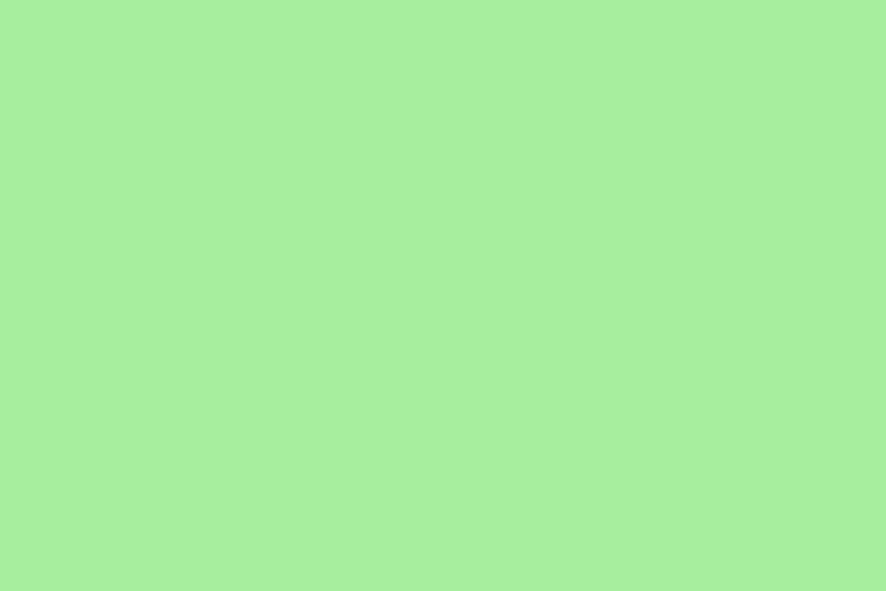}
		\quad
		\includegraphics[width=5.2cm]{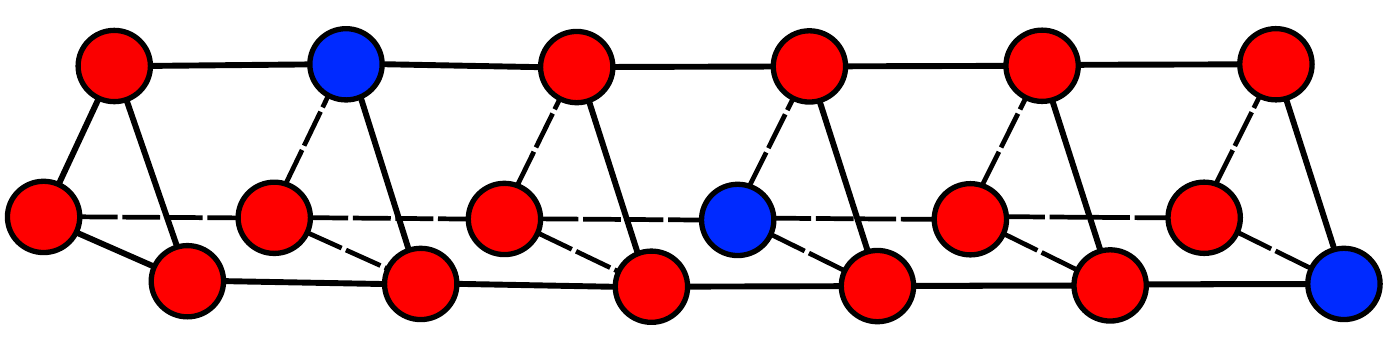}
		\includegraphics[width=2.4cm]{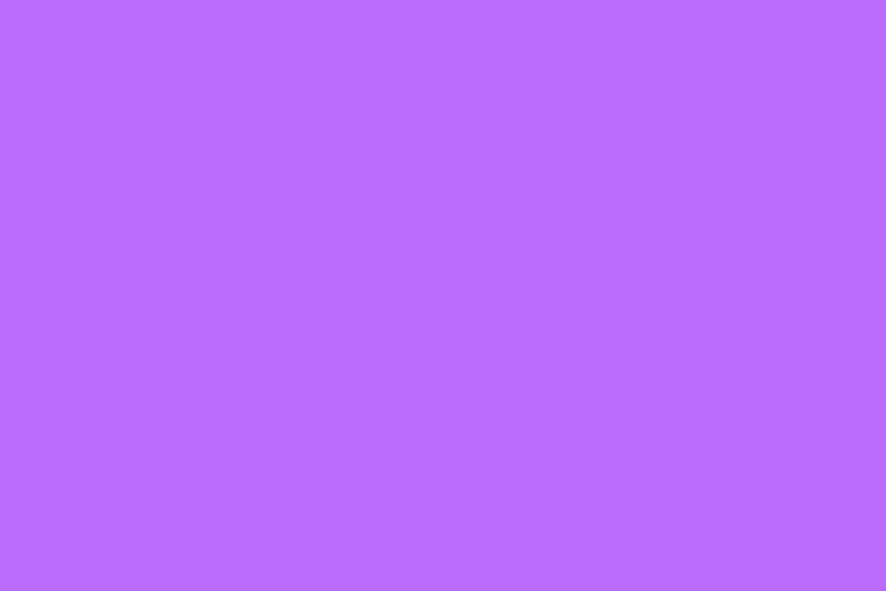}
		\quad
		\includegraphics[width=5.2cm]{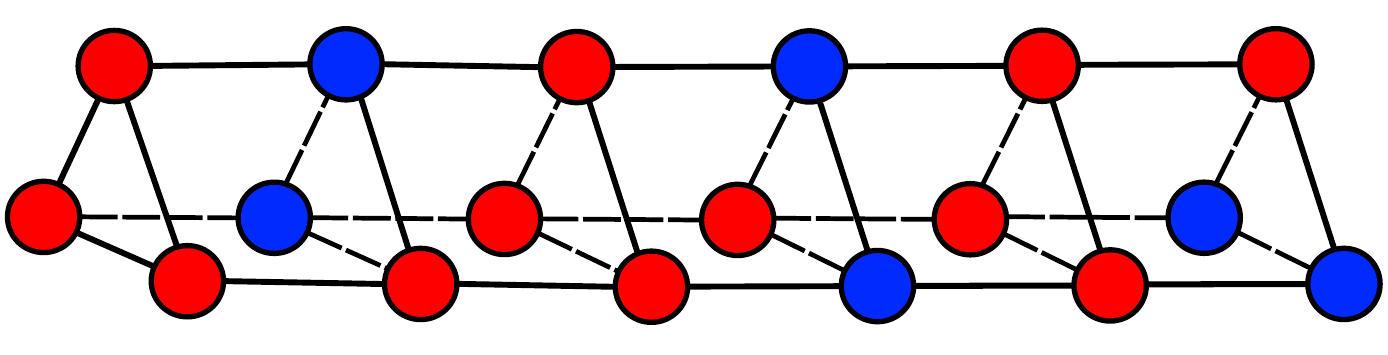}
		\quad
		\includegraphics[width=2.4cm]{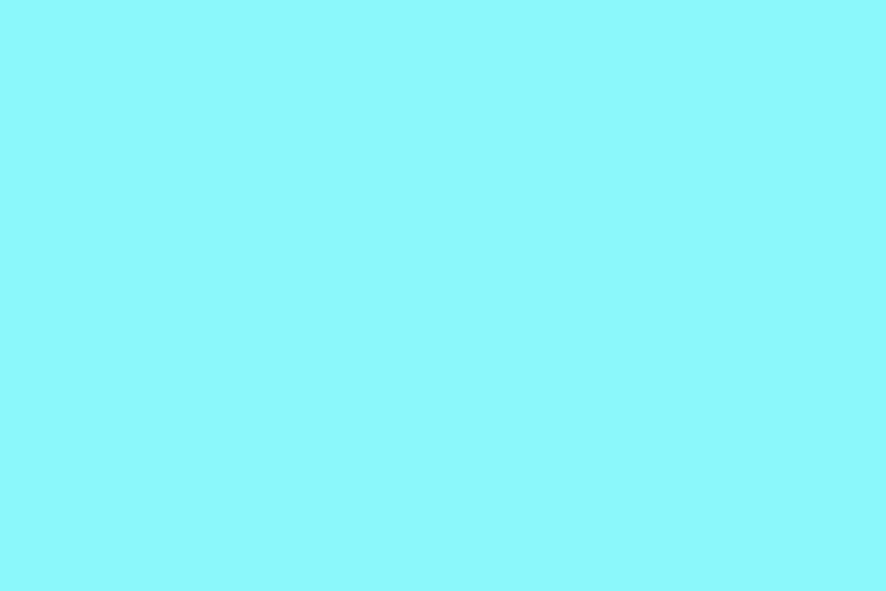}
		\quad
		\includegraphics[width=5.2cm]{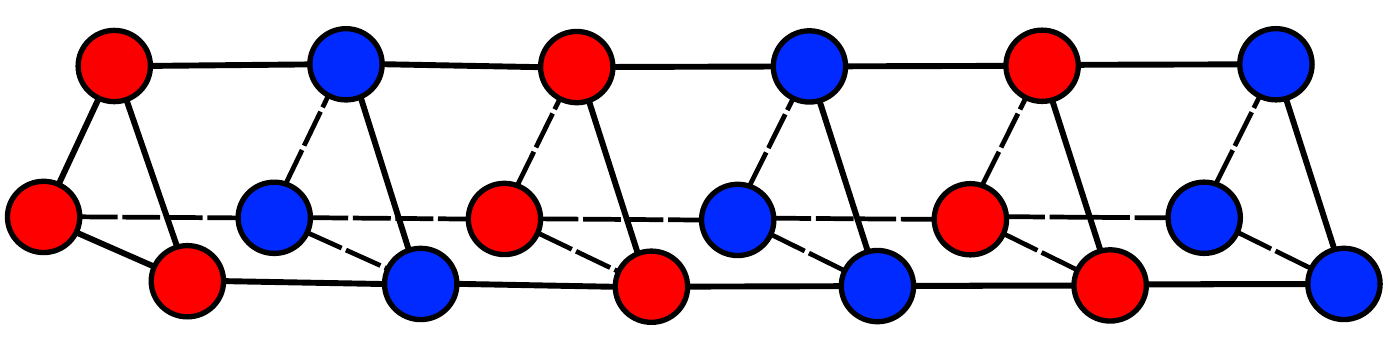}
		\includegraphics[width=2.4cm]{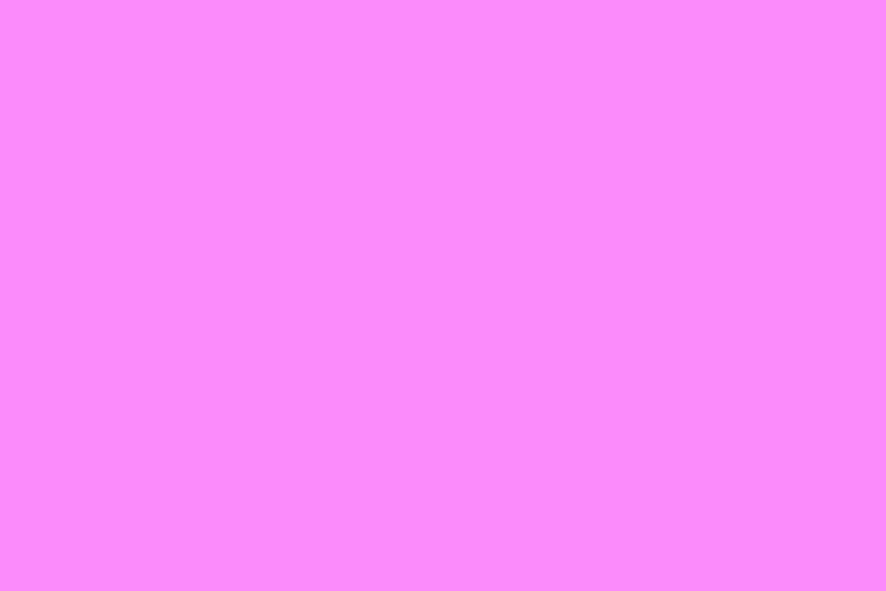}
		\quad
		\includegraphics[width=5.2cm]{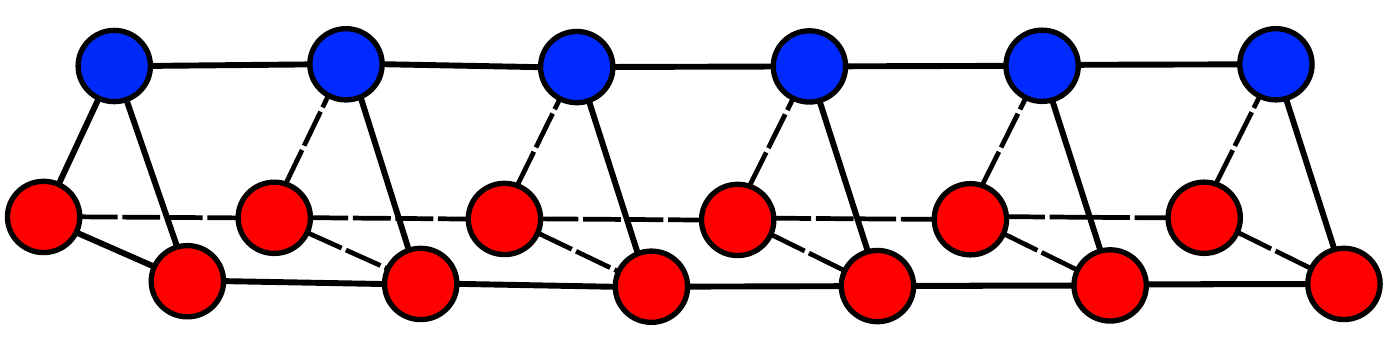}
		\quad
		\includegraphics[width=2.4cm]{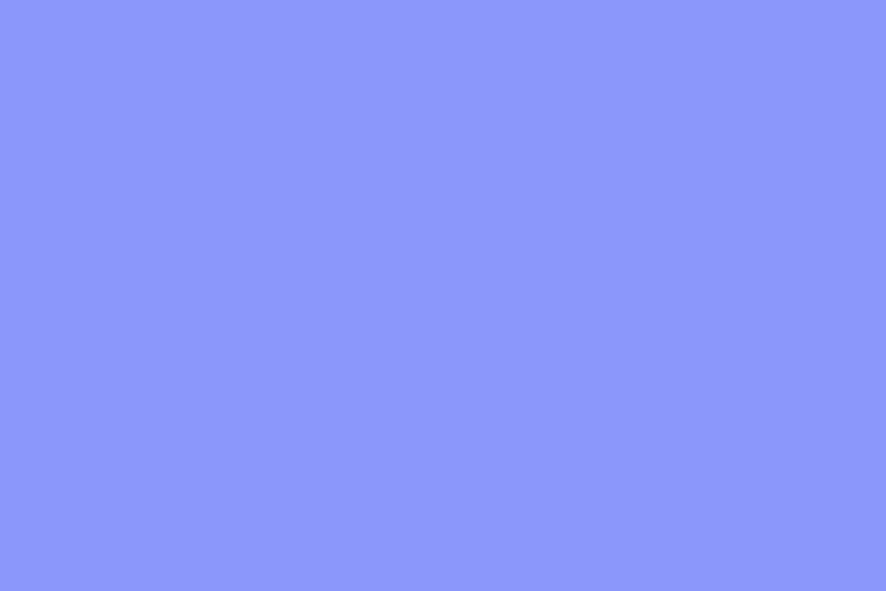}
		\quad
		\includegraphics[width=5.2cm]{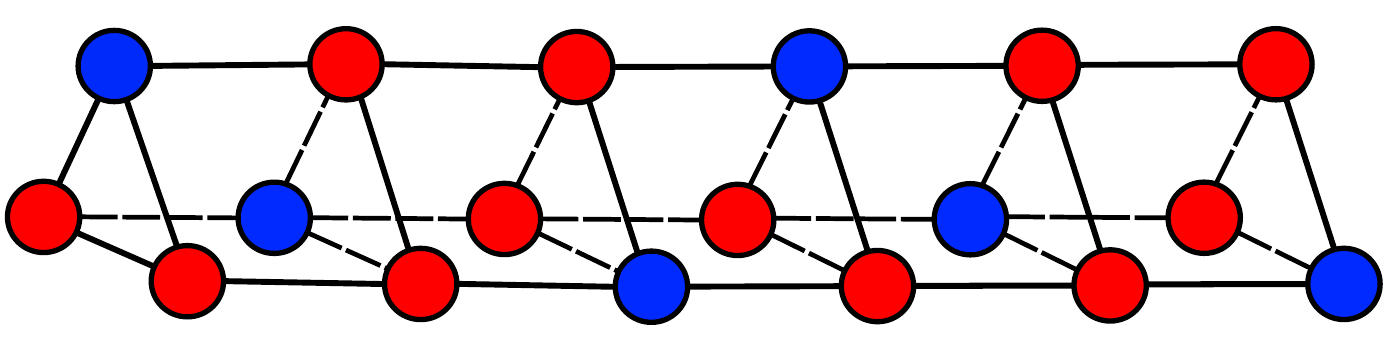}
		\includegraphics[width=2.4cm]{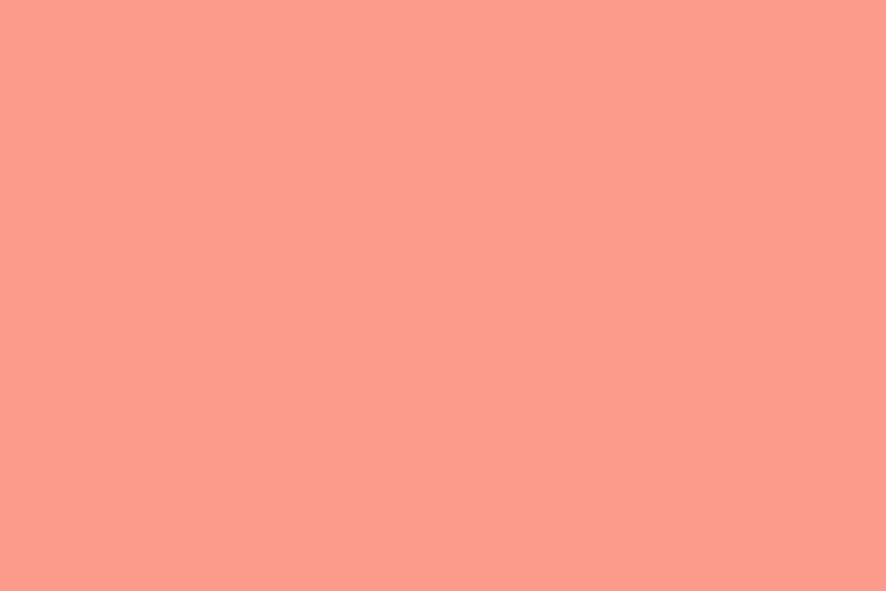}
		\quad
		\includegraphics[width=5.2cm]{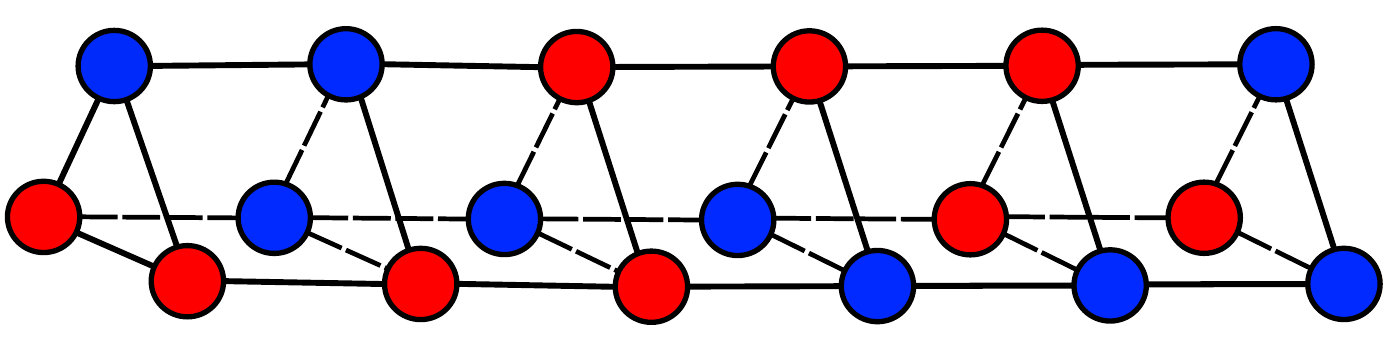}
		\quad
		\includegraphics[width=2.4cm]{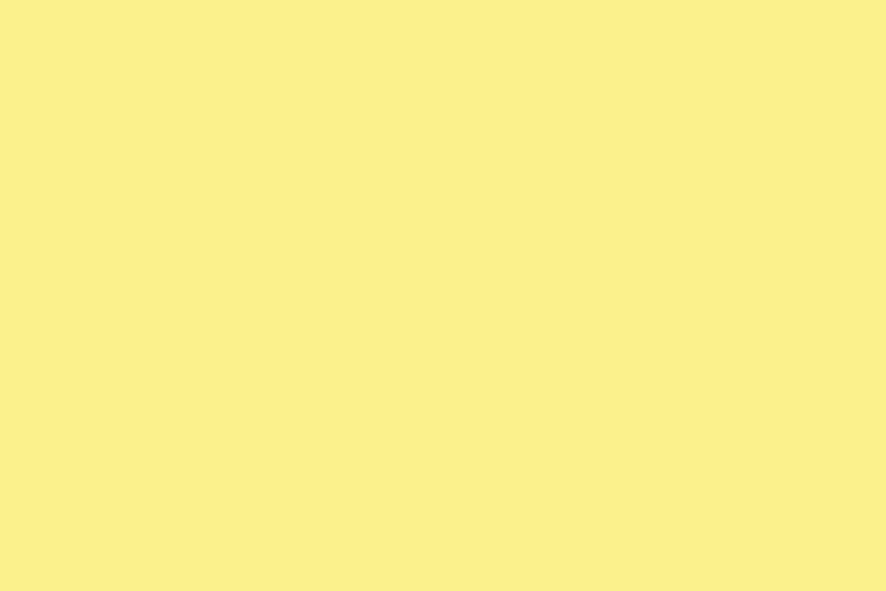}
		\quad
		\includegraphics[width=5.1cm]{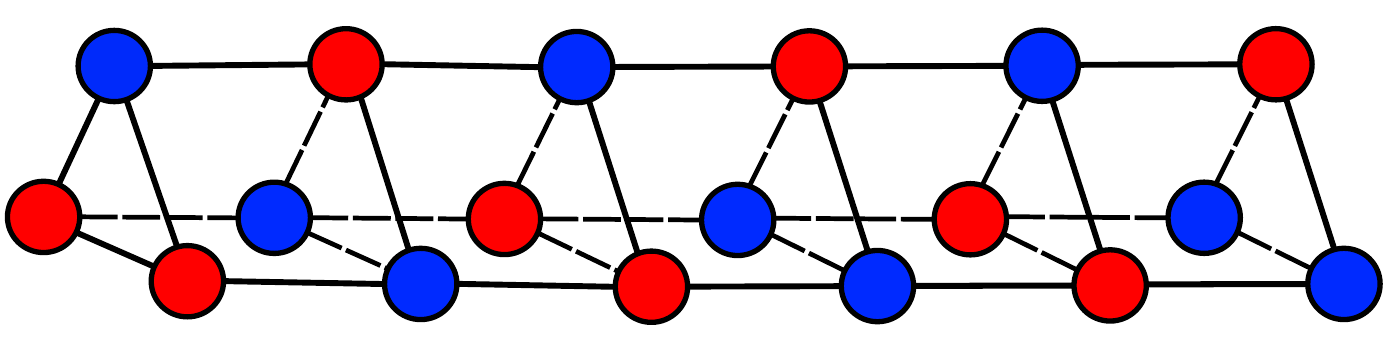}
		\caption{\label{figure6} Color codes and one representative spin configuration for each ground state shown in Figs.~\ref{figure2}--\ref{figure5}.}
	\end{figure*}

	\begin{table}[!htbp]\centering
		\caption{\label{table1}  Coordinates of the boundary points of the ground-state regions shown in Figs.~\ref{figure2} and~\ref{figure4}.}
		\begin{ruledtabular}
			\begin{tabular}{cccc}
				$A$ & Coordinates ($\frac{J_1}{J},~\frac{J_3}{J},~\frac{J_5}{J}$) & $B$ & Coordinates $\frac{J_1}{J},~\frac{J_2}{J},~\frac{J_5}{J}$\\
				\colrule
				$A_1$& (15,15,$-15$) & $B_1$ & (15,15,$-15$) \\
				$A_2$& (15,$\frac{5}{2}$,$-15$) & $B_2$ & (15,4,$-15$) \\
				$A_3$& (2,9,$-15$) & $B_3$ & (4,15,$-15$) \\
				$A_4$& ($-4,15,-15$) & $B_4$ & (2,15,$-15$) \\
				$A_5$& ($-15,15,-15$) & $B_5$ & ($-15,15,-15$) \\
				$A_6$& ($-15,15,-4$) & $B_6$ & ($-15$,15,$-6$) \\
				$A_7$& ($-15$,15,15) & $B_7$ & ($-15$,15,15) \\
				$A_8$& (15,15,15) & $B_8$ & ($-6$,15,15) \\
				$A_9$& (15,$-\frac{5}{3}$,15) & $B_9$ & ($-6,15,-6$) \\
				$A_{10}$& ($-\frac{2}{3}$,$-\frac{5}{3}$,15) & $B_{10}$ & (2,15,$-14$) \\
				$A_{11}$& ($-6$,1,15) & $B_{11}$ & (2,1,$-7$) \\
				$A_{12}$& ($-15$,1,15) & $B_{12}$ & (2,1,$-15$) \\
				$A_{13}$& ($-15$,10,15) & $B_{13}$ & ($-15,1,-15$) \\
				$A_{14}$& ($-15$,10,1) & $B_{14}$ & ($-15$,1,1) \\
				$A_{15}$& ($-15$,1,1) & $B_{15}$ & ($-6$,1,1) \\
				$A_{16}$& ($-15$,1,$-15$) & $B_{16}$ & ($-6$,1,15) \\
				$A_{17}$& (2,1,$-15$) & $B_{17}$ & ($-15$,1,15) \\
				$A_{18}$& (15,1,$-15$) & $B_{18}$ & ($-15,-15$,15) \\
				$A_{19}$& (15,1,$-\frac{27}{2}$) & $B_{19}$ & (2,$-15$,15) \\
				$A_{20}$& (2,1,$-7$) & $B_{20}$ & (2,$-\frac{13}{3}$,15) \\
				$A_{21}$& ($-6$,1,1) & $B_{21}$ & (2,$-\frac{13}{3}$,$-\frac{5}{3}$) \\
				$A_{22}$& ($-\frac{2}{3}$,$-\frac{5}{3}$,1) & $B_{22}$ & (2,$-15$,$-7$) \\
				$A_{23}$& (2,$-\frac{5}{3}$,$-\frac{5}{3}$) & $B_{23}$ & ($-15,-15$,$-7$) \\
				$A_{24}$& (15,$-\frac{5}{3}$,$-\frac{49}{6}$) & $B_{24}$ & ($-15,-15,-15$) \\
				$A_{25}$& (15,$-\frac{17}{2},-15$) & $B_{25}$ & (2,$-15,-15$) \\
				$A_{26}$& (15,$-15,-15$) & $B_{26}$ & (15,$-15,-15$) \\
				$A_{27}$& (2,$-15,-15$) & $B_{27}$ & (15,1,$-15$) \\
				$A_{28}$& ($-15,-15,-15$) & $B_{28}$ & (15,1,$-\frac{27}{2}$) \\
				$A_{29}$& ($-14,-15$,1) & $B_{29}$ & (15,$-\frac{13}{3}$,$-\frac{49}{6}$) \\
				$A_{30}$& ($-15,-15$,1) & $B_{30}$ & (15,$-15$,$-\frac{27}{2}$) \\
				$A_{31}$& ($-15,-15$,15) & $B_{31}$ & (15,$-15$,15) \\
				$A_{32}$& ($-14,-15$,15) & $B_{32}$ & (15,$-\frac{13}{3}$,15) \\
				$A_{33}$& ($15,-15$,15) & $B_{33}$ & (15,15,15)
			\end{tabular}
		\end{ruledtabular}
	\end{table}

	Note that all model configurations belonging to the same class \cite{Kashapov} as the configurations shown in Fig.~\ref{figure6} also correspond to the ground states marked with the corresponding color.

	Moreover, in the thermodynamic limit, the ground states associated with the transfer-matrix elements $b$, $c$, $f$, and $g$ are infinitely degenerate, as follows from the structure of the transfer matrix.

	The boundary points of the ground-state regions are of particular interest, because the behavior of the correlation length may change both near these points and exactly at them, depending on the temperature. The inverse correlation length can be written as \cite{Yokota}
	\begin{equation} \label{76}
		\xi^{-1} = \ln\biggl(\min_{\lambda_q \neq \lambda_{\max}} \bigg|\frac{\lambda_{\max}}{\lambda_q}\bigg|\biggr).
	\end{equation}

	We now consider several representative slices of the parameter space. As an example, plots of the inverse correlation length in the low-temperature region near selected boundary points, as well as directly at the boundary points themselves, are shown in Fig.~\ref{figure7} for the regions presented in Figs.~\ref{figure2} and~\ref{figure4}. The figures demonstrate that at the boundaries of the ground-state regions, the inverse correlation length does not approach zero in the limit $T/J \to 0$.

	\begin{figure*}[!htbp]
		\begin{minipage}{\linewidth}
			\begin{minipage}{0.32\linewidth}
				\includegraphics[width=6.2cm]{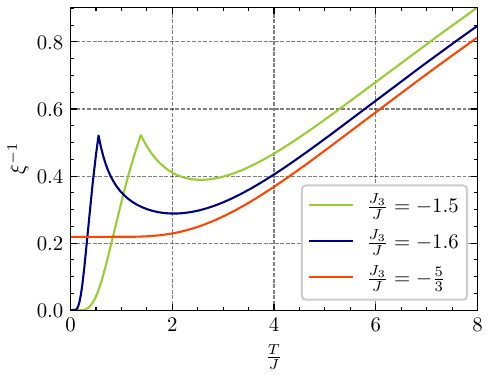} \\ (a)
			\end{minipage}
			\hfill
			\begin{minipage}{0.32\linewidth}
				\includegraphics[width=6.2cm]{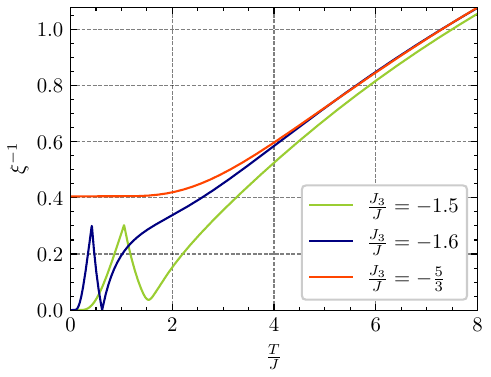} \\ (b)
			\end{minipage}
			\hfill
			\begin{minipage}{0.32\linewidth}
				\includegraphics[width=6.2cm]{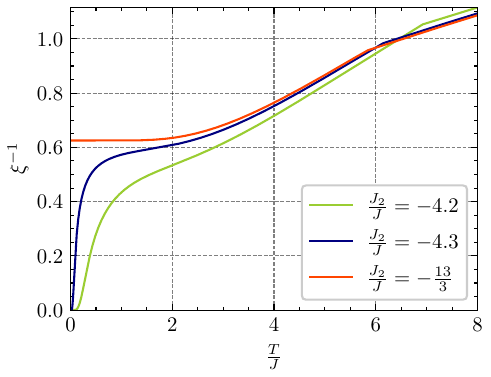} \\ (c)
			\end{minipage}
			\caption{\label{figure7} Inverse correlation length in the low-temperature regime for the following parameter choices: \\
				(a): $\frac{J_1}{J} = 2, \frac{J_5}{J} = -\frac{5}{3}$, $J_2=J_4=J_6=J_7=J$ and $\frac{J_3}{J} = -1.5, -1.6, -\frac{5}{3}$, \\
				(b): $\frac{J_1}{J} = -\frac{2}{3}, \frac{J_5}{J} = 1$,  $J_2=J_4=J_6=J_7=J$ and $\frac{J_3}{J} = -1.5, -1.6, -\frac{5}{3}$, \\
				(c): $\frac{J_1}{J} = 2, \frac{J_5}{J} = -\frac{5}{3}$, $J_3=J_4=J_6=J_7=J$ and $\frac{J_2}{J} = -4.2, -4.3, -\frac{13}{3}$.}
		\end{minipage}
	\end{figure*}
	
	\newpage
	\bibliography{bibliography}
\end{document}